%% file: diss.tex
\documentclass[final,11pt,a4paper]{report}
\usepackage{a4}       
\usepackage{latexsym} 
\usepackage{amssymb}  
\usepackage{amsmath}  

\usepackage{float}    
\usepackage{framed}   
\usepackage{graphics} 
\usepackage{graphicx} 
\usepackage{calc}     
\usepackage{qip}      

\usepackage{bbm}      
\usepackage{tabularx} 
\usepackage{xspace}   

\usepackage[final,pagebackref,colorlinks=false, pdftitle={Cryptography
  in the Bounded-Quantum-Storage Model}, pdfauthor={Christian
  Schaffner},pdfsubject={PhD Dissertation},pdfkeywords={dissertation,
  quantum cryptography, bounded-quantum-storage model, oblivious
  transfer, bit commitment, quantum key
  distribution}]{hyperref} 

\usepackage{myjeffe}

\usepackage{makeidx}  
\makeindex

\renewcommand*{\backref}[1]{
}
\renewcommand*{\backrefalt}[4]{%
  \par
  \ifnum#1=1 %
    cited on page %
  \else
    cited on pages %
  \fi
  #2.\par
}

{\begin{list}{}{%
    \settowidth{\labelwidth}{#1\hspace{.1em}}%
    \setlength{\leftmargin}{\labelwidth+\labelsep}}}%
{\end{list}}

\newtheorem{theorem}{Theorem}[chapter]
\newtheorem{definition}[theorem]{Definition}
\newtheorem{proposition}[theorem]{Proposition}
\newtheorem{lemma}[theorem]{Lemma}
\newtheorem{corollary}[theorem]{Corollary}

\newcounter{itm}
\newenvironment{myprotocol}[2][]
  {\begin{minipage}{\columnwidth} 
    \begin{framed}\hspace{0ex} 
     \begin{minipage}{0.99\columnwidth}
       {{\bf #2:} #1}
       \setcounter{itm}{1}
       \begin{list}{\arabic{itm}.}{\usecounter{itm}}}
   {    \end{list}
       \vspace{-1.5ex} 
       \end{minipage} 
     \end{framed} 
    \end{minipage}\vspace{-0.6ex}}

\newenvironment{myfigure}[1]    
         {\begin{figure}[#1] \centering}
         { \end{figure}}

\newcommand*{\comment}[1]{\textsf{[#1]}}
\newcommand*{\remove}[1]{}

\newcommand*{\clearemptydoublepage}{\newpage{\pagestyle{empty}\cleardoublepage}}

\renewcommand*{\S}{\ensuremath{{\sf S}}} 
\renewcommand*{\R}{\ensuremath{{\sf R}}} 
\newcommand*{\dS}{\ensuremath{\tilde{\sf S}}} 
\newcommand*{\dR}{\ensuremath{\tilde{\sf R}}} 

\newcommand*{\A}{\ensuremath{{\sf S}}} 
\newcommand*{\B}{\ensuremath{{\sf R}}} 
\newcommand*{\dA}{\ensuremath{\tilde{\sf S}}} 
\newcommand*{\dB}{\ensuremath{\tilde{\sf R}}} 
\newcommand*{\V}{\ensuremath{{\sf V}}} 
\renewcommand*{\C}{\ensuremath{{\sf C}}} 
\newcommand*{\dC}{\ensuremath{\tilde{\sf C}}} 

\newcommand*{\qot}{{\sc qot}}
\newcommand*{\epr}{{\sc epr}}
\newcommand*{\eprqot}{\epr-\qot}
\newcommand*{\comm}{{\sc comm}}
\newcommand*{\eprcomm}{{\sc epr-comm}}
\newcommand*{\BBqot}{{\sc bb84-qot}}
\newcommand*{\BBeprqot}{{\sc bb84-epr-qot}}
\newcommand*{\OTUOT}{{\sf \textsl{OT2UOT}}}      

\newcommand*{\onetwo}[1][2]{\mbox{\textsl{1\hspace{-0.1ex}-#1}}}
\newcommand*{\onen}{\mbox{\textsl{1\hspace{-0.1ex}-\hspace{-0.1ex}n}}}
\newcommand*{\onethree}{\mbox{\textsl{1\hspace{-0.1ex}-\hspace{-0.1ex}3}}}

\newcommand*{\pOT}{\index{oblivious transfer}\textsl{OT}\xspace}                    
\newcommand*{\OT}[1][2]{\index{oblivious transfer!one-out-of-two}\textsl{\onetwo[#1]\:OT}\xspace}               
\newcommand*{\StringOT}{\index{oblivious transfer!one-out-of-two}\textsl{\onetwo\:String\:OT\xspace}}    
\newcommand*{\lStringOT}[1][2]{\textsl{\onetwo[#1]\:OT\,$^{\ell}$}\xspace}    
\newcommand*{\onenOT}{\index{oblivious transfer!one-out-of-$n$}\textsl{\onen\:OT}\xspace}             
\newcommand*{\onenlStringOT}{\textsl{\onen\:OT\,$^{\ell}$}\xspace}    
\newcommand*{\onethreeOT}{\textsl{\onethree\:OT}\xspace}               

\newcommand*{\Rand}{\textsl{Rand}}                 
\newcommand*{\RandOT}{\Rand\:\OT}                  
\newcommand*{\RandStringOT}{\Rand\:\StringOT}      
\newcommand*{\RandlStringOT}[1][2]{\Rand\:\lStringOT[#1]}    
\newcommand*{\onenRandOT}{\Rand\:\onenOT}                  
\newcommand*{\onenRandlStringOT}{\Rand\:\onenlStringOT}    

\newcommand*{\Randqot}[1][2]{{\small \sc Rand 1\hspace{-0.1ex}-\hspace{-0.1ex}#1 QOT}\,}
\newcommand*{\Randlqot}[1][2]{{\small \sc Rand 1\hspace{-0.1ex}-\hspace{-0.1ex}#1 QOT}\,$^{\ell}$}
\newcommand*{\eprRandlqot}[1][2]{{\small\sc EPR Rand 1\hspace{-0.1ex}-\hspace{-0.1ex}#1 QOT}\,$^{\ell}$}
\newcommand*{\RabinOT}{\index{oblivious transfer!Rabin}\textsl{Rabin\:OT}\xspace}            
\newcommand*{\RabinlStringOT}{\index{oblivious transfer!Rabin}\textsl{Rabin\:OT\,$^{\ell}$}\xspace}      
\newcommand*{\XOT}{\index{oblivious transfer!XOR}\textsl{\onetwo\:XOT}\xspace}             
\newcommand*{\GOT}{\index{oblivious transfer!generalized}\textsl{\onetwo\:GOT}\xspace}             
\newcommand*{\BUOT}{\index{oblivious transfer!universal}\textsl{\onetwo\:UOT}\xspace}            
\newcommand*{\pUOT}{\index{oblivious transfer!universal}\textsl{UOT}\xspace}                     
\newcommand*{\UOT}[3]{\index{oblivious transfer!universal}\textsl{$(#1,#2)$-UOT$(#3)$}\xspace}     

\newcommand*{\BC}{\index{bit commitment}\textsl{BC}\xspace}  
\newcommand*{\QKD}{\index{quantum key distribution}\textsl{QKD}\xspace}

\newcommand*{\assign}{\ensuremath{\kern.5ex\raisebox{.1ex}{\mbox{\rm:}}\kern -.3em =}}
\def\={\hspace{-0.5ex}=\hspace{-0.4ex}}

\renewcommand*{\sp}{\kappa}      

\newcommand*{\ol}[1]{\overline{#1}}

\renewcommand*{\b}[1]{\text{\boldmath${\bf #1}$}}

\renewcommand*{\id}{\mathbbm{1}}   
\newcommand*{\cB}{\mathcal{B}}
\newcommand*{\cE}{\mathcal{E}}    
\newcommand*{\cH}{\mathcal{H}}
\newcommand*{\cU}{\mathcal{U}}
\newcommand*{\cX}{\mathcal{X}}
\newcommand*{\cY}{\mathcal{Y}}
\newcommand*{\cZ}{\mathcal{Z}}
\newcommand*{\bs}{\b{s}}           
\renewcommand*{\zero}{\b{o}}       
\newcommand*{\Sone}[1][s]{{\cal S}_1(\b{#1})}  
\newcommand*{\Stwo}[1][s]{{\cal S}_2(\b{#1})}  

\newcommand*{\negl}[1]{\mathit{negl}({#1})}
\newcommand*{\unif}{\mbox{\sc unif}}
\newcommand*{\nbit}{\set{0,1}^n}
\newcommand*{\bal}{\beta}              
\newcommand*{\ip}[1]{\langle#1\rangle} 
\newcommand*{\set}[1]{\{#1\}}          
\newcommand*{\Set}[2]{\{#1:\,#2\}}     
\newcommand*{\card}[1]{\big|#1\big|}   
\newcommand*{\bset}{\mathcal{S}}       
\newcommand*{\pad}{|^{\circ}}          

\newcommand*{\Qp}{Q^{+}}
\newcommand*{\Qt}{Q^{\times}}
\newcommand*{\qp}{q^{+}}
\newcommand*{\qt}{q^{\times}}

\newcommand*{\ball}[1]{{\mathrm{B}}^{#1}}  

\newcommand*{\Ba}[1]{\ensuremath{\mathcal{B}^{#1}}}  
\renewcommand*{\I}{\mathbbm{1}}

\newcommand*{\dens}[1]{\mathcal{P}(#1)}  

\newcommand*{\phix}{\phi_{\mbox{\scriptsize{\sc x}}}}
\newcommand*{\phidc}{\phi_{\mbox{\scriptsize{\sc dc}}}}
\newcommand*{\etamq}{\eta_{\mbox{\scriptsize{\sc mq}}}}
\newcommand*{\etaab}{\eta_{\mbox{\scriptsize{\sc ab}}}}

\newcommand*{\prei}[2][i-1]{#2^{#1}} 

\renewcommand*{\H}{\operatorname{H}}   
\newcommand*{\hmin}{\ensuremath{\H_{\infty}}}
\newcommand*{\hmax}{\ensuremath{\H_{0}}}
\newcommand*{\hie}[2]{\ensuremath{\hmin^{#1}(#2)}}
\newcommand*{\hiee}[1]{\hie{\varepsilon}{#1}}
\newcommand*{\hmaxe}[2]{\ensuremath{\hmax^{#1}(#2)}}
\newcommand*{\hmaxee}[1]{\hmaxe{\varepsilon}{#1}}
\newcommand*{\qhmin}{\ensuremath{\H_{\rm min}}}  
\newcommand*{\qhmax}{\ensuremath{\H_{\rm max}}}  
\newcommand*{\hminee}{\qhmin^{\varepsilon}} 
\newcommand*{\tH}{\tilde{\H}}         

\newcommand*{\dist}[1]{\delta\big(#1\big)}  
\newcommand*{\eps}{\varepsilon}
\newcommand*{\epsclose}{\approx_{\varepsilon}}

\renewcommand*{\P}{P}            
\newcommand*{\E}{\mathbb{E}}   
\newcommand*{\ev}{\ensuremath{\mathcal{E}}\xspace} 

\newcommand*{\regE}{E}

\newcommand*{\univ}{two-universal\xspace}
\newcommand*{\hf}{f}             
\newcommand*{\Hf}{F}             
\newcommand*{\chf}[1]{\ensuremath{\mathcal{F}_{#1}}} 

\newcommand*{\UH}{{\cal F}}      

\begin{document}

\pagestyle{empty} 
\pagenumbering{roman} 
\setcounter{secnumdepth}{-1}

\include{frontpages}


\clearemptydoublepage
\phantomsection
\addcontentsline{toc}{chapter}{Contents}
\tableofcontents
\clearemptydoublepage
\pagenumbering{arabic}
\setcounter{secnumdepth}{2}

\pagestyle{myheadings}
\renewcommand*{\chaptermark}[1]{\markboth{\textsc{\chaptername\
      \thechapter. #1}}{}} 
\renewcommand*{\sectionmark}[1]{\markright{\textsc{\thesection. #1}}}

\include{intro}
\clearemptydoublepage
\chapter{Preliminaries}   \label{chap:prelim}
In this chapter, we introduce notation and basic concepts used
throughout the rest of the thesis. In addition, most of the following
chapters have an individual preliminary section introducing concepts
that are exclusively used in those specific chapters.

This chapter does \emph{not} give a thorough introduction to
probability theory, information theory and quantum information
processing, but we rather assume the reader familiar with the basic
concepts from the standard literature like \cite{CT91, NC00}.
Instead, we give a specific overview of the concepts which are
required for understanding this thesis.

\section{Notation and Basic Tools} \label{sec:notation}
For a sequence of variables $x_1,\ldots, x_n$, we use the abbreviation
$\prei[i]{x} \assign x_1,\ldots,x_i$ for the collection of variables
up to index $i$, and we define \mbox{$\prei[0]{x} \assign \emptyset$} to
be the empty string.

For a set $I=\{i_1, i_2, \ldots, i_{\ell} \} \subseteq \{1, \ldots,
n\}$ and a $n$-bit string $x \in \nbit$, we define $x|_I \assign
x_{i_1} x_{i_2} \cdots x_{i_{\ell}}$. It is sometimes convenient that
all \index{substring}substrings of this form have the same length, irrespective of the
actual size $\ell$ of the index set $I$. Therefore, we define the
$n$-bit string $x \pad_I \assign x_{i_1} x_{i_2} \cdots
x_{i_{\ell}} 0 \cdots 0$ to be the original substring padded with $n-
\ell$ zeros.

Most logarithms in this thesis are with respect to base 2 and denoted
by \index{log@$\log(\cdot)$}$\log(\cdot)$. However, when needed,
\index{ln@$\ln(\cdot)$}$\ln(\cdot)$ denotes the natural logarithm to base
$e$.

We write $\ball{\delta n}(x)$ for the ball of all $n$-bit strings at
\index{Hamming distance}Hamming distance at most $\delta n$ from $x$. Note that the number of
elements in $\ball{\delta n}(x)$ is the same for all $x$, we denote it
by $\ball{\delta n} \assign |\ball{\delta n}(x)|$. It is well known
that $\ball{\delta n} \leq 2^{n h(\delta)}$, where
\[h(p)\assign -\big(p\cdot\log{p} + (1-p)\cdot\log{(1-p)}\big)\]
is the \index{binary entropy function}binary entropy function.

We denote by \index{negl@$\negl{n}$}$\negl{n}$ any function of $n$
smaller than the inverse of any polynomial provided $n$ is
sufficiently large.

If we want to choose two symbols $+$ or $\times$ according to the bit
$b \in \{0,1\}$, we write $[ +, \times ]_b$. The \index{Kronecker
delta}Kronecker delta function is defined as
\[
\delta_{i,j}= \left\{ \begin{array}{r@{\quad \mbox{if} \quad}l} 1 & i=j, \\ 0 & i
    \neq j. \end{array} \right.
\]
The \index{indicator random variable}indicator random variable
$\id_\ev$ equals 1, if the event $\ev$ occurs and $0$ else.

\index{convex function}\index{concave function}
\begin{definition}[convex/concave function] \label{def:convexfunction}
A function $f: \reals \rightarrow \reals$ is \emph{convex} on the
interval $[a,b]$, if for any two points $x,y \in [a,b]$ and
$0 \leq s \leq 1$, it holds that
\[
f(s x + (1-s) y) \leq s f(x) + (1-s) f(y).
\]
Analogously, the function is \emph{concave} on $[a,b]$, if 
\[
f(s x + (1-s) y) \geq s f(x) + (1-s) f(y).
\]
\end{definition}

\index{Jensen's inequality}
\begin{lemma}[Jensen's inequality] \label{lem:jensen}
Let $f: \reals \rightarrow \reals$ be a convex function on $\reals$ and let $x_1,
\ldots, x_n \in \reals$. Let $p_1, \ldots, p_n \in [0,1]$ be such that
$\sum_i p_i=1$. Then,
\[ f\left(\sum_{i=1}^n p_i x_i \right) \leq \sum_{i=1}^n p_i f(x_i) \,
.
\]
For $x_1=x_2=\ldots=x_n$, equality holds.
\end{lemma}

\index{Cauchy-Schwarz inequality}
\begin{lemma}[Cauchy-Schwarz inquality] \label{lem:cauchyschwarz}
For real numbers $x_1,\ldots,x_n$ and $y_1,\ldots,y_n$, the following holds
\[ \left(\sum_{i=1}^n x_i \cdot y_i \right)^2 \leq \left(\sum_{i=1}^n x_i^2\right)
\cdot \left(\sum_{i=1}^n y_i^2\right) \, .
\]
\end{lemma}
\begin{proof}
Note that $\sum_{i=1}^n (x_i \cdot z + y_i)^2$ is a quadratic
polynomial $a \cdot z^2 + b z +c$ without real roots unless all
$x_i/y_i$ are equal. Therefore, its discriminant $b^2 - 4ac$ is non-positive:
\[ 4 \left(\sum_{i=1}^n x_i \cdot y_i \right)^2 - 4 \left(\sum_{i=1}^n
  x_i^2 \right) \cdot \left(\sum_{i=1}^n y_i^2 \right) \leq 0 \, .
\]
\end{proof}

\section{Probability Theory}
For a discrete probability space $(\Omega,P)$, we write $P[\ev]$ for
the probability of the \index{event $\ev$}event $\ev \subset \Omega$,
and we write $P_X$ for the \index{probability distribution!of random
variable}distribution of the random variable $X:\Omega\to\cX$ taking
values in the finite set $\cX$. As is common practice, we do not refer
to the probability space $(\Omega,P)$ but leave it implicitly defined
by the joint probabilities of all considered events and random
variables. For two random variables $X$ and $Y$ with joint
distribution $P_{XY}$ over $\cX \times \cY$, the \index{probability
distribution!conditional}conditional
probability distribution of \emph{$X$ given $Y$} is defined as $P_{X |
  Y}(x|y) \assign \frac{P_{XY}(x,y)}{P_Y(y)}$ for all $x \in \cX$ and
$y \in \cY$ with $P_Y(y)>0$.  For a probability distribution $Q$ over
$\cal X$, we abbreviate the (overall) probability of a set $L
\subseteq \cal X$ with $Q(L) \assign \sum_{x \in L} Q(x)$.

Let $P$ and $Q$ be two probability distributions over the same finite
domain $\cX$. The {\em \index{variational distance}variational distance}\footnote{also called
\emph{statistical} or \emph{Kolmogorov} distance} $\dist{P,Q}$ between $P$ and $Q$
is defined as $$\dist{P,Q} \assign \frac{1}{2} \sum_{x \in \cX}
\big|P(x)-Q(x)\big| \, .$$  Note that this definition makes sense also for
{\em non-normalized} distributions, and indeed we define and use
$\dist{P,Q}$ for arbitrary positive-valued functions $P$ and $Q$ with
common domain.  In case $\cX$ is of the form $\cX = {\cal U} \times
{\cal V}$, we can expand $\dist{P,Q}$ to $\dist{P,Q} = \sum_u
\dist{P(u,\cdot),Q(u,\cdot)} = \sum_v \dist{P(\cdot,v),Q(\cdot,v)}$.
We write \smash{$P \epsclose Q$} to denote that $P$ and $Q$ are
\index{$\varepsilon$-close}$\varepsilon$-close, i.e., that $\dist{P,Q} \leq \varepsilon$.

By \index{unif@$\unif$}$\unif$ we denote a uniformly distributed binary random variable
independent of anything else, such that $P_{\unif}(b) = \frac12$ for
both $b \in \set{0,1}$, and $\unif^{\ell}$ stands for $\ell$
independent copies of $\unif$.

For a random variable $R$ over the reals $\reals$, its \index{expected
value}expected value is denoted by $\E[R]$. \index{$\E[\cdot]$|see
{expected value}}

\index{Markov's inequality}
\begin{lemma}[Markov's inequality] \label{lem:markov}
For a non-negative real random variable $X$ and $\eps>0$, we have
\[ \Pr\left[ X \geq \frac{\E[X]}{\eps} \right] \leq \eps \, .\]
\end{lemma}
\begin{proof}
For the indicator function $\id_{\ev}$ which equals 1 if the
event $\ev$ occurs and 0 else, we observe that
\[ \frac{\E[X]}{\eps} \cdot \id_{\left\{ X \geq \frac{\E[X]}{\eps}
  \right\} } \leq X \, .
\]
Taking the expected values on both sides, using linearity of the
expectation and rearranging the terms yields the claim.
\end{proof}

\index{Chernoff's inequality}
\begin{lemma}[Chernoff's inequality] \label{lem:chernoff}
  Let $X_1,\ldots,X_n$ be identically and independently distributed
  random variables with Bernoulli distribution, i.e. $X_i=1$ with
  probability $p$ and $X_i=0$ with probability $1-p$. Then $S \assign
  \sum_{i=1}^n X_i$ has binomial distribution with parameters $(n,p)$
  and it holds that
\[ P\left[ \: |S - pn| > \eps n \: \right] \leq 2 e^{- 2 \eps^2 n} \, .
\]
\end{lemma}
See \cite{AS00} or \cite{MR95} for a proof.

\section{Quantum Information Theory} \label{sec:qit}
In this section, we give a very brief introduction to the quantum
notions we use in this thesis, we refer to \cite{NC00, Renner05} for
further explanations.

For any positive integer $d \in \naturals$, $\cH_d$ stands for the
complex \index{Hilbert space}Hilbert space of dimension $d$.
Sometimes, we omit the dimension and simply write $\cH$. The state of
a quantum-mechanical system in $\cH$ is described by a
\emph{\index{density operator}density operator} $\rho$. A density
operator $\rho$ is normalized with respect to the trace norm
($\trace{\rho}=1$), \index{Hermitian}Hermitian ($\rho^*=\rho$) and has
non-negative eigenvalues. \index{$\dens{\cH}$ (set of density
operators)}$\dens{\cH}$ denotes the set of all \emph{density
operators} acting on $\cH$. \index{$\I$ (fully mixed state)}$\I$
denotes the identity matrix (describing the fully mixed state)
renormalized by the appropriate dimension.

A quantum state $\rho \in \dens{\cH}$ is called \emph{\index{pure state}pure} if it is of the form
$\rho=\proj{\varphi}$ for a (normalized) vector $\ket{\varphi} \in \cH$.

A \emph{\index{positive operator-valued measurement}positive
  operator-valued measurement (POVM)} is a family $M = \set{E_x}_{x
  \in \cX}$ of non-negative operators such that $\sum_{x \in \cX} E_x$ equals the identity matrix.
The probability distribution $P_X$ obtained when
applying the POVM $M$ to the quantum state $\rho$ is defined as
$P_X(x) \assign \trace{E_x \rho}$.

The general evolution (like unitary transforms, measurements, applying
noise etc.) of a quantum system in state $\rho$ can be described by a
\emph{\index{quantum operation}quantum operation} \index{$\cE(\rho)$
  quantum operation}$\cE(\rho)$, which is a completely positive
and trace-preserving map, i.e. $\cE$ is linear and maps non-negative
normalized operators $\rho \in \dens{\cH}$ into non-negative
normalized operators $\cE(\rho) \in \dens{\cH}$.

The notion of (variational) distance of two random variables can be
naturally extended to the \emph{\index{trace distance}trace distance} between two density
operators $\rho, \sigma \in \dens{\cH}$ defined by
$\dist{\rho,\sigma} \assign \frac12 \trace{|\rho-\sigma|}$, where we
define $|A| \assign \sqrt{A^* A}$ to be the positive square-root of
$A$. As in the classical case, we write $\rho \approx_{\eps} \sigma$
to denote that $\rho$ and $\sigma$ are $\eps$-close, i.e.
$\dist{\rho,\sigma} \leq \eps$. The trace distance has an operational
meaning in that the value $\frac12 + \frac12 \dist{\rho,\sigma}$ is
the average success probability when distinguishing $\rho$ from
$\sigma$ via a measurement. In fact, the relation to the classical
variational distance becomes evident in $\dist{\rho,\sigma} = \max_M
\dist{M(\rho),M(\sigma)}$ where the maximization is over all POVMs
$M$ and $M(\rho)$ refers to the probability distribution obtained when
measuring $\rho$ using $M$.  Ruskai~\cite{Ruskai94} showed that the trace distance
does not increase under (trace-preserving) quantum operations,
formally $\dist{ \rho,\sigma } \leq \dist{ \cE(\rho),\cE(\sigma) }$ for
any quantum operation $\cE$.

\index{basis!rectilinear|see{basis}}
The pair $\{\ket{0},\ket{1}\}$ denotes the \index{basis!computational}computational or
rectilinear or ``$+$'' basis for the $2$-dimensional Hilbert space
${\mathcal H}_2$.  The \index{basis!diagonal}diagonal or ``$\times$'' basis is defined as
$\{\ket{0}_\times,\ket{1}_{\times}\}$ where
\smash{$\ket{0}_{\times}=(\ket{0}+\ket{1})/\sqrt2$} and
\smash{$\ket{1}_{\times}=(\ket{0}-\ket{1})/\sqrt2$}. The \index{basis!circular}circular or
``$\circlearrowleft$'' basis consists of vectors \smash{$(\ket{0}+i
  \ket{1})/\sqrt2$} and \mbox{$(\ket{0} - i \ket{1})/\sqrt2$}.
Measuring a qubit in the $+\,$-basis (resp.\ $\times$-basis) means
applying the measurement described by projectors $\proj{0}$ and
$\proj{1}$ (resp.  projectors $\ket{0}_\times \bra{0}_\times$ and
$\ket{1}_{\times}\bra{1}_\times$).  When the context requires it, we
write $\ket{0}_+$ and $\ket{1}_+$ instead of $\ket{0}$ respectively
$\ket{1}$. For a $n$-bit string $x \in \nbit$, $\ket{x}_+$ stands for
the state $\bigotimes_{i=1}^n \ket{x_i}_+ \in \cH_{2^n}$ and analogous for
  $\ket{x}_{\times}$.



As mentioned above, the behavior of a quantum state in a register
$\regE$ is fully described by its density matrix~$\rho_\regE$. We
often consider cases where a quantum state may depend on some
classical random variable $X$, in that it is described by the density
matrix $\rho_\regE^x$ if and only if $X = x$. For an observer who has
only access to the register $\regE$ but not to $X$, the behavior of
the state is determined by the density matrix $\sum_x P_X(x)
\rho_\regE^x$. The joint state, consisting of the \emph{c}lassical $X$ and
the \emph{q}uantum register $\regE$ and therefore called 
\emph{\index{cq-state}cq-state}, is
described by the density matrix $\sum_x P_X(x) \proj{x} \otimes
\rho_\regE^x$.  In order to have more compact expressions, we use the
following notation. We write
$$
\rho_{X\regE} = \sum_x P_X(x) \proj{x} \otimes \rho_\regE^x 
\qquad\text{and}\qquad
\rho_\regE = \tr_X(\rho_{X\regE}) = \sum_x P_X(x)\rho_\regE^x \, .
$$
More general, for any event $\ev$, we write $$
\rho_{X\regE|\ev} = \sum_x P_{X|\ev}(x) \proj{x} \otimes \rho_\regE^x
\quad\text{and}\quad
\rho_{\regE|\ev} = \tr_X(\rho_{X\regE|\ev}) = \sum_x P_{X|\ev}(x)\rho_\regE^x \, .
$$
We also write $\rho_X = \sum_x P_X(x) \proj{x}$ for the quantum
representation of the classical random variable $X$ (and similarly for
$\rho_{X|\ev}$).
This notation extends naturally to quantum states that depend on
several classical random variables (i.e. to ccq-states, cccq-states etc.).
Given a cq-state $\rho_{X \regE}$ as above, by saying that there
exists a random variable $Y$ such that $\rho_{XY\regE}$ satisfies some
condition, we mean that $\rho_{X \regE}$ can be understood as
$\rho_{X\regE} = \tr_Y(\rho_{XY\regE})$ for a ccq-state
$\rho_{XY\regE}$ that satisfies the required condition.


Obviously, $\rho_{X\regE} = \rho_X \otimes \rho_\regE$ holds if and only if
the quantum part is independent of $X$ (in that $\rho_\regE^x = \rho_\regE$ for
any $x$), where the latter in particular implies that no information
on $X$ can be learned by observing only $\rho_{\regE}$.  Furthermore, if
$\rho_{X\regE}$ and $\rho_X \otimes \rho_\regE$ are
$\varepsilon$-close in terms of their trace distance
$\dist{ \rho,\sigma } = \frac{1}{2} \tr(|\rho-\sigma|)$, then the real
system $\rho_{X\regE}$ ``behaves'' as the ideal system $\rho_X \otimes
\rho_\regE$ except with probability~$\varepsilon$ (as explained by
Renner and K\"onig in~\cite{RK05}) in that
for any evolution of the system no observer can distinguish the real
from the ideal one with advantage greater than $\varepsilon$.

\section{Entropies}\label{sec:entropies}
\index{entropy!classical R\'enyi}
\subsection{Classical R\'enyi Entropy} \label{app:Renyi}
\begin{definition} \label{def:ordersum}
  Let $P$ be a probability distribution over the finite set $\cX$
  and $\alpha \in [0,\infty]$. The \index{$\alpha$-order sum}\emph{$\alpha$-order sum} of the
  probability distribution~$P$ is defined as $\pi_\alpha(P) \assign
  \sum_{x \in \cX} P(x)^\alpha$.
\end{definition}

In the limits $\alpha \rightarrow \infty$ and $\alpha \rightarrow 0$,
we set $\pi_\infty(P) \assign \max_{x \in \cX} P(x)$ and $\pi_0(P) \assign
|\Set{x \in \cX}{P(x)>0}$.

\begin{definition}[R\'enyi entropy~\cite{Renyi61}] \label{def:renyi}
  Let $P$ be a probability distribution over the finite set $\cX$ and
  $\alpha \in [0,\infty]$. The {\em R\'enyi entropy of order $\alpha$}
  is defined as
$$
\H_{\alpha}(P) \assign \frac{1}{1-\alpha} \, \log \left(
  \pi_{\alpha}(P) \right) = - \log \Big( \big(\sum_{x \in \cX}
P(x)^{\alpha})^{\frac{1}{\alpha-1}} \Big) \, .
$$
\end{definition}

In the limit \ \!$\alpha \rightarrow \infty$, we obtain the {\em
\index{entropy!min-}min-entropy} $\H_{\infty}(P) = - \log \big(\max_{x \in \cX}
P(x)\big)$ and for $\alpha \rightarrow 0$, we obtain
\emph{\index{entropy!max-}max-entropy} $\H_0(P) = \log | \Set{x \in \cX}{P(x)>0} |$.
Another important special case is the case $\alpha = 2$, also known as
{\em \index{collision probability}collision probability} $\pi_2(P)= \sum_{x \in \cX} P(x)^2$ and
\index{entropy!collision}\emph{collision entropy} $\H_2(P) = - \log \big(\sum_x P(x)^2\big)$.

For the limit $\alpha \rightarrow 1$, we can use Jensen's inequality
(Lemma~\ref{lem:jensen}) with $p_x \assign P(x)$ to obtain
\[
- \frac{1}{\alpha-1} \log \left(\sum_x p_x P(x)^{\alpha-1} \right) \leq - \sum_x
p_x \log \left( (P(x)^{\alpha-1})^{\frac{1}{\alpha-1}} \right)  \, .
\]
In the limit $\alpha \rightarrow 1$, all $P(x)^{\alpha-1}$ go to 1 and
therefore, equality holds and we obtain the standard definition of
\emph{\index{entropy!Shannon}Shannon entropy} $\H(P) \assign - \sum_x P(x) \log P(x)$ as in~\cite{Shannon48}.

For a random variable $X$ with probability distribution $P_X$, we will
most often slightly abuse notation and use the common shortcut
$\H_{\alpha}(X)$ instead of $\H_{\alpha}(P_X)$. For a fixed random
variable $X$ over the finite set $\cX$, $\alpha \mapsto \H_\alpha(X)$ is a decreasing
function on $[0,\infty]$:
\[ \log|\cX| \geq H_0(X) \geq \H(X) \geq \H_2(X) \geq H_\infty(X) \, ,
\]
with equality if and only if $X$ is uniform over a subset of
$\cX$. Furthermore, we have that for $\alpha>1$, $\pi_\alpha(X) =
\sum_x P_X(x)^\alpha \geq \max_x P_X(x)^\alpha$ and therefore,
\[ \H_\alpha(X) = \frac{1}{1-\alpha} \log \pi_\alpha(X) \leq
\frac{1}{1-\alpha} \log \max_x P_X(x)^\alpha = \frac{\alpha}{1-\alpha}
\log \max_x P_X(x) \, , 
\]
which implies the following relation between R\'enyi entropies of
order $\alpha>1$:
\begin{equation} \label{eq:renyirelation}
\frac{\alpha-1}{\alpha} \H_\alpha(X) \leq \H_{\infty}(X) \, .
\end{equation}

\index{entropy!conditional R\'enyi}
\subsubsection{Conditional R{\'e}nyi entropy} \label{sec:conditionalrenyientropy}
The R{\'e}nyi entropy $\H_{\alpha}(X|Y\=y)$ of $X$ given the event
$Y=y$ is naturally defined as $\H_{\alpha}(X|Y\=y) =
\frac{1}{1-\alpha} \, \log\big(\sum_x P_{X|Y=y}(x)^{\alpha}\big)$. We
can define the \index{$\alpha$-order sum}\emph{conditional $\alpha$-order sum of $X$ given $Y$}
and \emph{conditional R{\'e}nyi entropy} by
\begin{align*}
\pi_\alpha(X|Y) &\assign \max_y \sum_{x} P_{X|Y=y}(x)^{\alpha} \quad\mbox{and}\quad  \H_\alpha(X|Y) \assign \frac{1}{1-\alpha} \log(\pi_\alpha(X|Y)) \, . 
\end{align*}
In the limits we have, $\pi_\infty(X|Y)=\max_{x,y} P_{X|Y=y}(x)$,
$\pi_0(X|Y) = \max_y |\Set{x \in \cX}{P_{X|Y=y}(x)>0}|$. For the
conditional min-, collision- and max-entropy, we get
\begin{align*}
\H_{\infty}(X|Y) &\assign \min_y \H_{\infty}(X|Y=y) = \min_{x,y} -\log
P_{X|Y=y}(x), \\
\H_2(X|Y) &\assign \min_y \H_2(X|Y=y) = \min_y -\log \left( \sum_x
  P_{X|Y=y}(x)^2 \right) \, ,\\
\H_0(X|Y) &\assign \max_y \H_0(X|Y=y) = \max_y \, \log | \Set{x \in
  \cX}{P_{X|Y=y}(x)>0} | .
\end{align*}
In the limit $\alpha \downarrow 1$, we get $\H_{\downarrow 1}(X|Y) =
\min_y \H(X|Y=y)$ and for $\alpha \uparrow 1$, we get $\H_{\uparrow
  1}(X|Y) = \max_y \H(X|Y=y)$ which might be different. However, the
standard definition of conditional \index{entropy!Shannon}Shannon entropy is neither of
those, but ``in between'':
\[
\H(X|Y) \assign \sum_y P_Y(y) \H(X|Y=y) = \sum_{x,y} P_{XY}(x,y) \log
P_{X|Y=y}(y) \, .
\]

We note that in the literature, $\H_{\alpha}(X|Y)$ is sometimes
defined as average over $Y$, $\sum_y P_Y(y) \, \H_{\alpha}(X|Y\=y)$,
like for Shannon entropy.  However, we define the more natural
following notion. For $1 < \alpha < \infty$, we define the {\em
  average conditional} R{\'e}nyi \index{entropy!average conditional
R\'enyi}entropy $\tH_{\alpha}(X|Y)$ as
$$
\tH_{\alpha}(X|Y) \assign - \log \Big( \sum_y P_Y(y) \big(\sum_x
P_{X|Y}(x|y)^{\alpha})^{\frac{1}{\alpha-1}} \Big) \, ,
$$
and as $\tH_{\infty}(X|Y) = -\log \big(\sum_y P_Y(y) \max_x
P_{X|Y}(x|y)\big)$ for $\alpha = \infty$. This notion is useful in
particular because it has the property that if the {\em average}
conditional R{\'e}nyi entropy is large, then the conditional R{\'e}nyi
entropy is large with high probability:
\begin{lemma}\label{lemma:average}
  Let $\alpha > 1$ (allowing $\alpha = \infty$) and $t \geq 0$. Then
  with probability at least $1 - 2^{-\sp}$ (over the choice of $y$)
  $\H_{\alpha}(X|Y=y) \geq \tH_{\alpha}(X|Y) - \sp$.
\end{lemma}
\begin{proof}
By definition of average conditional R\'enyi entropy, we have
\[ 2^{-\tH_{\alpha}(X|Y)} = \E_y\left[
  (\pi_{\alpha}(X|Y=y))^\frac{1}{\alpha-1} \right] \, .
\]
By the \index{Markov's inequality}Markov's inequality (Lemma~\ref{lem:markov}), we get that 
\[ \Pr_y \left[ \pi_\alpha(X|Y=y)^{\frac{1}{\alpha-1}} \geq
  2^{-\tH_\alpha(X|Y) + \sp} \right] \leq 2^{-\kappa} \,
\]
and therefore, the probability (over $y$) that $\H_\alpha(X|Y=y)
\leq \tH_\alpha(X|Y) - \sp$ is at most $2^{-\sp}$.
\end{proof}

As long as $\alpha>1$, the minimization (or average) over $y$ is the
same for all orders of R\'enyi entropy hence,
Equation~\eqref{eq:renyirelation} translates to (average) conditional
R{\'e}nyi entropy:
\begin{lemma}\label{lemma:bounds}
  For any $1 < \alpha < \infty$, we have
\begin{align*}
\H_2(X|Y) &\geq \H_{\infty}(X|Y) \geq \frac{\alpha-1}{\alpha}
\H_{\alpha}(X|Y) \\
\tH_2(X|Y) &\geq \tH_{\infty}(X|Y) \geq \frac{\alpha-1}{\alpha}
\tH_{\alpha}(X|Y) .
\end{align*}
\end{lemma}

\subsubsection{Concavity}
\begin{lemma} \label{lem:concavity}
  For $0 \leq \alpha \leq 1$, R\'enyi Entropy is a \emph{concave
    \index{entropic functional}entropic functional}, i.e., for $0 \leq s \leq 1$ and
  distributions $P,Q$, we have
\[ \H_\alpha(sP+(1-s)Q) \geq s
  \H_\alpha(P) + (1-s) \H_\alpha(Q) \, . \]
\end{lemma}
For the case of Shannon entropy, note that the function $f(p)\assign
-p \log{p}$ has derivatives $f'(p)=-1-\log p$ and $f''(p)=-1/p$ and
$f''(p) \leq 0$ for $0\leq p \leq 1$. Therefore, $f(p)$ is concave and
we have 
\begin{align*}
\H(sP+(1-s)Q) &= \sum_x f(sP(x)+(1-s)Q(x)) \geq \sum_x
sf(P(x)) + (1-s)f(Q(x))\\ 
&= s \sum_x f(P(x)) + (1-s) \sum_x f(Q(x)) = s \H(P) + (1-s) \H(Q).
\end{align*}

Higher-order R\'enyi entropy is not necessarily concave as the
following example illustrates. Consider the distributions
$P(x)=\delta_{x,0}$ and $Q(x)=2^{-n}$ over $\set{0,1}^n$ with
$\H_2(P)=0$ and $\H_2(Q)=n$. For the equal mixture of these
distributions holds $\H_2((P+Q)/2) = -\log(1/4)+O(2^{-n}) \approx 2 <
n/2 = (\H_2(P)+\H_2(Q))/2$ for $n>5$.

\subsubsection{Fano's Inequality}
\index{Fano's inequality}
\begin{lemma}[Fano's Inequality] \label{lem:fano}
Let $X \!\leftrightarrow\! Y \!\leftrightarrow\! X'$ be a \index{Markov
chain}Markov
chain\footnote{Think of $X'$ as guess of $X$ based only on
  $Y$.}. Then, for the error probability $p_e \assign P[X \neq X']$,
it holds
\[ \H(X|Y) \leq h(p_e) + p_e \cdot \log(|\cX|-1) \, .
\]
\end{lemma}
\begin{proof}
We denote by $E \assign \id_{\set{X \neq X'}}$ the indicator random
variable  of the event $\set{X \neq X'}$ that the guess was not successful. By
the chain rule for Shannon entropy, we can write
\[ \H(XE|Y) = \H(X|Y)+\H(E|XY) = \H(E|Y) + \H(X|EY) \]
We observe that $H(E|Y) \leq h(p_e)$, $\H(E|XY)\geq 0$ and 
\[ \H(X|EY) = (1-p_e) \H(X| \set{X=X'} Y) + p_e \H(X|\set{X \neq X'} Y) = p_e
\log(|\cX|-1) \]
and the claim follows by rearranging the terms.
\end{proof}

\subsection{Smooth R\'enyi Entropy} \label{sec:defsmoothrenyientropy}
\index{entropy!smooth}Smooth min- and max-entropies were introduced by Renner and Wolf in~\cite{Renner05,RW05}\footnote{The notion of \emph{smoothing a
    probability distribution} was already used in~\cite{ILL89}.
  Furthermore, a different kind of \emph{smooth R\'enyi entropy} (not
  equivalent to the ones used here) was introduced by Cachin~\cite{Cachin97}.}. They are families of entropy measures
parametrized by non-negative real numbers $\varepsilon$, called the
{\em smoothness}. It is a generalization of the notions of
conditional min- and max-entropy defined in the last section.
\begin{align*}
\hiee{X|Y} &\assign \max_{\ev} \min_{x,y} - \log
\left(\frac{P_{XY\ev}(x,y)}{P_Y(y)} \right) \, ,\\
\hmaxee{X|Y} &\assign \min_{\ev} \max_{y} \log
|\Set{x \in \cX}{\frac{P_{XY\ev}(x,y)}{P_Y(y)} > 0}| \,
\end{align*}
where the maximum/minimum ranges over all events $\ev$ with
probability $\Pr[\ev] \geq 1-\eps$.  $P_{XY\ev}(x,y)$ is the
probability that $\ev$ occurs and $X,Y$ take values $x,y$. Hence, the
``distribution'' $P_{XY\ev}$ is not normalized.

For a given distribution $P_{XY}$, it is easy to compute its smooth
min-entropy (max-entropy), simply by cutting a maximum mass of
$\eps$ off the largest (smallest) probabilities.

Informally, the statement $\hiee{X} = r$ can be understood that the
standard min-entropy of $X$ is close to $r$, except with probability
$\varepsilon$. As $\varepsilon$ can be interpreted as an error
probability, we typically require $\varepsilon$ to be negligible in
the security parameter. 

The reason why we only define the min- and max-versions of smooth
R\'enyi entropy is that it is shown in \cite{RW05} that for example
smooth R\'enyi entropy of order $\alpha>1$ obeys
\[ \hie{\eps+\eps'}{X|Y} + \frac{\log(1/\eps')}{\alpha-1} \geq
\H_{\alpha}^{\eps}(X|Y) \geq \hiee{X|Y} \, .
\]
and hence is equivalent to smooth min-entropy up to an additive term
which depends on $\alpha$ and the smoothness $\eps'$. An analogue
statement holds for $\alpha<1$ and smooth max-entropy. As pointed out
in \cite{RW05}, for $\eps=0$ the relation above shows for example that
$\H_2(X)$ cannot be larger than $\hiee{X} + \log(1/\eps)$ whereas for
the non-smooth versions, we only know from
Equation~(\ref{eq:renyirelation}) that $\H_2(X) \leq 2 \hmin(X)$.

Most importantly, smooth min- and max-entropy have an
\emph{operational meaning} as they provide the answer to two
fundamental information-theoretic problems: 
\begin{itemize} 
\item $\hiee{X|Y}$ is the maximum amount\footnote{up to some small
    additive error term which depends logarithmically on $\eps$} of
  randomness that can be extracted from $X$ and an independent random
  string $R$, such that except with probability $\eps$, the extracted
  string looks completely uniform to an adversary who knows $Y$ and
  learns $R$. This falls into the setting of privacy amplification,
  see Section~\ref{sec:pa} below. \index{randomness extraction}
\item $\hmaxee{X|Y}$ is the minimal
  length\addtocounter{footnote}{-1}\footnotemark\ of an encoding
  computed from $X$ and some additional independent randomness $R$,
  such that except with probability $\eps$, someone knowing $Y$ and
  $R$ can reconstruct $X$ from the encoding. This is a
  data-compression problem which is often called
  \emph{\index{information reconciliation}information reconciliation}
  or \emph{\index{error correction}error correction} in cryptographic
  settings.
\end{itemize}

In \cite{RW05}, it is shown that smooth min- and max-entropies enjoy
several Shannon-like properties such as the \index{entropy!chain
rule}chain rule (see Lemma~\ref{lem:chain} below),
\index{entropy!sub-additivity}sub-additivity $\hiee{XY} \leq
\hie{\eps+\eps'}{X} + \hmaxe{\eps'}{Y}$ and
\index{entropy!monotonicity}monotonicity $\hiee{X} \leq \hiee{XY})$.
\index{sub-additivity|see{entropy}}
\index{monotonicity|see{entropy}}

\begin{lemma}[Chain Rule~\cite{RW05}]\label{lem:chain}
For all $\varepsilon,\varepsilon' > 0$, we have 
\[\hie{\varepsilon+\varepsilon'}{X | Y} > \hie{\varepsilon}{XY}
- \hmax(Y) - \log{\left(\frac{1}{\varepsilon'}\right)}. \] 
\end{lemma}

As a consequence of the asymptotic equipartition property (cf.~\cite{CT91}), smooth R\'enyi entropy is asymptotically equal to
Shannon entropy in the following sense. 
\begin{lemma}[\cite{RW05,HR06}] \label{lem:asymptshannon}
Let $(X_1,Y_1), \ldots, (X_n, Y_n)$ be $n$ independent pairs of random
variables distributed according to $P_{XY}$. Then, for any
$\alpha \neq 1$,
\[ \lim_{\eps \rightarrow 0} \lim_{n \rightarrow \infty}
\frac{\H_{\alpha}^{\eps}(X^n | Y^n)}{n} = \H(X | Y).
\]
\end{lemma} 
Note that such a lemma does \emph{not} hold at all for non-smooth
R\'enyi entropies.

To provide some intuition about smooth min-entropy, the following lemma shows how to translate smooth min-entropy
back to regular conditional min-entropy.
\begin{lemma}\label{lemma:smooth->ordinary}
  If $\hiee{X|Y} = r$ then there exists an event $\ev'$ such that
  $\Pr(\ev') \geq 1-2\varepsilon$ and $\hmin(X|\ev',Y\!=\!y) \geq r-1$
  for every $y$ with $P_{Y \ev'}(y) > 0$.
\end{lemma}
\begin{proof}
By definition of smooth min-entropy, there exists an event $\ev$ with
$\Pr(\ev) \geq 1 - \varepsilon$ and such that $\mbox{$\hmin(X\ev|Y\!=\!y)$}
\geq r$ for all $y$, and thus $P_{X\ev|Y}(x|y)$ $\leq 2^{-r}$ for all
$x$ and $y$.  Define $\ev'$ by setting for all $x$ and $y$
$$
P_{X\ev'|Y}(x|y) \assign \left\{
\begin{array}{ll}
P_{X\ev|Y}(x|y) & \text{if $P_{\ev|Y}(y) \geq \frac12$} \\
0 & \text{else}
\end{array}
\right.
$$
Then obviously for any $y$ with $P_{Y \ev'}(y) > 0$ and thus
$P_{\ev'|Y}(y) = P_{\ev|Y}(y) \geq \frac12$,
$$
P_{X|\ev'Y}(x|y) = \frac{P_{X\ev'|Y}(x|y)}{P_{\ev'|Y}(y)} \leq
\frac{2^{-r}}{P_{\ev'|Y}(y)} \leq 2^{-r+1} \, .
$$
Furthermore, 
\begin{align}
  1 &- \varepsilon \leq \Pr(\ev)  \nonumber \\[0.7ex]
  &= \Pr(\ev|P_{\ev|Y}(Y)<{\textstyle\frac12})\cdot\Pr(P_{\ev|Y}(Y)<{\textstyle\frac12}) \nonumber \\
 &\quad\; +
  \Pr(\ev|P_{\ev|Y}(Y)\geq{\textstyle\frac12})\cdot\Pr(P_{\ev|Y}(Y)\geq{\textstyle\frac12}) \label{eq:events}\\
  &\leq \frac{1}{2}\Pr(P_{\ev|Y}(Y)<{\textstyle\frac12}) +
  \Pr(P_{\ev|Y}(Y)\geq{\textstyle\frac12}) \nonumber
\end{align}
from which follows that $\Pr(P_{\ev|Y}(Y)<{\textstyle\frac12}) \leq
2\varepsilon$. Thus we can conclude that
\begin{align*}
\Pr(\ev') &\geq \Pr(\ev'|P_{\ev|Y}(Y)\geq{\textstyle\frac12})\cdot\Pr(P_{\ev|Y}(Y)\geq{\textstyle\frac12}) \\ 
&\geq \Pr(\ev|P_{\ev|Y}(Y)\geq{\textstyle\frac12})\cdot\Pr(P_{\ev|Y}(Y)\geq{\textstyle\frac12}) \\ 
&\geq 1 - \varepsilon - \frac{1}{2}\Pr(P_{\ev|Y}(Y)<{\textstyle\frac12}) \\
&\geq 1 - 2\varepsilon
\end{align*} 
where the second-last inequality follows from (\ref{eq:events}), and noting (once more) that $\Pr(\ev|P_{\ev|Y}(Y)<{\textstyle\frac12}) < \frac12$. 
\end{proof}

\subsection{Min-Entropy-Splitting Lemma} \index{min-entropy splitting lemma}
For proving reductions between variants of oblivious transfer in
Section~\ref{sec:application} and the security of \OT in the
bounded-quantum storage in Chapter~\ref{chap:12OT}, we will make use
of the following min-entropy splitting lemma. Note that if the joint
entropy of two random variables $X_0$ and $X_1$ is large, then one is
tempted to conclude that at least one of $X_0$ and $X_1$ must still
have large entropy, e.g.\ half of the original entropy. Whereas this
is indeed true for Shannon entropy, it is in general not true for
min-entropy. The following lemma, though, which first appeared in a
preliminary version of~\cite{Wullschleger07}, shows that it {\em is}
true in a randomized sense.
\begin{lemma}[Min-Entropy-Splitting Lemma]\label{lemma:ESL}
  Let $\varepsilon \geq 0$, and let $X_0,X_1$ be random variables with
  \mbox{$\hmin^{\varepsilon}(X_0 X_1) \geq \alpha$} 
Then, there exists a
  random variable $C \in \set{0,1}$ such that
  \mbox{$\hmin^{\varepsilon}(X_{1-C} C) \geq \alpha/2$}.
\end{lemma}
\begin{proof}
Below, we give the proof for $\varepsilon = 0$, i.e., for ordinary
(non-smooth) min-entropy. The general claim for smooth min-entropy
follows immediately by observing that the same argument also works for
non-normalized distributions with a total probability smaller than 1.

We extend the probability distribution $P_{X_0 X_1}$ as follows to
$P_{X_0 X_1 C}$. Let $C=1$ if $P_{X_1}(X_1) \geq 2^{-\alpha/2}$ and
$C=0$ otherwise.  We have that for all $x_1$, $P_{X_1 C}(x_1,0)$
either vanishes or is equal to $P_{X_1}(x_1)$. In any case, $P_{X_1
  C}(x_1,0) < 2^{-\alpha/2}$.
  
  On the other hand, for all $x_1$ with $P_{X_1 C}(x_1,1)>0$, we have
  that $P_{X_1 C}(x_1,1)=P_{X_1}(x_1) \geq 2^{-\alpha/2}$ and
  therefore, for all $x_0$,
\begin{equation*} 
 P_{X_0 X_1 C}(x_0,x_1,1) \leq 2^{-\alpha} =2^{-\alpha/2} \cdot
 2^{-\alpha/2} \leq 2^{-\alpha/2} P_{X_1}(x_1).
\end{equation*}
Summing over all $x_1$ with $P_{X_0 X_1 C}(x_0,x_1,1) > 0$, and thus with $P_{X_1 C}(x_1,1) > 0$, results in
$$
P_{X_0 C}(x_0,1) \leq \sum_{x_1} 2^{-\alpha/2} P_{X_1}(x_1)
\leq 2^{-\alpha/2}.
$$
This shows that $P_{X_{1-C} C}(x,c) \leq 2^{-\alpha/2}$ for all $x,c$.
\end{proof}

The corollary below follows rather straightforwardly by noting that
(for normalized as well as non-normalized distributions) $\hmin(X_0
X_1| Z) \geq \alpha$ holds exactly if $\hmin(X_0 X_1 | Z\!=\!z)
\geq \alpha$ for all $z$, applying the Min-Entropy Splitting Lemma,
and then using the chain rule, Lemma~\ref{lem:chain}.

\begin{corollary}\label{cor:ESL}
  Let $\eps \geq 0$ be given, and let $X_0,X_1,Z$ be random variables with $\hiee{X_0 X_1 | Z} \geq \alpha$. Then, there exists a binary
  random variable $C \in \set{0,1}$ such that for $\eps' > 0$,
\[
 \hie{\eps+\eps'}{X_{1-C} | Z C} \geq \alpha/2 - 1 - \log(1/\eps').
\]
\end{corollary}

\subsection{Entropy of Quantum States}
As pointed out in \cite{RK05}, R\'enyi entropy $\H_{\alpha}(\rho)$ can
also be defined for a quantum state $\rho \in \dens{\cH}$. For $\alpha
\in [0,\infty]$ and $\rho \in \dens{\cH}$, we have
\[ \H_{\alpha}(\rho) \assign \frac{1}{1-\alpha}
\log\left(\trace{\rho^{\alpha}} \right). \] 
In the limit cases $\alpha \rightarrow 0$ and $\alpha \rightarrow
\infty$, we obtain $\H_0(\rho)=\log(\rank{\rho})$ and
$\H_{\infty}(\rho) = -\log \left( \lambda_{\max}(\rho) \right)$,
where $\lambda_{\max}(\rho)$ denotes the maximum eigenvalue of $\rho$.
For $\alpha=2$, we obtain the \emph{collision entropy} $\H_2(\rho)
= -\log{\left( \sum_i \lambda_i^2\right)}$, where $\{\lambda_i\}_i$
are the eigenvalues of $\rho$.

For a classical random variable $X$ encoded in $\rho_X= \sum_x P_X(x)
\proj{x}$, it holds that that $\H_{\alpha}(\rho_X) = \H_{\alpha}(X)$.


For deriving our version of the privacy-amplification theorem in the
next section, we need the slightly more involved version of quantum
conditional min-entropy from \cite{Renner05}.
\begin{definition}[\cite{Renner05}] \label{def:qminentropy}
Let $\rho_{AB} \in \dens{\cH_A \otimes \cH_B}$ and $\sigma_B \in
\dens{\cH_B}$. The min-entropy of $\rho_{AB}$ relative to $\sigma_B$
is \index{entropy!min-}
\[ \qhmin(\rho_{AB} | \sigma_B) \assign - \log \lambda
\]
where $\lambda$ is the minimum real number such that $\lambda \cdot \id_A
\otimes \sigma_B - \rho_{AB}$ is non-negative.

The min-entropy of $\rho_{AB}$ given $\cH_B$ is
\[ \qhmin(\rho_{AB} | B) \assign \sup_{\sigma_B} \qhmin(\rho_{AB} | \sigma_B)
\]
where the supremum ranges over all $\sigma_B \in \dens{\cH_B}$.
\end{definition}

Similar to the classical case, the smooth version can be defined as follows.
\begin{definition}[\cite{Renner05}] \label{def:qsmoothminentropy}
Let $\rho_{AB} \in \dens{\cH_A \otimes \cH_B}$, $\sigma_B \in
\dens{\cH_B}$, and $\eps \geq 0$. The $\eps$-smooth min-entropy of
$\rho_{AB}$ relative to $\sigma_B$ is \index{entropy!smooth min-}
\[ \hminee(\rho_{AB} | \sigma_B) \assign \sup_{\ol{\rho}_{AB}} \qhmin(\ol{\rho}_{AB} | \sigma_B)
\]\
where the supremum ranges over the set $\cB^{\eps}(\rho_{AB})$
containing all Hermitian, non-negative operators $\ol{\rho}_{AB}$
acting on $\cH_A \otimes \cH_B$ such that $\dist{ \ol{\rho}_{AB},
\rho_{AB} } \leq 2 \eps$ and $\trace{\ol{\rho}_{AB}} \leq 1$.

The $\eps$-smooth min-entropy given $\cH_B$ is
\[ \hminee(\rho_{AB} | B) \assign \sup_{\sigma_B} \hminee(\rho_{AB} | \sigma_B)
\]
where the supremum ranges over all $\sigma_B \in \dens{\cH_B}$.
\end{definition}
To compute $\hminee(\rho_{XB} | \sigma_B)$ where $\rho_{XB}$ is a
cq-state, the supremum can be restricted to states
$\ol{\rho}_{XB} \in \cB^{\eps}(\rho_{XB})$ which are classical on
$\cH_X$ as well \cite[Remark~3.2.4]{Renner05}.

There is a chain rule for smooth min-entropy, proven in \cite[Lemma
3.2.9]{Renner05}.
\begin{lemma}[\cite{Renner05}] \label{lem:qchainrule}
  Let $\rho_{XUE} \in \dens{\cH_X \otimes \cH_U \otimes \cH_E}$,
  $\sigma_U \in \dens{\cH_U}$, and let $\sigma_E \in \dens{\cH_E}$ be
  the fully mixed state on the image of $\rho_E$, and let $\eps \geq
  0$. Then
\[ \hminee(\rho_{XUE} | \sigma_U) - \qhmax(\rho_E) \leq
  \hminee(\rho_{XUE} | \sigma_U \otimes \sigma_E).
\]
\end{lemma}

The following two lemmas state that dropping a quantum register
cannot increase the (smooth) min-entropy.
\begin{lemma} \label{lem:dropquantum_notsmooth}
Let $\rho_{XUQ} \in \dens{\cH_X \otimes \cH_U \otimes {\cal
    H}_Q}$ be a ccq-state. Then,
$$\qhmin(\rho_{XUQ} | \rho_U) \geq \qhmin(\rho_{XU} | \rho_U).$$
\end{lemma}
\begin{proof}
For $\lambda \assign 2^{-\qhmin(\rho_{XU} | \rho_U)}$, we have
by Definition~\ref{def:qminentropy} that $\lambda \cdot \id_X
\otimes \rho_U - \rho_{XU} \geq 0$. Using that both $X$ and $U$ are
classical, we derive that for all $x,u$, it holds $\lambda \cdot p_u -
p_{xu} \geq 0$, where $p_u$ and $p_{xu}$ are shortcuts for the
probabilities $P_U(u)$ and $P_{XU}(x,u)$.
Let the normalized conditional operator $\ol{\rho}_Q^{x,u}$ be the
quantum state conditioned on the event that $X=x$ and $U=u$,
i.e.
\[\sum_{x,u} p_{xu} \, \ol{\rho}_Q^{x,u} \otimes \proj{xu}= \rho_{QXU}.\]
Then,
\[\sum_{x,u} \lambda \cdot p_u \, \ol{\rho}_Q^{x,u} \otimes \proj{xu}
- p_{xu} \, \ol{\rho}_Q^{x,u} \otimes \proj{xu} \geq 0.\]
Because of $\ol{\rho}_Q^{x,u} \leq \id_Q$, we get
\[\sum_{x,u} \lambda \cdot p_u \, \id_Q  \otimes \proj{xu} - p_{xu} \,
\ol{\rho}_Q^{x,u} \otimes \proj{xu} \geq 0.\]
Therefore, $\lambda \cdot \id_{QX} \otimes \rho_U - \rho_{QXU} \geq 0$
holds, from which follows by definition that $\qhmin(\rho_{XUQ} | \rho_U) \geq
-\log(\lambda)$.
\end{proof}

\begin{lemma} \label{lem:dropquantum}
  Let $\rho_{XUQ} \in \dens{\cH_X \otimes \cH_U \otimes {\cal H}_Q}$
  be a ccq-state and let $\varepsilon \geq 0$. Then
$$\hminee(\rho_{XUQ} | \rho_U) \geq \hminee(\rho_{XU} | \rho_U).$$
\end{lemma}
\begin{proof}
  After the remark after Definition~\ref{def:qsmoothminentropy} above,
  there exists $\sigma_{XU} \in \ball{\varepsilon}(\rho_{XU})$
  classical on $\cH_X \otimes \cH_U$ such that $\hminee(\rho_{XU} |
  \rho_U) = \qhmin(\sigma_{XU} | \sigma_U)$. Because both $X$
  and $U$ are classical, we can write $\sigma_{XU} = \sum_{x,u} p_{xu}
  \proj{xu}$ and extend it to obtain $\sigma_{XUQ} \assign \sum_{x,u}
  p_{xu} \proj{xu} \otimes \ol{\rho}_Q^{x,u}$.
  Lemma~\ref{lem:dropquantum_notsmooth} from above yields $\H_{\rm
    min}(\sigma_{XU} | \sigma_U) \leq \qhmin(\sigma_{XUQ}
  | \sigma_U)$.  We have by construction that $\dist{ \sigma_{XUQ},
  \rho_{XUQ} } = \dist{ \sigma_{XU},\rho_{XU} } \leq 2 \varepsilon$.
  Therefore, $\sigma_{XUQ} \in \ball{\varepsilon}(\rho_{XUQ})$ and
  $\qhmin(\sigma_{XUQ} | \sigma_U) \leq \hminee(\rho_{XUQ}
  | \rho_U).$
\end{proof}

\section[Two-Universal Hashing and Privacy Amplification] {Two-Universal Hashing and Privacy Amplification against Quantum Adversaries} \label{sec:pa}
\subsection{History and Setting of Privacy Amplification}
\index{privacy amplification|(}
Assume two parties Alice and Bob share some information $X$ which is
only partly secure in the sense that an adversary Eve has some partial
knowledge about it. \emph{Privacy Amplification}, introduced by Bennett,
Brassard, and Robert \cite{BBR88}, is the art of transforming this
information $X$ into a highly secure key $K$ by public discussion. The
honest parties want to end up with an almost uniformly distributed key
$K$ about which Eve has only negligible information given the
communication.

A common way to achieve this is to have Alice pick a hash function $f$
at random from a two-universal class of hashing functions (see next
section for the definition), apply it to $X$ and announce it to Bob,
who applies it to $X$ as well. Due to the randomizing properties of a
two-universal function, the output $f(X)$ is close to uniformly
distributed from Eve's point of view. As shown in~\cite{BBR88} and by
Impagliazzo, Levin, Luby~\cite{ILL89} and Bennett, Brassard,
Cr\'epeau, and Maurer~\cite{BBCM95}, the classical \emph{privacy
  amplification theorem} or \emph{left-over hash lemma} (see
Corollary~\ref{thm:classicalpa} below) states that if Eve has some
classical knowledge $W$ about $X$, a secure key of length roughly the
uncertainty of Eve about $X$ (measured in terms of
min-\index{entropy!min-}entropy) can be extracted by
\index{two-universal hashing}two-universal hashing. It is pointed out
in \cite{RW05}, that the maximum amount of extractable randomness is
essentially given by the conditional smooth min-entropy
$\H_{\infty}^{\eps}(X | W)$. \index{hashing|see {two-universal hashing}}

It is interesting to investigate the case when Eve holds quantum
information about $X$. This scenario has been considered by K\"onig,
Maurer, and Renner~\cite{KMR03, RK05, Renner05} and the results
reproduced below show that two-universal hashing works just as well
against quantum as against classical adversaries.

We note that unlike in the classical case,
where many other forms of randomness extractors are known,
two-universal hashing is essentially the only way to perform privacy
amplification against quantum adversaries.\footnote{In a recent paper, K\"onig
and Terhal \cite{KT06} exhibit some extractors which work against
quantum adversaries, but the parameters are far from the classical
ones.} This tool is one of the key ingredients in all protocols
presented in this thesis. It has been widely used in other
applications as well, for example in security proofs of
quantum-key-distribution schemes by Christandl, Renner, Ekert,
Kraus, and Gisin~\cite{CRE04, KGR05, RGK05, Renner05}.

\subsection{Two-Universal Hashing}
An important tool we use is \index{two-universal hashing}\univ\ hashing.
\begin{definition} \label{def:two-universal}
A class \chf{n} of hashing functions from $\nbit$ to $\set{0,1}^{\ell}$ is 
called {\em \univ}, if for any pair $x, y\in\{0,1\}^n$ with $x \neq y$,
and $F$ uniformly chosen from \chf{n}, it holds that
\[
\P\bigl[F(x)=F(y)\bigr] \leq \frac{1}{2^{\ell}}.
\]
\end{definition}
We can also define a slightly stronger notion of \index{two-universal
hashing!strongly}two-universality as follows:
\begin{definition} \label{def:strongly-two-universal}
  A class \chf{n} of hashing functions from $\nbit$ to
  $\set{0,1}^{\ell}$ is called {\em strongly \univ}, if for any pair
  $x, y\in\{0,1\}^n$ with $x \neq y$, and $F$ uniformly chosen from
  \chf{n}, the random variables $F(x)$ and $F(y)$ are independent and
  uniformly distributed over $\set{0,1}^\ell$.
\end{definition}
Several \univ\ and strongly \univ\ classes of hashing functions are
such that evaluating and picking a function uniformly and at random in
$\chf{n}$ can be done efficiently, as pointed out by Wegman and Carter~\cite{WC77,WC79}.

\subsection{Privacy Amplification against Quantum Adversaries}
In the following, we consider the situation where a hash function
is picked randomly from $\chf{n}$ and applied to a classical value $X
\in \nbit$ which is correlated with a quantum register
$\cH_E$. Formally, starting with the cq-state $\rho_{XE} = \sum_{x \in
  \nbit} P_X(x) \, \proj{x} \otimes \rho_E^x$, we obtain
\begin{equation}  \label{eq:paterm}
\rho_{F(X) F E} = \sum_{f \in \chf{n}} \sum_{z \in
  \set{0,1}^{\ell}} \proj{z} \otimes \proj{f} \otimes \!\! \sum_{x \in
  f^{-1}(z)} \!\! P_X(x)  \, \rho_E^x.
\end{equation}
The following privacy-amplification theorem in the
presence of quantum adversaries was first derived in \cite{RK05}. The
version below is from \cite[Corollary 5.6.1]{Renner05}\footnote{Note
  that in \cite{Renner05}, the distance from uniform is defined in
  terms of the trace-norm distance which is twice the variational
  distance used in this thesis.}.
\begin{theorem}[Privacy Amplification~\cite{Renner05}]\label{thm:paoriginal}
  Let $\rho_{XB} \in \dens{\cH_X \otimes \cH_B}$ be a cq-state, where
  $X$ takes values in $\nbit$. Let $\chf{n}$ be a two-universal family
  of hash functions from $\nbit$ to $\set{0,1}^{\ell}$, and let $\eps
  \geq 0$. Then, for the ccq-state $\rho_{F(X) FB}$ defined by
  \eqref{eq:paterm}, it holds
\[ \dist{ \rho_{F(X) FB} , \I \otimes \rho_{FB} } \leq \eps + \frac12
2^{-\frac12 (\hminee(\rho_{XB} | B) - \ell)}.
\]
\end{theorem}

For large parts of this thesis, slightly weaker forms of this theorem
are used. These are derived in the following.
\begin{corollary} \label{thm:pasmooth}
  Let $\rho_{XUE}$ be a ccq-state, where $X$ takes values in $\nbit$, 
  $U$ in the finite domain $\cU$ and register $E$ contains q qubits.
  Let $\chf{n}$ be a two-universal family of hash functions from
  $\nbit$ to $\set{0,1}^{\ell}$, and let $\eps \geq 0$. Then, for the
  cccq-state $\rho_{F(X) FUE}$ defined analogous to \eqref{eq:paterm}, it holds
\begin{align}
  \dist{ \rho_{F(X)FUE} , \I \otimes \rho_{F U E} } &\leq \frac{1}{2} \,
  2^{-\frac{1}{2}\big(\H_{\infty}^{\varepsilon}
(X | U)-q-\ell\big)} + \varepsilon. \label{dbound}
\end{align}
\end{corollary}

Recall that by the definition of the trace-distance, we have that
if the rightmost term of (\ref{dbound}) is negligible, i.e.  say
smaller than $2^{-\lambda n}$, then this situation is $2^{-\lambda
  n}$-close to the ideal situation where $F(X)$ is perfectly uniform
and independent of $F, U$ and $\regE$. In particular, replacing $F(X)$
by an independent and uniformly distributed bit results in a common
state which essentially cannot be distinguished from the original one.

\begin{proof}
In our case, the quantum register
$B$ from Theorem~\ref{thm:paoriginal} consists of a classical part $U$ and a
quantum part $E$. Denoting by $\sigma_E$ the fully mixed state on
the image of $\rho_E$, we only need to consider the term in the
exponent to derive Theorem~\ref{thm:pasmooth} as follows
\begin{align}
 \hminee(\rho_{XUE} | UE) \nonumber &\geq \hminee(\rho_{XUE} |
  \rho_{U} \otimes \sigma_E) \nonumber\\
&\geq \hminee(\rho_{XUE} | \rho_{U}) - \qhmax(\rho_{E}) \label{eq:chainrule}\\
&\geq \hminee(\rho_{XU} | \rho_U) - \qhmax(\rho_{E}) \label{eq:lemma}\\
&= \hiee{X | U} - q. \nonumber
\end{align}
The first inequality follows by Definition~\ref{def:qsmoothminentropy}
of $\hminee$ as supremum over all $\sigma_{UE}$. Inequality
\eqref{eq:chainrule} is the chain rule for smooth min-entropy
(Lemma~\ref{lem:qchainrule}). Inequality~\eqref{eq:lemma} uses
that the smooth min-entropy cannot decrease when dropping the quantum
register which is proven in Lemma~\ref{lem:dropquantum} from the last
section. The last step follows by assumption about the quantum
register and observing that the state $\rho_{XU}$
is classical and the quantum Definition~\ref{def:qsmoothminentropy}
therefore reduces to classical smooth min-entropy.
\end{proof}

The following corollary is a direct consequence of
Corollary~\ref{thm:pasmooth}. 
In Chapter~\ref{chap:qbc}, this lemma will be useful for proving the
binding condition of our commitment scheme. Recall that for $X \in
\nbit$, $\ball{\delta n}(X)$ denotes the set of all $n$-bit strings at
\index{Hamming distance}Hamming distance at most $\delta n$ from $X$ and $\ball{\delta n}
\assign |\ball{\delta n}(X)|$ is the number of such strings.
\begin{corollary}\label{cor:guess}
  Let $\rho_{X U \regE}$ be a ccq-state, where $X$ takes values in
  $\nbit$, $U$ in the finite domain $\cU$ and register $\regE$
  contains $q$ qubits.  Let $\hat{X}$ be a guess for $X$ obtained by
  learning $U$ and measuring $\regE$, and let $\eps \geq 0$. Then, for
  all $\delta< \frac{1}{2}$ it holds that
\[ \P\bigl[ \hat{X} \in \ball{\delta n}(X) \bigr] \leq 2^{-\frac{1}{2}
  (\hiee{X|U}-q-1) + \log(\ball{\delta n})} + 2\eps \cdot \ball{\delta n}.
\]
\end{corollary}
In other words, given some classical knowledge $U$ and a quantum
memory of $q$ qubits arbitrarily correlated with a classical random
variable $X$, the probability to find $\hat{X}$ at Hamming distance at
most $\delta n$ from $X$ where $nh(\delta)< \frac{1}{2}
(\hiee{X|U})-q)$ is small.
\begin{proof}
  Here is a strategy to try to bias $F(X)$ when given $\hat{X}$ and
  $F\in_R \chf{n}$: Sample $X' \in_R \ball{\delta n}(\hat{X})$ and
  output $F(X')$.  Note that, using $p_{\text{succ}}$ as a short hand
  for the probability $\P\big[ \hat{X} \in \ball{\delta n}(X) \big]$
  to be bounded,
\begin{align*}
\P\bigl[ F(X')=F(X) \bigr]
&=  \frac{p_{\text{succ}}}{\ball{\delta n}} +  
  \bigg(1- \frac{p_{\text{succ}}}{\ball{\delta n}} \bigg) \frac{1}{2} \\[1ex]
&= \frac{1}{2} + \frac{p_{\text{succ}}}{2 \cdot \ball{\delta n}},
\end{align*}
where the first equality follows from the fact that if \mbox{$X'\neq
  X$} then, as $\chf{n}$ is \univ, $\P\left[ F(X)=F(X')
\right]=\frac{1}{2}$.  Note that, given $F$ and $U$ and being allowed
to measure $\regE$, the probability of correctly guessing a binary
$F(X)$ is upper bounded by 
$\frac{1}{2}+\dist{ \rho_{F(X)F U \regE},\I\otimes\rho_{F U \regE} }$~\cite{FG99}.  In
combination with Corollary~\ref{thm:pasmooth} (with $\ell=1$) the above results in
\begin{equation*}
\frac{1}{2} + \frac{p_{\text{succ}}}{2 \cdot \ball{\delta n}} \leq \frac{1}{2}+ 
  \frac{1}{2} 2^{-\frac{1}{2}({\hiee{X|U}-q-1}) }+ \eps
\end{equation*}
and the claim follows by rearranging the terms. 
\end{proof}

\subsection{Classical Privacy Amplification}
The classical privacy-amplification theorem follows as special case
from the results above. When there is no quantum correlation, we
(almost) recover the well-known classical \emph{\index{left-over hash lemma}left-over hash lemma}
\cite{ILL89, BBCM95, HILL99}:
\begin{corollary}\label{thm:classicalpa}
  Let $X$ be a random variable over $\nbit$, and let $F$ denote the
  uniform choice of a hash function in a two-universal family of
  hash functions $\chf{n}$ mapping from $\nbit$ to $\set{0,1}^{\ell}$. Then
\[\dist{ P_{F(X) F},P_{\unif^{\ell}} P_{F} } \leq
  \frac12 2^{-\frac12(\H_2(X) - \ell)} \, .\]
\end{corollary}
This corollary (with collision- instead of min-entropy in the exponent
on the right-hand side) cannot immediately be derived from
Theorem~\ref{thm:paoriginal} above, but rather from its proof in
\cite{Renner05}. The reason for this is that the easiest way of
proving both Theorem~\ref{thm:paoriginal} and
Corollary~\ref{thm:classicalpa} is by directly considering collision
\index{entropy!collision}entropy instead of
min-\index{entropy!min-}entropy. On the other hand, relaxing the
notion of collision entropy to smooth min-entropy gives the natural
operative meaning (see Section~\ref{sec:defsmoothrenyientropy}) and
interestingly, it only looks like we are losing something by doing
that, but in fact this achieves optimality \cite{RW05}.

\index{privacy amplification|)}

\include{ndlf}

\include{uncert}
\include{rabinot}

\include{12ot}

\include{commit}

\include{rest}


\clearemptydoublepage
\markboth{\textsc{Notation}}{\textsc{Notation}}
\phantomsection
\addcontentsline{toc}{chapter}{Notation}
\chapter*{{\Huge Notation}}

\begin{tabular}{ll}
\multicolumn{2}{l}{\bf General}  \\ \hline
$\log$ & binary logarithm \index{log@$\log(\cdot)$} \\
$\ln$  & natural logarithm \index{ln@$\ln(\cdot)$} \\
$\naturals$ & natural numbers: $1,2,3,\ldots$ \\
$\reals$   & real numbers \\
$[a,b]$ & set of real numbers $r$ such that $a \leq r \leq b$ \index{interval}\\
$(a,b]$ & set of real numbers $r$ such that $a < r \leq b$\\
$x|_I$  & substring of $x$ consisting of bit positions in index set
$I$ \index{substring}\\
$x\pad_I$ & as above, padded with $0$s \\
$\ball{\delta n}(x)$ & set of $n$-bit strings with Hamming distance at
most $\delta n$ from $x$ \index{negl@$\negl{n}$}$\negl{n}$\\
$\negl{n}$ & any function in $n$ smaller than the inverse of any
polynomial\\
& for large enough $n$  \\
$[+, \times]_b$ & $+$ for $b=0$ and $\times$ for $b=1$\\
$\delta_{i,j}$ & \index{Kronecker delta}Kronecker delta\\ \hline
\\
\multicolumn{2}{l}{\bf Classical Information Theory} \\ \hline
$P_{X|Y}$ & conditional probability distribution of $X$ given $Y$ \\
$\E[R]$ & expected value of the real random variable $R$
\index{expected value}\\
$\delta(P,Q)$ & variational distance between distributions $P$ and
$Q$ \index{variational distance}\\
$P \approx_\eps Q$ & $P$ and $Q$ are at variational distance at most
$\eps$ \\
$\unif$ & independent and uniformly distributed binary random
variable\\
$\unif^\ell$ & $\ell$ copies of it\\
$\ev$ & \index{event $\ev$}event \\
$\id_{\ev}$ & \index{indicator random variable}indicator random variable of event $\ev$\\ 
$X \!\leftrightarrow\! Z \!\leftrightarrow\! Y$ & \index{Markov chain}Markov chain\\ \hline
\\
\multicolumn{2}{l}{\bf Quantum Information Theory} \\ \hline
$\cH_d$ & \index{Hilbert space}Hilbert space of dimension $d$\\
$\dens{\cH}$ & set of \index{density operator}density operators on $\cH$ \\
$\rho$ & density operator: normalized, Hermitian, non-negative\\
$\trace{\rho}$ & \index{trace}trace of $\rho$ \\
$\id$ & \index{$\I$ (fully mixed state)}fully mixed state\\
$\delta(\rho,\sigma)$ & \index{trace distance}trace distance between $\rho$ and $\sigma$ \\
$\ket{b}_{\theta}$ & classical bit $b$ encoded in basis $\theta$\\
$\rho_{XE}$ & \index{cq-state}cq-state \\ \hline
\end{tabular}

\begin{tabular}{ll}
\multicolumn{2}{l}{\bf Entropies} \\ \hline
$h(\cdot)$ & binary Shannon \index{entropy!Shannon}entropy function \\
$\pi_\alpha(X|Y)$ & \index{$\alpha$-order sum}$\alpha$-order sum of $X$ given $Y$ with joint
distribution $P_{XY}$\\
$\H_\alpha(X|Y)$ &  R\'enyi \index{entropy!classical R\'enyi}entropy of order $\alpha$ of
$X$ given $Y$\\
$\H_\infty(X|Y)$ &  \index{entropy!min-}min-entropy of $X$ given $Y$\\
$\H_2(X|Y)$ &  \index{entropy!collision}collision entropy of $X$ given $Y$\\
$\H(X|Y)$ &  \index{entropy!Shannon}Shannon entropy of $X$ given $Y$\\
$\H_0(X|Y)$ &  \index{entropy!max-}max-entropy of $X$ given $Y$\\
$\tH_\alpha(X|Y)$ & \index{entropy!average conditional R\'enyi}average conditional R\'enyi entropy of order
$\alpha$\\
$\H_\alpha^\eps(X|Y)$ &  $\eps$-smooth R\'enyi entropy of order
$\alpha$ of $X$ given $Y$\\
$\hiee{X|Y}$ &  \index{entropy!smooth min-}$\eps$-smooth min-entropy of $X$ given $Y$\\
$\hmaxee{X|Y}$ &  $\eps$-smooth max-entropy of $X$ given $Y$\\
\\
$\H_\alpha(\rho)$ & R\'enyi entropy of order $\alpha$ of the state
$\rho$\\
$\qhmin(\rho_{AB} | \sigma_B)$ & min-entropy of $\rho_{AB}$
relative to $\sigma_B$\\
$\qhmin(\rho_{AB} | B)$ & min-entropy of $\rho_{AB}$
given $\cH_B$\\
$\hminee(\rho_{AB} | \sigma_B)$ & $\eps$-smooth min-entropy of $\rho_{AB}$
relative to $\sigma_B$\\
$\hminee(\rho_{AB} | B)$ & $\eps$-smooth min-entropy of $\rho_{AB}$
given $\cH_B$\\ \hline

\end{tabular}


\clearemptydoublepage
\markboth{\textsc{Bibliography}}{\textsc{Bibliography}}
\phantomsection
\addcontentsline{toc}{chapter}{Bibliography}
\bibliographystyle{alpha} 
\bibliography{qip,crypto,procs} 

\index{linear function|see {non-degenerate linear function}}
\index{NDLF|see {non-degenerate linear function}}
\index{quantum uncertainty relation|see {uncertainty relation}}
\index{entropic uncertainty relation|see {uncertainty relation}}
\index{high-order entropic uncertainty relation|see {uncertainty relation}}
\index{computational basis|see {basis}}
\index{diagonal basis|see {basis}}
\index{basis!mutually unbiased|see {mutually unbiased bases}}
\index{entropic uncertainty bound|see {average entropic uncertainty
    bound}}
\index{uncertainty bound|see {average entropic uncertainty bound}}
\index{statistical distance|see {variational distance}}
\index{Kolmogorov distance|see {variational distance}}

\clearemptydoublepage
\phantomsection
\addcontentsline{toc}{chapter}{Index}
\markboth{\textsc{Index}}{\textsc{Index}}
\small
\printindex

\end{document}

%% file: frontpages.tex

\vspace*{\fill}
\begin{flushright}
  {\Huge\sf Cryptography in the Bounded-Quantum-Storage Model}\\[3ex]
  {\huge\sf Christian Schaffner} 
\end{flushright}
\noindent\rule{\linewidth}{1mm}\\[-.5ex]
\noindent\rule{\linewidth}{2.5mm}
\vfill
\begin{center}
  {\huge\sf PhD Dissertation}\\[\fill]
  \includegraphics{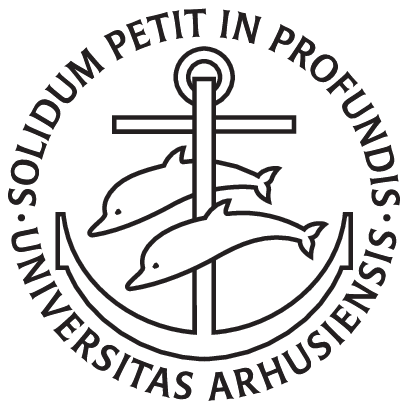}\\[\fill]
  {\sf BRICS Research School\\DAIMI -- Department of Computer Science\\University of Aarhus\\Denmark}
\end{center}
\vspace*{\fill}
\cleardoublepage
\begin{center}
  \vspace*{\stretch{1}}
  {\huge Cryptography in the\\[3mm]
Bounded-Quantum-Storage Model}\\[\fill]
  A Dissertation\\
  Presented to the Faculty of Science\\
  of the University of Aarhus\\
  in Partial Fulfillment of the Requirements for the\\
  PhD Degree\\[\stretch{2}]
  by\\
  Christian Schaffner\\
  official version submitted: March 2, 2007\\[5mm]

  final version: \makeatletter\@date\makeatother\\
\end{center}
\vspace*{\stretch{1}}


\clearemptydoublepage
\pagestyle{plain}
\phantomsection
\addcontentsline{toc}{chapter}{Abstract}
\chapter*{{\Huge Abstract}}

Cryptographic primitives such as oblivious transfer and bit commitment
are impossible to realize if unconditional security is required against
adversaries who are unbounded in running time and memory size.
Therefore, it is a great challenge to come up with restrictions on the
adversary's capabilities such that on one hand interesting
cryptographic primitives become possible, but on the other hand the
model is still realistic and as close to practice as possible.

The \emph{bounded-quantum-storage model} is a prime example of
such a cryptographic model. In this thesis, we initiate the study of
cryptographic primitives with unconditional security under the sole
assumption that the adversary's {\em quantum} memory is of bounded
size.

Oblivious transfer and bit commitment can be implemented in this model
using protocols where honest parties need no quantum memory, whereas
an adversarial player needs to store \emph{at least a large fraction} of the
total number of transmitted qubits in order to break the protocol.
This is in sharp contrast to the classical bounded-memory model, where
we can only tolerate adversaries with memory of size polynomially
larger than the honest players' memory size.

On the practical side, our protocols are efficient, non-interactive
and can be adapted to cope with various kinds of noise in the
transmission. In fact, they can be \emph{implemented using today's
technology}. 

On the theoretical side, new \emph{entropic uncertainty relations}
involving min-entropy are established and used to prove the security
of protocols in the bounded-quantum-storage model according to new
strong security definitions. The uncertainty relations lower bound
the min-entropy of the encoding used in most quantum-cryptographic
protocols and therefore contribute to the understanding of the quantum
effects which these protocols are based upon. The most direct way to
make use of these lower bounds is by assuming a quantum-memory bound on
the adversary. For instance, in the realistic setting of \emph{Quantum
  Key Distribution (\QKD)} against quantum-memory-bounded
eavesdroppers, the uncertainty relation allows to prove the security
of QKD protocols while tolerating considerably higher error rates
compared to the standard model with unbounded adversaries. 

\vspace{2mm} In addition, though not directly related to the
bounded-quantum-storage model, a classical result about
unconditionally secure 1-out-of-2 Oblivious Transfer (\OT) is
obtained. It is pointed out that the standard security requirement for
\OT\ of bits, namely that the receiver only learns one of the bits
sent, holds if and only if the receiver has no information on the XOR
of the two bits. This result generalizes to \OT\ of strings, in which
case the security can be characterized in terms of \emph{binary linear
  functions}.  More precisely, it is shown that the receiver learns
only one of the two strings sent, if and only if he has no information
on the result of applying any binary linear function which
non-trivially depends on both inputs to the two strings. This result
not only gives new insight into the nature of \OT, but it in
particular provides a \emph{powerful tool for analyzing \OT\ 
  protocols}. With this characterization at hand, the reducibility of
\OT\ of strings to a wide range of weaker primitives follows by a very
simple argument.


\clearemptydoublepage
\phantomsection
\addcontentsline{toc}{chapter}{Acknowledgements}
\chapter*{{\Huge Acknowledgements}}

I am grateful to everyone who helped and supported me during my PhD
studies here in {\AA}rhus. 

First of all, I want to cordially thank my supervisors and co-authors
Louis Salvail and Ivan Damg{\aa}rd and the whole cryptology group at
DAIMI for providing an excellent environment for cryptographic
research. Countless are the hours I have spent discussing scientific
as well as non-scientific issues with Louis, \emph{merci beaucoup}! I
thank my other co-authors Claude Cr\'epeau, Serge Fehr, Renato Renner,
George Savvides and J\"urg Wullschleger for many inspiring visits and
discussions.

I appreciated very much being a PhD student in a well-organized and
well-funded research group and to be able to work in a brand-new building
with plenty of space, great infrastructure and always helpful and
friendly staff and secretaries: Ellen, Hanne, Karen, Lene,
Michael, and Uffe.

Studying in {\AA}rhus has been a great experience mainly because of
all the friends from the constantly changing ``gang'' of foreign and
Danish fellows at DAIMI including Allan, Claudio, Claus, Doina, Gabi,
Henrik, Jan, Jesper, Jooyong, Johan, Kevin, Michael, Mikkel, Mirka,
Rune, Tord, Thomas M, Tomas, and Troels; but not to forget the ones
who have left Denmark and are now spread around the world:
Barnie, Christopher, Emanuela and Paolo, Fitzi, Gosia and Darek, Jens,
Jes\'us, Karl, Kirill, Marco, Nelly and Antonio, Philipp, Thomas P, and
Saurabh. I thank you all for the wonderful time, both at and off the
table-soccer table. Special thanks to
Gosia and Henrik for constructive comments on the introduction of this
thesis and to J\"urg and Serge for further comments.

I would also like to thank Claude Cr\'epeau for hosting me for a
fantastic summer half-year at McGill university in Montr\'eal where I
had the chance to meet many interesting people doing quantum research
and experience the exciting spot where the francophone part of North
America meets the anglophone rest of the continent.

I thank Prof.~Andreas Winter from the University of Bristol and
Prof.~Stefan Wolf from ETH Z\"urich as well as Prof.~Susanne
B{\o}dker from the University of Aarhus for agreeing to constitute the
evaluation committee for my PhD thesis.

Last but not least, I want to express my gratitude to my family for
their immense love and support from the distance. I am infinitely
grateful for the great childhood they gave me which was and still is an
invaluable source of self-confidence for me.

\vspace{8mm}
This research was partially supported by the EU Project SECOQC,
No:~FP6-2002-IST-1-506813.

\vspace{2ex}
\begin{flushright}
  \emph{Christian Schaffner,}\\
  \emph{{\AA}rhus, March 2, 2007.}
\end{flushright}


%% file: intro.tex
\clearemptydoublepage
\chapter{Introduction} \label{chap:intro}

In the quest for interesting \index{cryptographic model}cryptographic
models, bounding the quantum memory of adversarial players is a great
assumption.

\section{Cryptographic Models and Basic Primitives} \label{sec:cryptomodels}
It is a fascinating art to come up with
\emph{\index{protocol}protocols}\footnote{A protocol consists of
  clear-cut instructions for the participating players.} that achieve a
cryptographic task like encryption, authentication, identification,
voting, secure function evaluation to name just a famous few. To
define a notion of security for such protocols, one needs to specify a
\emph{\index{cryptographic model}cryptographic model}, i.e. an
environment in which the protocol is run. The model states for example
the number of honest and dishonest players, the allowed running time
and amount of memory available to honest and dishonest players, how
dishonest players are allowed to deviate from the protocol, the use of
external resources like (quantum) communication channels or other
already established cryptographic functionalities etc.

While coming up with more and more protocols for different models,
cryptographers realized that some basic \emph{primitives}
(i.e.~precisely defined cryptographic tasks) are useful as
``benchmarks'' of how powerful a particular cryptographic model is.
An example is the two-party primitive \emph{
\index{oblivious transfer}Oblivious Transfer} (\pOT).  It comes in different flavors,
but all of these variants are equivalent in the sense that anyone of
them can be implemented using (possibly several instances of) an
other.  The \emph{\index{oblivious transfer!one-out-of-two}one-out-of-two} variant \OT was originally
introduced by Wiesner around 1970 (but only published much later
in~\cite{Wiesner83}) in the very first paper about quantum
cryptography, and later rediscovered by Even, Goldreich, and Lempel~\cite{EGL82}. It lets a sender Alice transmit two bits to a
receiver Bob who can choose which of them to receive. A secure
implementation of \OT does not allow a dishonest sender to learn which
of the two bits was received and it does not allow a dishonest
receiver to learn any information about the second bit. It was a
surprising insight when Kilian showed that this simple primitive is \emph{complete} for
two-party cryptography \cite{Kilian88}. In other words, a model in
which \OT can be securely implemented allows to implement any
cryptographic functionality between two players\footnote{If the model
  can be reasonably extended to more players, this usually allows to
  implement secure multi-party protocols as well.}. Another variant we
are concerned with in this thesis was introduced by Rabin
\cite{Rabin81} and is hence called \index{oblivious transfer!Rabin}Rabin Oblivious Transfer
(\RabinOT). It is basically a ``secure erasure channel'': the sender
Alice sends a bit which with probability one half is absorbed and with
probability one half finds its way to the receiver Bob. The security
requirements are the following: whatever a dishonest Alice does, she cannot find
out whether the bit was received or not; and whatever a dishonest
receiver does, he does not get any information about the bit with
probability one half.

Yet another basic two-party primitive of interest is \index{bit
  commitment} Bit Commitment (\BC) which allows a player to commit
himself to a choice of a bit $b$ by communicating with a verifier. The
verifier should not learn $b$ (we say the commitment is
\emph{\index{bit commitment!hiding}hiding}), yet the committer can
later choose to reveal $b$ in a convincing way, i.e. only the value
fixed at commitment time will be accepted by the verifier (we say the
commitment is \emph{\index{bit commitment!binding}binding}).  Bit
Commitment is a fundamental building block of virtually every more
complicated cryptographic protocol. Implementing secure \BC with a
secure \OT at hand is not difficult\footnote{To commit to a bit $b$,
  the committer sends random bits of parity $b$ via (several instances
  of) \OT and the verifier picks randomly one of the bits. To open,
  the committer sends all the random bits he was using, the verifier
  checks whether these are consistent with what he received.}. On the
other hand, there are cryptographic models allowing to securely
implement \BC, but not \OT. Moran and Naor gave an example of such a
model by assuming the physical device of a tamper-proof seal~\cite{MN05}.

It is not hard to see that the two security requirements for \BC are
in a sense contradictory, so perfectly secure bit commitment cannot be
implemented ``from scratch'', that is if only error-free communication
is available and there is no limitation assumed on the computing power
and memory of the players. The informal reason for this is that the
hiding property implies that when 0 is committed to, exactly the same
information exchange could have happened when committing to 1.
Hence, even if 0 was actually committed to, the committer could always
compute a complete view of the protocol consistent with having
committed to 1, and pretend that this view was what he had in mind
originally. By the reduction of \BC to \OT follows that also \OT and
many other cryptographic functionalities cannot be perfectly secure
when built from scratch.

One might hope that allowing the protocol to make use of quantum
communication would make a difference. Here, information is stored in
qubits, i.e., in the state of two-level quantum mechanical systems,
such as the polarization state of a single photon. Quantum information
behaves in a way that is fundamentally different from classical
information, enabling, for instance, unconditionally secure key
exchange between two honest players (so-called
\emph{\index{quantum key distribution}Quantum Key Distribution}).
However, in the case of two mutually distrusting parties, we are not
so fortunate: even with quantum communication, unconditionally secure
\BC and \OT remain \index{impossibility!of quantum bit
  commitment}impossible. This is the infamous
impossibility result by Mayers and by Lo and Chau~\cite{Mayers97,LC97}.

For this reason, cryptographers have tried hard to exhibit more
restricted models where these impossibility results do not apply. The
high art in this process is to find assumptions that are as realistic
as possible -- thus only minimally restricting the model, but still
strong enough to allow for implementing interesting functionalities.
There are at least three kinds of possible assumptions, namely
\begin{itemize}
\item bounding the computing power of players,
\item using the noise in the communication channel,
\item exploiting some physical limitation of the adversary, e.g., if
  the size of the available memory is bounded.
\end{itemize}

The first scenario is the basis of many well known solutions based on
plausible but unproven complexity assumptions, such as hardness of
factoring or discrete logarithms. A term often used for such schemes
is ``\index{computational security}computational security'', meaning
that it is \emph{not impossible} for an adversary to behave
dishonestly, but it is \emph{computationally infeasible} for him to do
so. Security proofs are usually done by reduction in the sense that
breaking the security of the protocol would imply solving a hard
problem like factoring the product of two large prime numbers. The
second scenario has been used to construct both \BC and \pOT protocols
in various models for the noise by Cr\'epeau, Kilian, Damg{\aa}rd,
Salvail, Fehr, Morozov, Wolf, and Wullschleger
\cite{CK88,DKS99,DFMS04,CMW04,Wullschleger07}.

The third scenario is the focus of this thesis. In contrast to the
first scenario, we deal with ``\index{unconditional security}unconditional security'' where (depending on the task a
protocol aims to achieve) an adversary has no way whatsoever to gain
illegal information. Proofs are not done by reduction, but we can
prove in information-theoretic terms that except with negligible
probability, the adversary does not learn \emph{any information} that
is meant to remain secret.

\section{Classical Bounded-Storage Model} \label{sec:ClassicalBSMIntro}
In the \index{classical bounded-storage model}classical
bounded-storage model, we assume the players to use classical
error-free communication and to be computationally unbounded, but on the
other hand restrict the size of their memory. In the usual setting,
there is a large random source $R$ (often called the
\emph{\index{randomizer}randomizer}) which all players can access, but
which is too large (or transmitted too quickly) to store as a whole.
One can think of $R$ as a deep-space radio source or a satellite broadcasting
random bits at a very high rate.

When Maurer introduced the classical bounded-storage model
in~\cite{Maurer90}, the goal was \emph{secure message transmission}.
He showed that two honest parties Alice and Bob sharing an initial key
can \index{key expansion}expand that key unless the eavesdropper Eve
can store more than a large fraction of the randomizer.  The basic
idea of the technique allowing Alice and Bob to get an advantage
over Eve is that their initial secret key indexes some positions in
the randomizer about which Eve has some uncertainty if she cannot
store the whole randomizer. Therefore, the bits at these positions can
be combined to yield more secure key bits and so to expand the initial
key.

A line of subsequent work by Maurer, Cachin, Aumann, Ding, Rabin,
Dziembowski, Lu, and Vadhan \cite{Maurer92, CM97, ADR02, DM04, Lu04,
  Vadhan04} improved this original protocol in terms of efficiency and
security. Aumann, Ding and Rabin~\cite{ADR02} noticed that
protocols in this model enjoy the property of ``\index{everlasting
security}everlasting security'' in the sense that the newly
generated key remains secure even when the initial key is later
revealed and Eve is no longer memory-bounded, under the sole condition
that the original randomizer cannot be accessed any
more. Ding~\cite{Ding05} showed how to do 
\index{error correction!classical bounded-storage model}error correction in the
bounded-storage model and therefore how to cope with the situation
when the honest parties do not have exactly the same view on the
randomizer.

Cachin, Cr\'epeau and Marcil illustrated the power of the
bounded-storage model by exhibiting in~\cite{CCM98} a protocol for
\OT. Ding improved on this \cite{Ding01} and later showed a
constant-round protocol for oblivious transfer in joint work with
Harnik, Rosen and Shaltiel \cite{DHRS04}.

All these protocols are shown secure as long as the adversary's memory
size is at most quadratic in the memory size of the honest players.
Considering the ease and low cost of storing massive amounts of classical
data nowadays, it is questionable how practical such an assumption on the
memory size of the players is. It would be clearly more satisfactory
to have a larger than quadratic separation between the memory size of
honest players and that of the adversary. However, this was shown to
be impossible by Dziembowski and Maurer~\cite{DM04}.

\section{Contributions} \label{sec:contributions}
In this section, we give an overview of the contributions of this
thesis.
The results about classical oblivious transfer described in
Chapter~\ref{chap:ClassicalOT} and summarized in
Section~\ref{sec:ClassicalOTReductions} are joint
work with Damg{\aa}rd, Fehr and Salvail~\cite{DFSS06}. 
All other results are based on two papers co-authored with Damg{\aa}rd, Fehr,
Salvail and Renner: \cite{DFSS05} and \cite{DFRSS07}. A journal version of \cite{DFSS05}
is to appear in a special issue of the SIAM Journal of Computing
\cite{DFSS08journal}.

\subsection{Bounded-Quantum-Storage Model}
\index{bounded-quantum-storage model}In this thesis, we study for the
first time protocols where quantum communication is used and we place
a bound on the adversary's {\em quantum} memory size.  There are two
reasons why this may be a good idea: first, if we do not bound the
classical memory size, we avoid the impossibility result of
\cite{DM04}.  Second, the adversary's typical goal is to obtain a
certain piece of classical information that we want to keep hidden
from him. However, if he cannot store all the quantum information that
is sent, he must convert some of it to classical information by
measuring. This may irreversibly destroy information, and we may be
able to arrange it in such a way that the adversary cannot afford to lose
information this way, while honest players can.

It turns out that this can be achieved indeed: we present protocols for
both \BC and \pOT in which $n$ qubits are transmitted, where honest
players need {\em no quantum memory}, but where the adversary must
store at least a large fraction (typically $n/2$ or $n/4$) of the $n$
transmitted qubits to break the protocol. We emphasize that no bound
is assumed on the adversary's computing power, nor on his classical
memory. This is clearly much more satisfactory than the classical
case, not only from a theoretical point of view, but also in practice:
while sending qubits and measuring them immediately as they arrive is
well within reach of current technology, storing even a single qubit
for more than a fraction of a second is a formidable technological
challenge.

Furthermore, we show that our protocols also work in a non-ideal
setting where we allow the quantum source to be imperfect and the
quantum communication to be noisy. We emphasize that what makes \pOT and
\BC possible in our model is not so much the memory bound per se, but
rather the loss of information on the part of the adversary. Indeed,
our results also hold if the adversary's memory device holds an
arbitrary number of qubits, but is imperfect in certain ways.

All these factors make the assumption of \index{bounded-quantum-storage
model}bounded quantum memory a very attractive cryptographic model.
On one hand, as for the \index{classical bounded-storage
model}classical bounded-storage model, it is simple to work with and
yields beautiful theoretical results. On the other hand, it is much
more reasonable to assume the difficulty of storing quantum
information compared to storing classical one and hence, we are very
close to the physical reality and get schemes that can actually be
implemented!


\subsection{Characterization of Security of Classical \OT} \label{sec:ClassicalOTReductions}
While the task of formally defining \index{unconditional security}unconditional security of classical protocols for \RabinOT
and \BC is well understood, capturing the security of \OT in
information-theoretic terms is considerably more delicate, as was
pointed out by Cr\'epeau, Savvides, Schaffner and
Wullschleger~\cite{CSSW06}. 
For \OT of bits, it is clear that the
security for a honest sender against a cheating receiver guarantees
that the receiver does not learn any information about the XOR of the
two bits. Somewhat surprisingly, the converse is true as well, not
having any information about the XOR of the two bits sent implies that
we can point at one bit which the dishonest receiver does not know
(given the other).

This idea can be generalized to \OT of strings where the ignorance of
the XOR becomes ignorance of the outcome of all \index{non-degenerate
linear function}Non-Degenerate Linear binary Functions (NDLFs)
applied to the two strings sent. Such a characterization of
\index{sender-security!characterization of}sender-security in terms of
NDLF composes well with \emph{strongly \index{two-universal
    hashing!strongly}two-universal hashing} and hereby yields a
powerful technique to improve the analyses of the standard \index{reduction}reductions
from \OT to weaker variants of \pOT.

As a historical side note, the original motivation for this classical
characterization was the hope that it translates to the quantum
setting and thereby yields a security proof of the \OT scheme in the
\index{bounded-quantum-storage model}bounded-quantum-storage model. We
will point out why this approach does \emph{not} work.

\subsection{Quantum Security Definitions and Protocols}
When the players are allowed to use quantum communication, the output
of a dishonest player is a quantum state even when the protocol
implements a classical primitive. Therefore, security definitions for
\RabinOT, \OT and \BC have to be phrased in quantum terms. As an
easy-to-use \index{composability}composability framework has not yet
been established for quantum protocols\footnote{Some rather
complicated frameworks are known. They have been put forward by Ben-Or and
Mayers \cite{BM04} and Unruh \cite{Unruh02}.}, various \emph{ad-hoc}
security requirements are commonly used. The definitions in this
thesis are the strongest so far proposed, and as they are based on the
(classical) considerations in \cite{CSSW06}, we believe that they are
best suited to provide \emph{sequential composability}.

Most of the presented protocols in the bounded-quantum-storage model
can be cast in a non-interactive form, i.e.~only one party sends
information when doing \pOT, commitment or opening. We show the following.

\medskip
\noindent
{\bf {\em \pOT in the Bounded-Quantum-Storage Model:}} {\em There
  exist non-interactive protocols for \RabinOT and 1-out-of-2
  Oblivious Transfer (\OT[2]) of $\ell$-bit messages, secure in the
  bounded-quantum-storage model against adversaries with
  quantum-memory size at most $n/2- \ell$ for \RabinOT and $n/4 -
  2\ell$ for \OT. Here, $n$ is the number of qubits transmitted in the
  protocol and $\ell$ can be a constant fraction of $n$. Honest
  players need no quantum memory at all.}  \medskip

For the case of bit commitment, the standard definition of the
\index{bit commitment!binding}binding property used in the quantum
setting was introduced by Dumais, Mayers and Salvail~\cite{DMS00}. For
$b \in \set{0,1}$, let $p_b$ denote the probability that a dishonest
committer successfully opens the commitment to value $b$. The binding
condition then requires that the sum of $p_0$ and $p_1$ does
essentially not exceed 1. More formally, $p_0 + p_1 \leq 1 + \negl{n}$
where $\negl{n}$ stands for a term which is negligible in $n$ such as
$2^{-cn}$ (for a constant $c>0$) which is exponentially small in $n$.
This is to capture that a quantum committer can always commit to the
values $0$ and $1$ in superposition. We call this notion \emph{weakly
  binding} in the following. A shortcoming of this notion is that
committing bit by bit is not guaranteed to yield a secure string
commitment---the argument that one is tempted to use requires
independence of the $p_{b}$'s between the different executions, which
in general does not hold.


Instead, we propose the following \emph{strong binding} condition:
After the commitment phase, there exists a binary random variable $D
\in \set{0,1}$ such that a dishonest committer cannot open the
commitment to value $D$ except with negligible probability. The point
is that the distribution of $D$ is not under control of the dishonest
committer. We will point out that using this definition, we can easily
derive the security of a string commitment from the security of the
individual bits.

\medskip
\noindent
{\bf {\em \BC in the Bounded-Quantum-Storage Model:}} {\em There
  exists a protocol for bit commitment which is non-interactive.
 It is perfectly hiding and weakly binding in the
  bounded-quantum-storage model against dishonest committers with
  quantum-memory size at most $n/2$. It is strongly binding against
  memory sizes of at most $n/4$. Here, $n$ is the number of qubits
  transmitted in the protocol. Honest players need no quantum memory
  at all.}  \medskip

Furthermore, the commitment protocol has the interesting property that
the only message is sent \emph{to} the committer, i.e., it is possible
to commit while only {\em receiving} information.  Such a scheme
clearly does not exist without a bound on the committer's memory, even
under computational assumptions and using quantum communication: a
corrupt committer could always store (possibly quantumly) all the
information sent, until opening time, and only then follow the honest
committer's algorithm to figure out what should be sent to
convincingly open a 0 or a~1.  

Note that in the \index{classical bounded-storage model}classical
bounded-storage model, it has been shown by Moran, Shaltiel and
Ta-Shma~\cite{MST04} how to do \index{time-stamping}time-stamping
that is non-interactive in our sense: a player can time-stamp a
document while only receiving information.  However, no reasonable
protocol for \BC or for time-stamping a single bit exists in this
model.  It is straightforward to see that any such protocol can be
broken by an adversary with classical memory of size twice that of an
honest player, while our protocol requires no quantum memory for the
honest players and remains secure against any adversary unable to
store more than half the size of the quantum transmission.

We also note that it has been shown earlier by Salvail
\cite{Salvail98} that \BC is possible using quantum communication,
assuming a different type of physical limitation, namely a bound on
the size of coherent measurement that can be implemented. This
limitation is incomparable to ours: it does not limit the total size
of the memory, instead it limits the number of bits that can be
simultaneously operated on to produce a classical result. Our
adversary has a limit on the total quantum memory size, but can
measure all of it coherently. The protocol from \cite{Salvail98} is
interactive, and requires a bound on the maximal measurement size that
is sub-linear in $n$.


\subsection{Quantum Uncertainty Relations}
\index{uncertainty relation}A problem often encountered in
\index{quantum cryptography}quantum cryptography is the following:
through some interaction between the players, a quantum state is
generated and then measured by one of the players (we call her Alice
in the following). Assuming Alice is honest, we want to know how
unpredictable her measurement outcome is to the adversary.  Once a
lower bound on the adversary's uncertainty about Alice's measurement
outcome is established, it is usually easy to prove the desired
security property of the protocol. Many existing constructions in
quantum cryptography have been proven secure following this paradigm.

Typically, Alice does not make her measurement in a fixed basis, but
chooses at random from a set of different bases. These bases are
usually chosen to be pairwise {\em \index{mutually unbiased
bases}mutually unbiased}, meaning that if the quantum state is
such that the measurement outcome in one basis is fixed, then this
implies that the uncertainty about the outcome of the measurement in
the other basis is maximal. In this way, one hopes to keep the
adversary's uncertainty high, even if the state is (partially) under
the adversary's control.

An inequality that lower bounds the adversary's uncertainty in such a
scenario is called an {\em uncertainty relation}.  There exist
uncertainty relations for different measures of uncertainty but
cryptographic applications typically require the adversary's
\index{entropy!min-}min-entropy to be bounded from below. Such uncertainty relations are
the key ingredient in the security proofs of our protocols in the
bounded-quantum-storage model.

In this thesis, we introduce new general and tight high-order entropic
uncertainty relations. Since the relations are expressed in terms of
lower bounds on the min-entropy or upper-bounds on large probabilities
respectively, they are applicable to a large class of natural
protocols in quantum cryptography.

The first uncertainty relation is concerned with the situation where a
\smash{$n$-qubit} state $\rho$ is measured in one out of two mutually
unbiased bases, say either in the
\index{basis!computational}computational basis (the $+$-basis) or in
the \index{basis!diagonal}diagonal basis (the $\times$-basis).

\medskip
\noindent
{\bf {\em First Uncertainty Relation:}} {\em Let $\rho$ be an
  arbitrary state of $n$ qubits, and let $\Qp(\cdot)$ and $\Qt(\cdot)$
  be the respective probability distributions over $\nbit$ of the
  outcome when $\rho$ is measured in the $+$-basis respectively the
  $\times$-basis.  Then, for any two sets $L^+ \subset \set{0,1}^n$
  and $L^{\times} \subset \set{0,1}^n$ it holds that
\[ \Qp(L^+)+\Qt(L^{\times}) \leq 1 + 2^{-n/2} \sqrt{|L^+|
  |L^{\times}|}. \] } 

Another \index{uncertainty relation}uncertainty relation is derived
for situations where an $n$-qubit state $\rho$ has each of its qubits
measured in a random and independent basis sampled uniformly from a
fixed set ${\cal B}$ of bases.  ${\cal B}$ does not necessarily have
to be \index{mutually unbiased bases}mutually unbiased, but we assume
a lower bound $h$---the so-called {\em \index{average entropic
uncertainty bound}average entropic uncertainty bound}---on the
average \index{entropy!Shannon}Shannon entropy of the distribution $P_{\vartheta}$, obtained
by measuring an arbitrary one-qubit state in basis $\vartheta \in
{\cal B}$, meaning that $\frac{1}{|{\cal B}|}\sum_{\vartheta}
\H(P_{\vartheta}) \geq h$.

\medskip
\noindent
{\bf {\em Second Uncertainty Relation (informal):}} {\em Let $\cal B$ be a
  set of bases with an average entropic uncertainty bound $h$ as
  above.  Let $P_{\theta}$ denote the probability distribution defined
  by measuring an arbitrary $n$-qubit state $\rho$ in basis $\theta
  \in {\cal B}^n$. For a uniform choice $\Theta \in_R {\cal B}^n$,
  it holds except with negligible probability (over $\Theta$ and over
  $P_{\theta}$) that 
\begin{equation}\label{main}
\hmin(P_\theta \mid \Theta=\theta) \gtrsim n h.
\end{equation}
} \medskip 

Observe that (\ref{main}) cannot be improved significantly since the
min-entropy of a distribution is at most equal to the Shannon entropy.
Our uncertainty relation is therefore asymptotically tight when the
bound $h$ is tight.

Any lower bound on the Shannon entropy associated to a set of
measurements ${\cal B}$ can be used in (\ref{main}).  In the special
case where the set of bases is ${\cal B}=\{+,\times\}$ (i.e. the two
BB84 bases named after Bennett and Brassard who used them in the first
quantum-key-distribution protocol~\cite{BB84}), 
$h$ is known precisely using Maassen and Uffink's
entropic relation~\cite{MU88}, see~(\ref{eq:maassenuffink}).  We
get $h=\frac{1}{2}$ and (\ref{main}) results in $\hmin(P_{\theta} \mid
\Theta=\theta) \gtrsim \frac{n}{2}$.
Uncertainty relations for the \index{BB84 coding scheme}BB84 coding
scheme are useful, since this coding is widely used in
quantum cryptography.  Its resilience to imperfect quantum channels,
sources, and detectors is an important advantage in practice.

A major difference between the first and second uncertainty relation
is that while both relations can be used to bound the min-entropy
conditioned on an event, this event happens in the latter case with
probability essentially 1 (on average) whereas the corresponding event
from the first relation (defined in Corollary~\ref{cor:hadamard}) only
happens with probability about $1/2$.

\subsection{\QKD against Quantum-Memory-Bounded Eavesdropper}
We illustrate the versatility of our second uncertainty relation by
applying it to Quantum-Key-Distribution (\QKD) settings.  \QKD is the
art of distributing a secret key between two distant parties, Alice
and Bob, using only a completely insecure quantum channel and
authentic classical communication. \QKD protocols typically provide
\index{unconditional security}unconditional security, i.e., even an
adversary with unlimited resources cannot get any information about
the key.  A major difficulty when implementing \QKD schemes is that
they require a low-noise quantum channel.  The tolerated noise level
depends on the actual protocol and on the desired security of the key.
Because the quality of the channel typically decreases with its
length, the maximum tolerated noise level is an important parameter
limiting the maximum distance between Alice and Bob.

We consider a model in which the adversary has a limited amount of
quantum memory to store the information she intercepts during the
protocol execution. In this model, we show that the maximum
tolerated noise level is larger than in the standard scenario where
the adversary has unlimited resources.  
For {\em one-way \QKD protocols} which are protocols where error-correction is
performed non-interactively (i.e., a single classical message is sent
from one party to the other), we show the following result:

\medskip
\noindent
{\bf {\em \QKD Against Quantum-Memory-Bounded Eavesdroppers:}} {\em Let
  $\cB$ be a set of orthonormal bases of the two-dimensional Hilbert space $\cH_2$ with average entropic
  uncertainty bound $h$. Then, a \emph{one-way \QKD-protocol} produces
  a secure key against eavesdroppers whose quantum-memory size is
  sublinear in the length of the raw key at a positive rate, as long as
  the bit-flip probability $p$ of the quantum channel fulfills $h(p)
  < h $ where $h(\cdot)$ denotes the binary Shannon-entropy
  function.  } \medskip

Although this result does not allow us to improve (compared to
unbounded adversaries) the maximum error-rate for the BB84 protocol
(the \index{protocol!4-state}4-state protocol), the
\index{protocol!6-state}6-state (using three mutually unbiased bases)
protocol can be shown secure against adversaries with memory bound
sublinear in the secret-key length as long as the bit-flip error-rate
is less than $17\%$. This improves over the maximal error-rate of
$13\%$ for this protocol against unbounded adversaries.  We also show
that the generalization of the 6-state protocol to more bases (not
necessarily mutually unbiased) can be shown secure for a maximal
error-rate up to $20\%$ provided the number of bases is large enough.
Note that the best known one-way protocol based on qubits is proven
secure against general attacks for an error-rate of only up to roughly
$14.1\%$, and the theoretical maximum is $16.3\%$~\cite{RGK05}.

The quantum-memory-bounded eavesdropper model studied here is not
comparable to other restrictions on adversaries considered in the
literature (e.g. \emph{individual attacks}, where the eavesdropper is
assumed to apply independent measurements to each qubit sent over the
quantum channel as considered by Fuchs, Gisin, Griffiths, Niu, Peres,
and L\"utkenhaus~\cite{FGGNP97,Lutkenhaus00}).  In fact, these
assumptions are generally artificial and their purpose is to simplify
security proofs rather than to relax the conditions on the quality of
the communication channel from which secure key can be generated.  We
believe that the quantum-memory-bounded eavesdropper model is more
realistic.


\section{Outline of the Thesis}
In Chapter~\ref{chap:prelim}, we introduce notation and present some
basic concepts from probability and quantum information theory like
quantum states and various kinds of their entropies. We prepare the
stage by reproducing and slightly extending the results about privacy
amplification via two-universal hashing from Renner's PhD thesis \cite{Renner05}.

Chapter~\ref{chap:ClassicalOT} is the only (almost) exclusively
classical chapter. It introduces the different flavors of oblivious
transfer and gives a characterization of the security for the sender
of \OT in terms of non-degenerate linear functions. It is cast in a
stand-alone manner and the rest of the thesis can be understood
without reading this chapter.

In Chapter~\ref{chap:uncertrelations}, the basis for the security
proofs of the following chapters is laid by establishing the quantum
min-entropic uncertainty relations. The following
Chapters~\ref{chap:RabinOT} and \ref{chap:12OT} contain the quantum
definitions, protocols and security proofs for \RabinOT and \OT,
respectively. Chapter~\ref{chap:qbc} treats quantum bit commitment.
Two flavors of the ``binding property'' are defined and the techniques
from the two previous chapters are used to prove security in the
bounded-quantum-storage model.

Chapter~\ref{chap:qkd} is devoted to another application of the
(second) uncertainty relation, quantum key distribution against a
quantum-memory-bounded eavesdropper. The last
Chapter~\ref{chap:conclusions} addresses some practical issues in
greater detail and concludes.

A short summary of the notation, the bibliography and an index
can be found at the end of the thesis.


\section{Related Work}
The classical bounded-storage model is described in
Section~\ref{sec:ClassicalBSMIntro}. Besides work pointed out
in the overview of the contributions in
Section~\ref{sec:contributions} above, it is worth mentioning that several protocols aiming at achieving quantum oblivious transfer have been proposed. After Wiesner's original conjugate-coding protocol~\cite{Wiesner83}, Bennett, Brassard, Cr\'epeau, and Skubiszewska proposed an interactive protocol for \OT~\cite{BBCS91}, whose security was subsequently analyzed by Cr\'epeau~\cite{Crepeau94}, Mayers, Salvail~\cite{MS94, Mayers95}, and Yao~\cite{Yao95}. The protocol from \cite{BBCS91} is interactive and can be easily broken by a dishonest receiver with unbounded quantum memory. To ensure that the receiver actually performs a measurement, it was suggested to use (quantum) bit-commitment schemes such as~\cite{BCJL93} which were believed to be secure against such adversaries at this point in time. After the impossibility proofs of quantum bit-commitment by Lo and Chau~\cite{LC97}, and Mayers~\cite{Mayers97}, and of oblivious transfer by Lo~\cite{Lo97}, it became clear that assumptions are necessary in order to securely realize these primitives. Compared to these previous attempts, the protocols in this thesis are simpler, non-interactive, and provably secure according to stronger security definitions.

Work related to classical OT-reductions is referred to in the introductory sections to
Chapter~\ref{chap:ClassicalOT} in Sections~\ref{sec:introtoNDLF}
and~~\ref{sec:reductions}.  Previous work about quantum uncertainty
relations is described in Section~\ref{sec:uncerthistory}.

%% file: ndlf.tex
\chapter{Classical Oblivious Transfer}  \label{chap:ClassicalOT}
Most of the results presented in this chapter are published in \cite{DFSS06}.

\index{sender-security!of classical \OT|(}

\section{Introduction and Outline} \label{sec:introtoNDLF}
As already mentioned in Section~\ref{sec:cryptomodels}, 1-out-of-2
Oblivious-Transfer, \OT for short, is a two-party primitive which
allows a sender to send two bits (or, more generally, strings) $B_0$
and $B_1$ to a receiver, who is allowed to learn one of the two
according his choice $C$. Informally, it is required that the receiver
only learns $B_C$ but not $B_{1-C}$ (what we call security for the
honest sender, hence {\em sender-security}), while at the same time
the sender does not learn $C$ ({\em \index{receiver-security!of \OT}receiver-security}).
Interestingly, \OT was introduced by Wiesner around 1970 (but only
published much later~\cite{Wiesner83}) under the name of
``multiplexing'' in the context of quantum cryptography, and, inspired
by~\cite{Rabin81} where a different flavor was introduced, later
re-discovered by Even, Goldreich and Lempel~\cite{EGL82}.

\OT turned out to be very powerful as Kilian~\cite{Kilian88} showed it
to be sufficient for secure general two-party computation.  For this
reason, much effort has been put into reducing \OT to seemingly weaker
flavors of \pOT, like \RabinOT, \XOT,
etc.~\cite{Crepeau87,BC97,Cachin98,Wolf00,BCW03,CS06}.

In this chapter, we focus on a slightly modified notion of \OT, which we
call {\em Randomized} \OT, \RandOT for short, where the bits (or
strings) $B_0$ and $B_1$ are not {\em in}put by the sender, but
generated uniformly at random during the \RandOT\ and then {\em
  out}put to the sender. It is still required that the receiver only
learns the bit (or string) of his choice, $B_C$, whereas the sender
does not learn any information on $C$. It is obvious that a \RandOT\ 
can easily be turned into an ordinary \OT simply by using the
generated $B_0$ and $B_1$ to mask the actual input bits (or strings).
Furthermore, all known constructions of unconditionally secure \OT 
protocols make implicitly the detour via \RandOT.

In a first step, we observe that the sender-security condition of a
\RandOT\ of {\em bits} is equivalent to requiring the XOR $B_0 \oplus
B_1$ to be close to uniformly distributed from the receiver's point of
view. The proof is very simple, and it is kind of surprising that---to
the best of our knowledge---this has not been realized before. We then
ask and answer the question whether there is a natural generalization
of this result to \RandOT\ of {\em strings}. Note that requiring the
bit wise XOR of the two strings to be uniformly distributed is
obviously not sufficient. We show that the sender-security for \RandOT\ of
strings can be characterized in terms of {\em non-degenerate linear
  functions} (bivariate binary linear functions which non-trivially
depend on both arguments, as defined in Definition~\ref{def:linear}):
sender-security holds if and only if the result of applying any
non-degenerate linear function to the two strings is (close to)
uniformly distributed from the receiver's point of view.

We then show the usefulness of this new understanding of \OT. We
demonstrate this on the problem of reducing \OT to weaker primitives.
Concretely, we show that the reducibility of an ordinary \OT to weaker
flavors via a non-interactive reduction follows by a trivial argument
from our characterization of sender-security.  This is in sharp
contrast to the current literature: The proofs given by Brassard,
Cr\'epeau and Wolf~\cite{BC97,Wolf00,BCW03} for reducing \OT to
\XOT, \GOT\ and \BUOT\ (we refer to Section~\ref{sec:application} for
a description of these flavors of \pOT) are rather complicated and
tailored to a particular class of privacy-amplifying hash functions;
whether the reductions also work for a less restricted class is left
as an open problem~\cite[page~222]{BCW03}. And, the proof given by
Cachin~\cite{Cachin98} for reducing \OT to one execution of a
general \pUOT\ is not only complicated, but also incorrect, as we will
point out.  Thus, our characterization of the condition for
sender-security allows to simplify existing reducibility proofs and,
along the way, to solve the open problem posed in~\cite{BCW03}, as
well as to improve the reduction parameters in most cases, but it also
allows for new, respectively until now only incorrectly proven
reductions. In recent work by Wullschleger \cite{Wullschleger07}, the
analysis of these reductions is further improved.

Furthermore, we extend our result and show how our characterization of
\RandOT\ in terms of non-degenerate linear functions translates
to \onenOT. 

\vspace{2mm}
As historical side note, we note that the original motivation for
characterizing sender-security with the help of NDLFs was to prove
sender-security of the quantum protocol for \OT described in
Chapter~\ref{chap:12OT}. We point out by an example in
Section~\ref{sec:quantumdoesnotwork} at the end of this
chapter why this approach does not work.

\section{Defining \OT}\label{sec:Definition}

\subsection{Randomized \OT of Bits}

Formally capturing the intuitive understanding of the security of \OT 
is a non-trivial and subtle task.  For instance requiring the sender's
view to be independent of the receiver's choice bit $C$ is too strong
a requirement, since his input might already depend on $C$. The best
one can hope for is that his view is independent of $C$ {\em
  conditioned on his input} $B_0,B_1$. Security against a dishonest
receiver is even more subtle.
We refer to the security definition by Cr\'epeau, Savvides, Schaffner and
Wullschleger of~\cite{CSSW06}, where it is argued that this definition
is the ``right" way to define unconditionally secure \OT. In their
model, a secure \OT protocol is as good as 
an ideal \OT functionality.

In this thesis, we will mainly focus on a slight modification of \OT,
which we call {\em Randomized \OT} (although {\em sender-}randomized
\OT would be a more appropriate, but also rather lengthy name). A
Randomized \OT, or \RandOT\ for short, essentially coincides with an
ordinary \OT, except that the two bits $B_0$ and $B_1$ are not {\em
  input} by the sender but generated uniformly at random during the
protocol and {\em output} to the sender. This is formalized in
Definition~\ref{def:RandOT} below.

There are two main justifications for focusing on \RandOT. First, an
ordinary \OT can easily be constructed from a \RandOT: the sender can
use the randomly generated $B_0$ and $B_1$ to one-time-pad encrypt his
input bits for the \OT, and send the masked bits to the receiver (as
first realized by Beaver~\cite{Beaver95}). For a formal proof of this we
refer to the full version of~\cite{CSSW06}. 
And second, all information-theoretically secure constructions of \OT 
protocols we are aware of in fact do implicitly build a \RandOT and
use the above reduction to achieve \OT.

We formalize \RandOT in such a way that it 
minimizes and simplifies as much as possible the security restraints,
while at the same time remaining sufficient for \OT.

\begin{definition}[\RandOT]\label{def:RandOT}
  An $\varepsilon$-secure {\em \RandOT} is a protocol between sender
  $\S$ and receiver $\R$, with $\R$ having input $C \in \{0,1\}$
  (while $\S$ has no input), such that for any distribution of $C$,
  the following properties hold:
\begin{description}
\item[\boldmath$\varepsilon$-Correctness:] For honest $\S$ and $\R$, $\S$ has output $B_0,B_1 \in \{0,1\}$ and $\R$ has output
  $B_C$, except with probability~$\varepsilon$.\vspace{1ex}
\item[\boldmath$\varepsilon$-Receiver-security:] For honest $\R$ and any
  (dishonest) $\dS$ with output $V$,
\[\dist{ P_{C V} , P_{C} \cdot P_{V} } \leq \eps.\] 
\item[\boldmath$\varepsilon$-Sender-security:] For honest $\S$ and any
  (dishonest) $\dR$ with output $W$, there exists a binary
  random variable $D$ such that 
\[ \dist{ P_{B_{1-D} W B_D D} , P_{\unif} \cdot P_{W B_D D} } \leq \eps.
\]
\end{description}
\end{definition}
\index{correctness!of classical \RandOT}
\index{sender-security!of classical \RandOT}
\index{receiver-security!of classical \RandOT}
The condition for receiver-security simply says that $\S$ learns no
information on $C$, and sender-security requires that there exists a
choice bit $D$, supposed to be $C$, such that when given the choice
$D$ and the corresponding bit $B_D$, then the other bit, $B_{1-D}$, is
completely random from $\R$'s point of view.

We would like to point out that the definition of \RandOT
given in \cite{CSSW06} look syntactically slightly different than
our Definition~\ref{def:RandOT}.
However, it is not hard to see that they are actually equivalent. The
main difference is that the definition in~\cite{CSSW06} involves
an \index{auxiliary input}auxiliary input $Z$, which is given to the dishonest player, and
receiver- and sender-security as we define them are required to hold
{\em conditioned on $Z$} for any $Z$. Considering a {\em constant} $Z$
immediately proves one direction of the claimed equivalence, and the
other follows from the observation that if receiver- and
sender-security as we define them hold for {\em any} distribution
$P_{B_0 B_1 C}$ (respectively $P_C$), then they also hold for the
conditional distribution $P_{B_0 B_1 C|Z=z}$ (respectively
$P_{C|Z=z}$). The other difference is that in~\cite{CSSW06}, in
the condition for sender-security of \RandOT, $B_{1-D}$ is required to
be random and independent of $W$, $B_D$, $D$ {\em and $C$}. This of
course implies our sender-security condition (which is without $C$), but
it is also implied by our definition as $C$ may be part of the output
$W$.  We feel that simplifying the definitions as we do, without
changing their meaning, allows for an easier handling.

\subsection{Randomized \OT of Strings}

In a \onetwo\,{\em String}\:\pOT the sender inputs two {\em strings}
of the same length, and the receiver is allowed to learn one and only
one of the two. Formally, for any positive integer $\ell$, \lStringOT
and \RandlStringOT can be defined along the same lines as \OT and
\RandOT of {\em bits}: the binary random variables $B_0$ and $B_1$ as
well as $\unif$ in Definition~\ref{def:RandOT} are simply replaced by
random variables $S_0$ and $S_1$ and $\unif^{\ell}$ with range
$\set{0,1}^{\ell}$.

\section{Characterizing Sender-Security}\label{sec:Main}

\subsection{The Case of Bit \pOT}

It is well known and it follows from sender-security that in a
(\Rand)\:\OT the receiver $\R$ should in particular learn essentially
no information on the XOR $B_0 \oplus B_1$ of the two bits. The
following proposition shows that this is not only necessary for
sender-security but also {\em sufficient}.
\begin{theorem}\label{thm:xor}
  The condition for $\varepsilon$-sender-security for a \RandOT is
  satisfied for a particular (possibly dishonest) receiver $\dR$ with
  output $W$ if and only if
$$
\dist{ P_{(B_0 \oplus B_1) W} , P_{\unif} \cdot P_{W} } \leq \eps
\, .
$$
\end{theorem}

Before going into the proof which is surprisingly simple, consider
the following example. Assume a candidate protocol for \RandOT and a
dishonest receiver $\dR$ which is able to output $W = 0$ if $B_0 = 0 =
B_1$, $W = 1$ if $B_0 = 1 = B_1$ and $W = 0$ or $1$ with probability
$1/2$ each in case $B_0 \neq B_1$. Then, it is easy to see that
conditioned on, say, $W = 0$, $(B_0,B_1)$ is $(0,0)$ with probability
$\frac12$, and $(0,1)$ and $(1,0)$ each with probability $\frac14$,
such that the condition on the XOR from Theorem~\ref{thm:xor} is
satisfied. On the other hand, neither $B_0$ nor $B_1$ is uniformly distributed conditioned on $W = 0$, and 
it appears as if the receiver has some
joint information on $B_0$ and $B_1$ which is forbidden by a (\Rand)
\OT. But that is not so. Indeed, the same view can be obtained when
attacking an {\em ideal} \RandOT: submit a random bit $C$ to obtain $B_{C}$
and output $W = B_{C}$. In the light of Definition~\ref{def:RandOT},
if $W = 0$ we can split the event $(B_0,B_1) = (0,0)$ into two
disjoint subsets (subevents) ${\cal E}_0$ and ${\cal E}_1$ such that
each has probability $\frac14$, and we define $D$ by setting $D =
0$ if ${\cal E}_0$ or $(B_0,B_1) = (0,1)$, and $D = 1$ if ${\cal E}_1$
or $(B_0,B_1) = (1,0)$. Then, obviously, conditioned on $D = d$, the
bit $B_{1-d}$ is uniformly distributed, even when given $B_d$. The
corresponding holds if $W = 1$.

\begin{proof}
The ``only if'' implication is well-known and straightforward. For the ``if'' implication, we first argue the perfect case where 
\smash{$P_{(B_0 \oplus B_1) W} = P_{\unif} \cdot P_{W}$}. For any value $w$ with $P_W(w)>0$, the non-normalized distribution $P_{B_0 B_1 W}(\cdot,\cdot,w)$ can be expressed as depicted in the left table of Figure~\ref{fig:xor}, where we write $a$ for $P_{B_0 B_1 W}(0,0,w)$, $b$ for $P_{B_0 B_1 W}(0,1,w)$, $c$ for $P_{B_0 B_1 W}(1,0,w)$ and $d$ for $P_{B_0 B_1 W}(1,1,w)$.
Note that $a + b + c + d = P_W(w)$ and, by assumption, $a+d = b+c$. 
Due to symmetry, we may assume that $a \leq b$. We can then define $D$ by
extending $P_{B_0 B_1 W}(\cdot,\cdot,w)$ to $P_{B_0
B_1 D W}(\cdot,\cdot,\cdot,w)$ as depicted in the right two tables in
Figure~\ref{fig:xor}: $P_{B_0 B_1 D W}(0,0,0,w) = P_{B_0 B_1 D
  W}(0,1,0,w) = a$, $P_{B_0 B_1 D W}(1,0,0,w) = P_{B_0 B_1 D
  W}(1,1,0,w) = c$ etc. Important to realize is that $P_{B_0 B_1 D W}(\cdot,\cdot,\cdot,w)$ is indeed a valid extension since by assumption $c + (b-a) = d$. 

\begin{myfigure}{H}
$$
{
\begin{array}{|c|c|}
\hline & \\[-1ex] \;\;\;a\,\;\; & \;\;\;b\,\;\; \\[1ex] \hline & \\[-1ex] c & d \\[1ex] \hline
\end{array}
\atop P_{B_0 B_1 W}(\cdot,\cdot,w) }
\qquad\qquad\qquad
{
\begin{array}{|c|c|}
\hline & \\[-1ex] \;\;\;a\;\;\; & \;\;\;a\;\;\; \\[1ex] \hline & \\[-1ex] c & c \\[1ex] \hline
\end{array}
\atop P_{B_0 B_1 D W}(\cdot,\cdot,0,w) }
\qquad
{
\begin{array}{|c|c|}
\hline & \\[-1ex] \;\;\;0\;\;\; & \,b\!-\!a\, \\[1ex] \hline & \\[-1ex] 0 & b\!-\!a \\[1ex] \hline
\end{array}
\atop P_{B_0 B_1 D W}(\cdot,\cdot,1,w) }
$$
\vspace{-3ex}
\caption{Distributions $P_{B_0 B_1 W}(\cdot,\cdot,w)$ and $P_{B_0 B_1 D W}(\cdot,\cdot,\cdot,w)$}\label{fig:xor}
\end{myfigure}

It is now obvious that $P_{B_0 B_1 D W}(\cdot,\cdot,0,w) = \frac12 P_{B_0 D W}(\cdot,0,w)$ as well as $P_{B_0 B_1 D W}(\cdot,\cdot,1,w) = \frac12 P_{B_1 D W}(\cdot,1,w)$.
This finishes 
the perfect case. 

Concerning the general case, the idea is the same as above, except
that one has to take some care in handling the error parameter
$\varepsilon \geq 0$. As this does not give any new insight, and we
anyway state and fully prove a more general result in
Theorem~\ref{thm:main}, we skip this part of the
proof.\footnote{Although the special case $\ell = 1$ in
  Theorem~\ref{thm:main} is quantitatively slightly weaker than
  Theorem~\ref{thm:xor}. } 
\end{proof}

\subsection{The Case of String \pOT}

The obvious question after the previous section is whether there is a
natural generalization of Theorem~\ref{thm:xor} to 
\lStringOT for $\ell \geq 2$.  Note that the straightforward
generalization of the XOR-condition in Theorem~\ref{thm:xor},
requiring that any receiver has no information on the bit-wise XOR of
the two strings, is clearly too weak, and does not imply
sender-security for \RandlStringOT: for instance the receiver could
know the first half of the first string and the second half of the
second string.

\subsubsection{The Characterization}

Let $\ell$ be an arbitrary positive integer. 

\begin{definition}\label{def:linear}
A 
function $\bal:\set{0,1}^{\ell}\times\set{0,1}^{\ell} \rightarrow
\set{0,1}$ is called a \index{non-degenerate linear
function}{\em non-degenerate linear function} (NDLF) if it is of the
form
$$\bal: (s_0,s_1) \mapsto \ip{a_0,s_0} \oplus \ip{a_1,s_1}$$
for two non-zero $a_0,a_1 \in \set{0,1}^{\ell}$, i.e.,
if it is linear and non-trivially depends on both input strings.
\end{definition}
Even though this is the main notion we are using, the following more relaxed notion allows to make some of our claims slightly stronger. 

\begin{definition}\label{def:balanced} 
\index{2-balanced|see {balanced function}}
  A binary function $\bal:\set{0,1}^{\ell}\times\set{0,1}^{\ell}
  \rightarrow \set{0,1}$ is called {\em 2-\index{balanced function}balanced} if for any
  $s_0,s_1 \in \set{0,1}^{\ell}$ the functions $\bal(s_0,\cdot)$ and
  $\bal(\cdot,s_1)$ are balanced in the usual sense, meaning that
  $\card{\Set{\sigma_1 \in \set{0,1}^{\ell}}{\bal(s_0,\sigma_1)\=0}} =
  2^{\ell}/2$ and $\card{\Set{\sigma_0 \in
      \set{0,1}^{\ell}}{\bal(\sigma_0,s_1)\=0}} = 2^{\ell}/2$.
\end{definition}
The following is easy to see and the proof is omitted. 
\begin{lemma}
Every non-degenerate linear function is 2-balanced. 
\end{lemma}
In case $\ell = 1$, the XOR is a NDLF and thus 2-balanced, and it is
the {\em only} NDLF and up to addition of a constant the only
2-balanced function.  Based on this notion of non-degenerate linear
functions, sender-security of \RandStringOT can be characterized as
follows.

\begin{theorem}\label{thm:main} 
\index{sender-security!characterization of}
  The condition of $\varepsilon$-sender-security for a \RandlStringOT
  is satisfied for a particular (possibly dishonest) receiver $\dR$
  with output $W$ if
$$
\dist{ P_{\bal(S_0,S_1) W} , P_{\unif} \cdot P_W } \leq \varepsilon/2^{2\ell+1}
$$
for every NDLF $\bal$, and, on the other hand, 
$\varepsilon$-sender-security may be satisfied only if $\dist{ P_{\bal(S_0,S_1) W} , P_{\unif}
\cdot P_W } \leq \eps$ for every NDLF~$\bal$. 
\end{theorem}
The number of \index{non-degenerate linear function} NDLFs is
exponential in~$\ell$, namely $(2^{\ell}-1)^2$.  Nevertheless, we show
in Section~\ref{sec:application} that this characterization turns out
to be very useful. There, we will also argue that an exponential
overhead in $\ell$ in the sufficient condition is unavoidable. The
proof of Theorem~\ref{thm:main} also shows that the set of NDLFs forms
a minimal set of functions among all sets that imply sender-security.
In this sense, our characterization is tight.

At first glance, Theorem~\ref{thm:main} appears to be related to the
so-called (information-theoretic) \index{XOR-Lemma}XOR-Lemma, commonly attributed to
Vazirani~\cite{Vazirani86} and nicely explained by
Goldreich~\cite{Goldreich95}, which states that a string is close to
uniform if the XOR of the bits of any non-empty substring are. As far
as we can see, neither follows Theorem~\ref{thm:main} from the
XOR-Lemma in an obvious way nor can it be proven by modifying the
proof of the XOR-Lemma, as given in~\cite{Goldreich95}.

Furthermore, we would like to point out that Theorem 4 in~\cite{BCW03} also
provides a tool to analyze sender-security of \OT 
protocols in terms of linear functions; however, the condition that
needs to be satisfied is much stronger than for our
Theorem~\ref{thm:main}: it additionally requires that one of the two
strings is {\em a priori} uniformly distributed from the receiver's
point of view.\footnote{Concretely, it is additionally required that
  every non-trivial parity of that string is uniform, but by the
  XOR-Lemma this is equivalent to the whole string being uniform. }
This difference is crucial, because showing that one of the
two strings is uniform (conditioned on the receiver's view) is usually
technically involved and sometimes not even possible, as the example
given after Theorem~\ref{thm:xor} shows. This is also demonstrated by
the fact that the analysis in~\cite{BCW03} of the considered \OT 
protocol is tailored to one particular class of privacy-amplifying
hash functions, and it is stated as an open problem how to prove their
construction secure when a different class of hash functions is used.
The condition for Theorem~\ref{thm:main}, on the other hand, is
naturally satisfied for typical constructions of \OT protocols, as we
shall see in Section~\ref{sec:application}. As a result,
Theorem~\ref{thm:main} allows for much simpler and more elegant
security proofs for \OT protocols, and, as a by-product, allows to
solve the open problem from~\cite{BCW03}. We explain this in detail in
Section~\ref{sec:application}, and the interested reader may well jump
ahead and save the proof of Theorem~\ref{thm:main} for later.

\subsubsection{Proof of Theorem~\ref{thm:main} (``only if'' part)}
\label{app:necessary}
We start with the proof for the ``only if'' part of
Theorem~\ref{thm:main}. In fact, a slightly stronger statement is
shown, namely that $\varepsilon$-sender-security implies
$\dist{ P_{\bal(S_0,S_1) W} , P_{\unif} \cdot P_W } \leq \eps$ for any
{\em 2-balanced} function. \index{balanced function}

According to Definition~\ref{def:RandOT}, $\varepsilon$-sender-security for \RandOT is satisfied for a receiver $\R$ with output $W$ if there
exists a random variable $D$ with range $\{0,1\}$ such that
$$
\frac12 \sum_{w,d,s_0,s_1} \big| P_{S_{1-D}S_D D W}(s_{1-d},s_d,d,w) - 2^{-\ell}
P_{S_D D W}(s_d,d,w) \big| \leq \varepsilon. 
$$
In order to upper bound 
$$
  \dist{ P_{\beta(S_0,S_1)W},P_{\unif} \cdot P_{W} } = \frac12 \sum_{w,b}
  \big|P_{\beta(S_0,S_1) W}(b,w) - \frac12 P_W(w) \big| \label{toupperbound}
$$
we expand the terms on the right hand side as follows. 
\begin{align*}                   
P_{\beta(S_0,S_1) W}(b,w) &= \sum_d P_{\beta(S_0,S_1) D W} (b,d,w)\\
 &= \sum_d \!\sum_{{s_d,s_{1-d}}\atop{\beta(s_0,s_1)=b}}\!\!\!\!\! P_{S_{1-D}
  S_D D W} (s_{1-d},s_d,d,w)
\end{align*}
and
$$
P_{W}(w) = \sum_d \sum_{s_d} P_{S_D D W}(s_d,d,w) 
= \sum_d 2^{-\ell+1} \cdot \!\!\!\!\!
\sum_{{s_d,s_{1-d}}\atop{\beta(s_0,s_1)=b}}\!\!\!\!\! P_{S_D D W} (s_d,d,w) 
$$
where the last equality holds because there are $2^{\ell-1}$ values for $s_{1-d}$
such that $\beta(s_0,s_1)=b$, as $\beta$ is a 2-balanced function. 
Using those two expansions we conclude that
\begin{align*}
 \dist{ &P_{\beta(S_0,S_1) W} , P_{\unif} \cdot P_{W} } \\[0.5ex]
&\leq \frac12 \sum_{w,b} \sum_{d} \!\sum_{{s_d,s_{1-d}}\atop{\beta(s_0,s_1)=b}}\!\! \big| P_{S_{1-D}
  S_D D W} (s_{1-d},s_d,d,w) - 2^{-\ell} P_{S_D D W} (s_d,d,w)
\big|  \\
&= \frac12 \sum_{w,d,s_0,s_1}\! \big| P_{S_{1-D}S_D D W}(s_{1-d},s_d,d,w) - 2^{-\ell}
P_{S_D D W}(s_d,d,w) \big|  \,\leq\, \varepsilon. 
\end{align*}
where the first inequality follows follows from the above expansions and the triangle inequality and the last inequality is our initial assumption. \qed

The ``if'' part, which is the interesting direction, is proven below. 

\subsubsection{The Case $\ell = 2$}
We feel that in order to understand the proof of
Theorem~\ref{thm:main}, it is useful to first consider the case $\ell
= 2$. Let us focus on trying to develop a condition that is sufficient
for {\em perfect} sender-security. Fix an arbitrary output $w$, and
consider an arbitrary non-normalized probability distribution $P_{S_0
  S_1 W}(\cdot,\cdot,w)$ of $S_0$ and $S_1$ when $W = w$. This is
depicted in the left table of Figure~\ref{fig:l=2}, where we write $a$
for $P_{S_0 S_1 W}(00,00,w)$, $b$ for $P_{S_0 S_1 W}(00,01,w)$, etc.
We may assume that $a \leq b,c,d$. We now extend this distribution to
$P_{S_0 S_1 D W}(\cdot,\cdot,\cdot,w)$ similar as in the proof of
Theorem~\ref{thm:xor}. This is depicted in the two right tables in
Figure~\ref{fig:l=2}. We verify what conditions $P_{S_0 S_1
  W}(\cdot,\cdot,w)$ must satisfy such that $P_{S_0 S_1 D W}$ is
indeed a valid extension, i.e., that $P_{S_0 S_1 D
  W}(\cdot,\cdot,0,w)+P_{S_0 S_1 D W}(\cdot,\cdot,1,w)=P_{S_0 S_1
  W}(\cdot,\cdot,w)$.

\def\noex#1{\makebox[0ex]{$#1$}}
\def\m{\! - \!}

\begin{myfigure}{H}
$$
{\setlength{\arraycolsep}{2.5ex}
\begin{array}{|c|c|c|c|} \hline 
& & & \\[-1ex] 
\noex a & \noex b & \noex c & \noex d \\[0.5ex] \hline 
& & & \\[-1ex] 
\noex e & \noex f & \noex g & \noex h \\[0.5ex] \hline
& & & \\[-1ex] 
\noex i & \noex j & \noex k & \noex l \\[0.5ex] \hline
& & & \\[-1ex] 
\noex m & \noex n & \noex o & \noex p \\[0.5ex] \hline
\end{array}
\atop P_{S_0 S_1 W}(\cdot,\cdot,w) }
\qquad\qquad
{\setlength{\arraycolsep}{2.5ex}
\begin{array}{|c|c|c|c|} \hline 
& & & \\[-1ex] 
\noex a & \noex a & \noex a & \noex a \\[0.5ex] \hline 
& & & \\[-1ex] 
\noex e & \noex e & \noex e & \noex e \\[0.5ex] \hline
& & & \\[-1ex] 
\noex i & \noex i & \noex i & \noex i \\[0.5ex] \hline
& & & \\[-1ex] 
\noex m & \noex m & \noex m & \noex m \\[0.5ex] \hline
\end{array}
\atop P_{S_0 S_1 D W}(\cdot,\cdot,0,w) }
\qquad
{\setlength{\arraycolsep}{2.5ex}
\begin{array}{|c|c|c|c|} \hline 
 & & & \\[-1ex] 
\noex 0 & \noex{b \m a} & \noex{c \m a} & \noex{d \m a} \\[0.5ex] \hline 
 & & & \\[-1ex] 
\noex 0 & \noex{b \m a} & \noex{c \m a} & \noex{d \m a} \\[0.5ex] \hline 
 & & & \\[-1ex] 
\noex 0 & \noex{b \m a} & \noex{c \m a} & \noex{d \m a} \\[0.5ex] \hline 
 & & & \\[-1ex] 
\noex 0 & \noex{b \m a} & \noex{c \m a} & \noex{d \m a} \\[0.5ex] \hline 
\end{array}
\atop P_{S_0 S_1 D W}(\cdot,\cdot,1,w) }
$$
\vspace{-3ex}
\caption{Distributions $P_{S_0 S_1 W}(\cdot,\cdot,w)$ and $P_{S_0 S_1
D W}(\cdot,\cdot,\cdot,w)$}\label{fig:l=2}
\end{myfigure}

For instance, looking at the second row and second column we get
equation $e + (b-a) = f$. Altogether, we get the following system of
equations.
\begin{align*}
b+e &= a+f   & b+i &= a+j   & b+m &= a+n \\
c+e &= a+g   & c+i &= a+k   & c+m &= a+o \\
d+e &= a+h   & d+i &= a+l   & d+m &= a+p
\end{align*}
Note that if all these equations do hold for any $w$, then $P_{S_0
  S_1DW}(\cdot,\cdot,\cdot,\cdot)$ is well defined and satisfies
$P_{S_0 S_1 D W}(\cdot,\cdot,0,\cdot) = \frac14 P_{S_0 D
  W}(\cdot,0,\cdot)$ and $P_{S_0 S_1 D W}(\cdot,\cdot,1,\cdot) =
\frac14 P_{S_1 D W}(\cdot,1,\cdot)$, in other words, perfect
sender-security holds.

The idea now is to show that the above equation system is equivalent
to another equation system, in which every equation expresses that a
certain \index{non-degenerate linear function} NDLF applied to $S_0$
and $S_1$ is uniformly distributed when $W=w$, which holds by
assumption.

For example, by adding all the equations in the original system while
taking every second equation with negative sign, one gets the equation
$$
b+d+e+g+j+l+m+o = a+c+f+h+i+k+n+p \, .
$$
Define the function $\bal:\set{0,1}^2 \times \set{0,1}^2 \rightarrow \set{0,1}$ as follows. Let $\bal(s_0,s_1)$ be $0$ if the entry which corresponds to $(s_0,s_1)$ in the left table in Figure~\ref{fig:l=2} appears on the left hand side of the above equation, and else we let $\bal(s_0,s_1)$ be $1$. Then the above equation simply says that $\bal(S_0,S_1) = 0$ with the same probability as $\bal(S_0,S_1) = 1$ (when $W = w$). 
Note that it is crucial that in the above equation every variable $a$ up to $p$ occurs with multiplicity exactly 1.
By comparing the function tables, it is now easy to verify that $\bal$
coincides with the function $(s_0,s_1) \mapsto s_{02} \oplus s_{12}$, where $s_{i2}$ denotes the second coordinate of $s_i \in \set{0,1}^2$, thus is a NDLF.


One can now show (and we are going to do this below for an arbitrary $\ell$) that there are enough such equations, corresponding to NDLFs, such that these equations imply the original ones. This implies that if $\bal(S_0,S_1)$ is distributed uniformly and independently of $W$ for every NDLF $\bal$, then the original equation system is satisfied (for any~$w$), and thus $P_{S_0 S_1 D W}$ is well-defined.

\subsubsection{Proof of Theorem~\ref{thm:main} (``if'' part). }

First, we consider the perfect case: if $P_{\bal(S_0,S_1) W}$ {\em
  equals} $P_{\unif} \cdot P_{W}$ for every NDLF $\bal$, then
  sender-security for Rand\ \lStringOT holds perfectly.
\index{non-degenerate linear function}

\paragraph{\sc The Perfect Case: }

Since the case $\ell = 1$ is already settled, we assume that $\ell
\geq 2$.  We generalize the idea from the case $\ell = 2$. The main
issue will be to transform the equations guaranteed by the assumption
on the linear functions into the ones required for $P_{S_0 S_1 D
  W}(\cdot,\cdot,0,w)+P_{S_0 S_1 D W}(\cdot,\cdot,1,w)=P_{S_0 S_1
  W}(\cdot,\cdot,w)$.

Fix an arbitrary output $w$ of the receiver, and consider the
non-normalized probability distribution $P_{S_0 S_1
  W}(\cdot,\cdot,w)$.
We use the variable $p_{s_0,s_1}$ to refer to $P_{S_0 S_1
  W}(s_0,s_1,w)$, and we write $\zero$ for the all-zero string $(0,\ldots,0)
\in \set{0,1}^{\ell}$. We assume that $p_{\zero,\zero} \leq
p_{\zero,s_1}$ for any \mbox{$s_1 \in \set{0,1}^{\ell}$}; we show
later that we may do so. We extend this distribution to $P_{S_0 S_1 D
  W}(\cdot,\cdot,\cdot,w)$ by setting
\begin{equation}\label{eq:extension}
P_{S_0 S_1 D W}(s_0,s_1,0,w) = p_{s_0,\zero}
\quad\text{and}\quad
P_{S_0 S_1 D W}(s_0,s_1,1,w) = p_{\zero,s_1}-p_{\zero,\zero}
\end{equation}
for any strings $s_0,s_1 \in \set{0,1}^{\ell}$, and we collect the
equations resulting from the condition that $P_{S_0 S_1
  W}(\cdot,\cdot,w) = P_{S_0 S_1 D W}(\cdot,\cdot,0,w)+P_{S_0 S_1 D
  W}(\cdot,\cdot,1,w)$ needs to be satisfied: for any two $s_0,s_1 \in \set{0,1}^{\ell}
\setminus \set{\zero}$
\begin{equation}\label{eq:original}
p_{s_0,\zero} + p_{\zero,s_1} = p_{\zero,\zero} + p_{s_0,s_1} \, .
\end{equation}
If all these equations do hold for any $w$, then as in the case of
$\ell = 1$ or $\ell = 2$, the random variable $D$ is well defined and
$P_{S_{1-D} S_D W D} = P_{\unif^{\ell}} \cdot P_{S_D W D}$ holds,
since $P_{S_0 S_1 D W}(s_0,s_1,0,w)$ does not depend on $s_1$ and
$P_{S_0 S_1 D W}(s_0,s_1,1,w)$ not on $s_0$.


We proceed by showing that the equations provided by the assumed
uniformity of $\bal(S_0,S_1)$ for any $\bal$ imply the equations given
by (\ref{eq:original}). Consider an arbitrary pair $a_0,a_1 \in
\set{0,1}^{\ell} \setminus{\set{\zero}}$ and let $\beta$ be the
associated NDLF, i.e., such that $\bal(s_0,s_1) = \ip{a_0,s_0} \oplus
\ip{a_1,s_1}$.  By assumption, $\bal(S_0,S_1)$ is uniformly
distributed, independent of $W$. Thus, for any fixed $w$, this can be
expressed as
\begin{equation}\label{eq:linear}
\sum_{\sigma_0,\sigma_1: \atop \ip{a_0,\sigma_0} = \ip{a_1,\sigma_1}}\!\!\!\!\! p_{\sigma_0,\sigma_1}\; = \!\!\!\!\sum_{\sigma_0,\sigma_1: \atop \ip{a_0,\sigma_0} \neq \ip{a_1,\sigma_1}}\!\!\!\!\! p_{\sigma_0,\sigma_1} \, ,
\end{equation}
where both summations are over all $\sigma_0,\sigma_1 \in
\set{0,1}^{\ell}$ subject to the indicated respective properties.
Recall, that this equality holds for any pair $a_0,a_1 \in
\set{0,1}^{\ell} \setminus{\set{\zero}}$. Thus, for fixed $s_0,s_1 \in
\set{0,1}^{\ell} \setminus \set{\zero}$, if we sum over all such
pairs $a_0,a_1$ subject to $\ip{a_0,s_0} = \ip{a_1,s_1} = 1$, we get
the equation
$$
\sum_{a_0,a_1: \atop \ip{a_0,s_0} = \ip{a_1,s_1} = 1}
\sum_{\sigma_0,\sigma_1: \atop \ip{a_0,\sigma_0} =
  \ip{a_1,\sigma_1}}\!\!\!\!\! p_{\sigma_0,\sigma_1}\; =
\!\!\!\!\sum_{a_0,a_1: \atop \ip{a_0,s_0} = \ip{a_1,s_1} = 1}
\sum_{\sigma_0,\sigma_1: \atop \ip{a_0,\sigma_0} \neq
  \ip{a_1,\sigma_1}}\!\!\!\!\! p_{\sigma_0,\sigma_1} \, ,
$$
which, after re-arranging the terms of the summations, leads to
\begin{equation}\label{eq:generated}
\sum_{\sigma_0,\sigma_1}
\sum_{a_0,a_1: \atop {\ip{a_0,s_0} = \ip{a_1,s_1} = 1 \atop \ip{a_0,\sigma_0} = \ip{a_1,\sigma_1}}}\!\!\!\!\! p_{\sigma_0,\sigma_1}\; = \;
\sum_{\sigma_0,\sigma_1}
\sum_{a_0,a_1: \atop {\ip{a_0,s_0} = \ip{a_1,s_1} = 1 \atop \ip{a_0,\sigma_0} \neq \ip{a_1,\sigma_1}}}\!\!\!\!\! p_{\sigma_0,\sigma_1} \, .
\end{equation}
We will now argue that, up to a constant multiplicative
factor, equation (\ref{eq:generated}) coincides with equation
(\ref{eq:original}).

First, it is straightforward to verify that the variables
$p_{\zero,\zero}$ and $p_{s_0,s_1}$ occur only on the left hand side,
both with multiplicity $2^{2(\ell-1)}$ (the number of pairs $a_0,a_1$
such that $\ip{a_0,s_0} = \ip{a_1,s_1} = 1$), whereas $p_{s_0,\zero}$
and $p_{\zero,s_1}$ only occur on the right hand side, with the same
multiplicity $2^{2(\ell-1)}$.

Now, we argue that any other $p_{\sigma_0,\sigma_1}$ equally often
appears on the right and on the left hand side, and thus cancel out.
Note that the set of pairs $a_0,a_1$, over which the summation runs on
the left respectively the right hand side, can be understood as the
set of solutions to a binary non-homogeneous linear equations system:
$$
\left(
\begin{array}{cc}
s_0 & 0 \\
0   & s_1 \\
\sigma_0 & \sigma_1 
\end{array}
\right)
\left(
\begin{array}{c}
a_0 \\ a_1
\end{array}
\right) 
=
\left(
\begin{array}{c}
1 \\ 1 \\ 0
\end{array}
\right)
\;\text{respectively}\;
\left(
\begin{array}{c}
1 \\ 1 \\ 1
\end{array}
\right)\, .
$$
Also note that the two linear equation systems consist of three
equations and involve at least 4 variables, because $a_0,a_1 \in
\set{0,1}^{\ell}$ and $\ell \geq 2$. Therefore, using basic linear
algebra, one is tempted to conclude that they both have solutions,
and, because they have the same homogeneous part, they have the same
number of solutions, equal to the number of homogeneous solutions.
However, this is only guaranteed if the matrix defining the
homogeneous part has full rank. In our situation, this is precisely the case if
and only if $(\sigma_0,\sigma_1) \not\in
\set{(\zero,\zero),(s_0,\zero),(\zero,s_1),(s_0,s_1)}$, where those
four exceptions have already been treated above. It follows that the
equations (\ref{eq:linear}), which are guaranteed by assumption, imply
the equations (\ref{eq:original}).
 
It remains to justify the assumption that
$p_{\zero,\zero} \leq p_{\zero,s_1}$ for any $s_1$. In general, we
choose $t \in \set{0,1}^{\ell}$ such that $p_{\zero,t} \leq
p_{\zero,s_1}$ for any $s_1 \in \set{0,1}^{\ell}$, and we set $P_{S_0
  S_1 D W}(s_0,s_1,0,w) = p_{s_0,t}$ and $P_{S_0 S_1 D W}(s_0,s_1,1,w)
= p_{\zero,s_1}-p_{\zero,t}$, resulting in the equation $p_{s_0,t} +
p_{\zero,s_1} = p_{\zero,t} + p_{s_0,s_1}$ that needs to be satisfied
for $s_0 \in \set{0,1}^{\ell}\setminus\set{\zero}$ and $s_1 \in
\set{0,1}^{\ell}\setminus\set{t}$. This equality, though, can be
argued as for equation (\ref{eq:original}), which we did above, simply
by replacing $p_{\sigma_0,\sigma_1}$ on both sides of
(\ref{eq:linear}) by $p_{\sigma_0,\sigma_1 \oplus t}$ (where $\oplus$
is the bit wise XOR). We may safely do so: doing a suitable variable
substitution and using linearity of the inner product, it is easy to
see that this modified equation still expresses uniformity of
$\bal(S_0,S_1)$. This concludes the proof for the perfect case.

\paragraph{\sc The General Case: }

Now, we consider the general case where there exists some $\varepsilon
> 0$ such that $\dist{P_{\bal(S_0,S_1) W}, P_{\unif} \cdot P_{W}} \leq
2^{-2\ell-1}\varepsilon$ for any NDLF $\bal$. We use the observations
from the perfect case but additionally keep track of the ``error
term''.
\index{non-degenerate linear function}

For any $w$ with $P_W(w) > 0$ and any NDLF $\bal$, set 
$$
\varepsilon_{w,\bal} = \dist{P_{\bal(S_0,S_1)
    W}(\cdot,w),P_{\unif} \cdot P_W(w)}\, .
$$
Note that $\sum_w \varepsilon_{w,\bal} =\dist{ P_{\bal(S_0,S_1) W} ,
P_{\unif} \cdot P_{W} } \leq 2^{-2\ell-1}\varepsilon$, independent of
$\bal$.  Fix now an arbitrary $w$ with $P_W(w) > 0$. Then,
(\ref{eq:linear}) only holds up to an error of
$2\varepsilon_{w,\bal}$, where $\bal$ is the NDLF associated to
$a_0,a_1$. As a consequence, Equation~\eqref{eq:generated} only holds
up to an error of $2 \sum_{\bal}\varepsilon_{w,\bal}$ and thus
(\ref{eq:original}) holds up to an error of $ \delta_{s_0,s_1} =
\frac{2}{2^{2\ell-2}} \sum_{\bal}\varepsilon_{w,\bal} \, , $ where the
sum is over the $2^{2\ell-2}$ functions associated to the pairs
$a_0,a_1$ with $\ip{a_0,s_0} = \ip{a_1,s_1} = 1$.  Note that
$\delta_{s_0,s_1}$ depends on $w$, but the set of $\bal$'s, over which
the summation runs, does not.  Adding up over all possible $w$'s gives
$$
\sum_w \delta_{s_0,s_1} = \frac{2}{2^{2\ell-2}} \sum_w
\sum_{\bal}\varepsilon_{w,\bal} = \frac{2}{2^{2\ell-2}} \sum_{\bal}
\sum_w \varepsilon_{w,\bal} \leq 2^{-2\ell} \varepsilon \, .
$$

Since (\ref{eq:original}) only holds approximately, $P_{S_0 S_1 D W}$
as in~(\ref{eq:extension}) is not necessarily a valid extension, but
close. This can obviously be overcome by instead setting
\begin{align*}
P_{S_0 S_1 D W}(s_0,s_1,0,w) &= p_{s_0,\zero} \pm \delta'_{s_0,s_1}
\quad \text{ and }\\
P_{S_0 S_1 D W}(s_0,s_1,1,w) &= p_{\zero,s_1}-p_{\zero,\zero} \pm \delta''_{s_0,s_1}
\end{align*}
with suitably chosen $\delta'_{s_0,s_1},\delta''_{s_0,s_1} \geq 0$
with $\delta'_{s_0,s_1}+\delta''_{s_0,s_1} = \delta_{s_0,s_1}$, and
with suitably chosen signs ``$+$'' or ``$-$''.\footnote{Most of the
  time, it probably suffices to correct one of the two, say, choose
  $\delta'_{s_0,s_1} = \delta_{s_0,s_1}$ and $\delta''_{s_0,s_1} = 0$;
  however, if for instance $p_{s_0,\zero}$ and $p_{\zero,s_1} -
  p_{\zero,\zero}$ are both positive but $P_{S_0 S_1 W}(s_0,s_1,w) =
  0$, then one has to correct both. } Using that every $P_{S_0 S_1 D
  W}(s_0,s_1,0,w)$ differs from $p_{s_0,\zero}$ by at most
$\delta'_{s_0,s_1}$, it follows from a straightforward computation
that $ \dist{P_{S_{1-D} S_D D W}(\cdot,\cdot,0,w),P_{\unif} P_{S_D D
    W}(\cdot,0,w)} \leq \sum_{s_0,s_1} \delta'_{s_0,s_1} \, .  $ The
corresponding holds for $P_{S_0 S_1 D W}(\cdot,\cdot,1,w)$. It follows
that
\begin{align*}
  \dist{P_{S_{1-D} S_D W D},P_{\unif} P_{S_D W D}} \leq \sum_w
  \sum_{s_0,s_1} (\delta'_{s_0,s_1} + \delta''_{s_0,s_1}) =
  \sum_{s_0,s_1} \sum_w \delta_{s_0,s_1} \leq \varepsilon
\end{align*}  
which concludes the proof. \qed

\section{Applications}\label{sec:application}

In this section we will show the usefulness of Theorem~\ref{thm:main}
for the construction of \lStringOT, based on weaker primitives like a
noisy channel or other flavors of \pOT.  In particular, we will show
that the reducibility of \OT to any weaker flavor of \pOT follows as
a simple argument using Theorem~\ref{thm:main}.

\subsection[Reducing \lStringOT to Repetitions of Weak \OT\!\!s]{Reducing \lStringOT to Independent Repetitions of Weak \OT
  \!\!s}\label{sec:reductions}

\subsubsection{Background}
A great deal of effort has been put into constructing protocols for
\lStringOT based on physical assumptions like various models for noisy
channels~\cite{CK88,DKS99,DFMS04,CMW04} or a memory bounded
adversary~\cite{CCM98,Din01,DHRS04}, as well as into reducing
\lStringOT to (seemingly) weaker flavors of \pOT, like \RabinOT, \XOT,
\GOT and \BUOT~\cite{Crepeau87,BC97,Cachin98,Wolf00,BCW03,CS06,
  Wullschleger07}.  Note that the latter three flavors of \pOT are
weaker than \OT in that the dishonest receiver has more freedom in
choosing the sort of information he wants to get about the sender's
input bits $B_0$ and $B_1$: $B_0$, $B_1$ or $B_0 \oplus B_1$ in case
of \textsl{\onetwo\:XOR-\pOT} (which is abbreviated by \XOT),
$g(B_0,B_1)$ for an arbitrary one-bit-output function $g$ in case of
\textsl{\onetwo\: Generalized-\pOT (\GOT)}, and an arbitrary
probabilistic $Y$ with mutual information $I(B_0 B_1;Y) \leq 1$ in
case of \textsl{\onetwo\: Universal-\pOT (\BUOT)}.\footnote{As a
  matter of fact, reducibility has been proven for any bound on $I(B_0
  B_1;Y)$ strictly smaller than $2$.  Note that there is some
  confusion in the literature in what a {\em Universal} $\pOT$,
  $\pUOT$ is: In~\cite{BC97,Wolf00,BCW03}, a $\pUOT$ takes as
  input two {\em bits} and the receiver is doomed to have at least one
  bit or any other
non-trivial amount of {\em Shannon} entropy on them; we denote this 
by $\BUOT$. Whereas in~\cite{Cachin98}, a $\pUOT$ takes as input two
{\em strings} and the receiver is doomed to have some {\em R\'enyi}
entropy of order $\alpha > 1$ on them. We address this latter notion in more detail in
Section~\ref{sec:UOT}. }

All these reductions of \OT to weaker versions follow a specific
construction design, which is also at the core of the \OT protocols
based on noisy channels or a memory-bounded adversary. By repeated
independent executions of the underlying primitive, $\S$ transfers a
randomly chosen bit string $X = (X_0,X_1) \in \set{0,1}^n \times
\set{0,1}^n$ to $\R$ such that: 
\vspace{-3mm}
\begin{enumerate} \addtolength{\itemsep}{-3mm} 
\item depending on his choice bit $C$, the honest $\R$ knows either
  $X_0$ or $X_1$,
\item any $\dS$ has no information on which part of $X$ $\R$ learned,
  and 
\item any $\dR$ has some uncertainty in $X$. 
\end{enumerate}

Then, this is completed to a \RandOT by means of privacy
amplification (cf. Section~\ref{sec:pa}): $\S$ samples two functions $\hf_0$ and
$\hf_1$ from a \univ class $\UH$ of hash functions, sends them to
$\R$, and outputs $S_0 = \hf_0(X_0)$ and $S_1 = \hf_1(X_1)$, and $\R$
outputs $S_C = \hf_C(X_C)$. Finally, the \RandOT is transformed into
an ordinary \OT in the obvious way.

Correctness and receiver-security of this construction are clear, they follow
immediately from 1.\ and 2. How easy or hard it is to prove
sender-security depends heavily on the underlying primitive. In case of
\RabinOT it is rather straightforward. In case of \XOT and the other
weaker versions, this is non-trivial. The problem is that since $\R$
might know $X_0 \oplus X_1$, it is not possible to argue that there
exists $d \in \set{0,1}$ such that $\R$'s uncertainty on $X_{1-d}$ is
large when given $X_d$. This, though, would be necessary in order to
finish the proof by simply applying the privacy amplification
theorem (Corollary~\ref{thm:classicalpa}). This difficulty is overcome in
\cite{BC97,BCW03} by tailoring the proof to a particular \univ class
of hash functions, namely the class of all {\em linear} hash
functions. Whether the reduction also works for a less restricted
class of hash functions is left in~\cite{BC97,BCW03} as an open
problem, which we solve here as a side result. Using a smaller class of hash functions would allow for instance to reduce the communication complexity of the protocol. 

In \cite{CS06}, the difficulty is overcome by giving up on the
simplicity of the reduction. The cost of two-way communication
allowing for \index{interactive hashing}interactive hashing is traded for better reduction
parameters. We would like to emphasize that these parameters are
incomparable to ours, because a different reduction is used, whereas
our approach provides a \emph{better analysis} of the common
non-interactive reductions.

\subsubsection{The New Approach}

We argue that, independent of the underlying primitive,
sender-security follows as a simple consequence of
Theorem~\ref{thm:main}, in combination with a simple observation
regarding the composition of \index{non-degenerate linear
function}non-degenerate linear (respectively, more general,
2-\index{balanced function}balanced) functions with strongly \univ
hash functions, stated in Proposition~\ref{prop:balanced->u2} below.
\index{two-universal hashing!strongly} 
\index{2-universal hashing|see {two-universal hashing}}

Recall Definition~\ref{def:strongly-two-universal} of strong
two-universality.  A class $\UH$ of hash functions from $\set{0,1}^n$
to $\set{0,1}^{\ell}$ is {\em strongly \univ}, if for any
distinct $x,x' \in \set{0,1}^n$ the two random variables $\Hf(x)$ and
$\Hf(x')$ are independent and uniformly distributed over
$\set{0,1}^{\ell}$, where the random variable $\Hf$ represents the
random choice of a function in $\UH$.
\begin{proposition}\label{prop:balanced->u2}
  Let $\UH_0$ and $\UH_1$ be two classes of strongly \univ hash
  functions from $\set{0,1}^{n_0}$ respectively $\set{0,1}^{n_1}$ to
  $\set{0,1}^{\ell}$, and let
  $\bal:\set{0,1}^{\ell}\times\set{0,1}^{\ell}\rightarrow\set{0,1}$ be
  a 2-balanced function. Consider the class $\UH$ of all functions
  $\hf:\set{0,1}^{n_0}\times\set{0,1}^{n_1}\rightarrow\set{0,1}$ with
  $\hf(x_0,x_1) = \bal(\hf_0(x_0),\hf_1(x_1))$ where $\hf_0\in \UH_0$
  and $\hf_1 \in \UH_1$. Then, $\UH$ is strongly \univ.\footnote{It is
    easy to see that the claim does not hold in general for ordinary
    (as opposed to strongly) \univ classes: if $n_0 = n_1 = \ell$ and
    $\UH_0$ and $\UH_1$ both only contain the identity function
    $id:\set{0,1}^{\ell}\rightarrow\set{0,1}^{\ell}$ and thus are
    \univ, then $\UH$ consisting of the function $\hf(x_0,x_1) =
    \bal(id(x_0),id(x_1)) = \bal(x_0,x_1)$ is not \univ. }
\end{proposition} 
\begin{proof}
  Fix distinct $x = (x_0,x_1)$ and $x' = (x'_0,x'_1)$ in
  $\set{0,1}^{n_0}\times\set{0,1}^{n_1}$. Assume without loss of
  generality that $x_1 \neq x'_1$. Fix $\hf_0 \in \UH_0$, and set $s_0
  = \hf_0(x_0)$ and $s'_0 = \hf_0(x'_0)$. By assumption on $\UH_1$,
  the random variables $\Hf_1(x_1)$ and $\Hf_1(x'_1)$ are independent and
  uniformly distributed over $\set{0,1}^{\ell}$, where $\Hf_1$
  represents the random choice for $\hf_1 \in \UH_1$. By the
  assumption on $\bal$, this implies that
  $\bal(\hf_0(x_0),\Hf_1(x_1))$ and $\bal(\hf_0(x'_0),\Hf_1(x'_1))$
  are independent and uniformly distributed over $\set{0,1}$. This holds
  no matter how $\hf_0$ is chosen, and thus proves the claim.
\end{proof}

Now, briefly, sender-security for a construction as sketched above can
be argued as follows: The only restriction is that $\UH$ needs to be
{\em strongly} \univ. From the independent repetitions of the
underlying weak \pOT (\RabinOT, \XOT, \GOT or \BUOT) it follows that
$\dR$ has ``high'' collision entropy in $X$. Hence, for any NDLF
$\bal$, we can apply the privacy-amplification
Theorem~\ref{thm:classicalpa} with the strongly \univ\ hash function
$\bal(\hf_0(\cdot),\hf_1(\cdot))$ and argue that
$\bal(\hf_0(X_0),\hf_1(X_1))$ is close to uniform for randomly chosen
$\hf_0$ and $\hf_1$. Sender-security then follows immediately from
Theorem~\ref{thm:main}.

We save the quantitative analysis (Theorem~\ref{thm:UOT}) for next
section, where we consider a reduction of \OT to the weakest kind of
\pOT: to {\em one} execution of a \pUOT.  Based on this, we compare in
Section~\ref{sec:comparison} the quality of the analysis of the above
reductions based on Theorem~\ref{thm:main} with the results
in~\cite{BCW03}. It turns out that our analysis is tighter for \GOT 
and \BUOT, whereas the analysis in~\cite{BCW03} is tighter for \XOT;
but in all cases, our analysis is much simpler and, we believe, more
elegant.


\subsection{Reducing \lStringOT to One Execution of \pUOT}\label{sec:UOT}

In this section, we use the definition and some elementary properties
of R\'enyi entropy introduced in Section~\ref{sec:conditionalrenyientropy}.

\subsubsection{Universal Oblivious Transfer}
Probably the weakest flavor of \pOT is the {\em Universal} \pOT 
(\pUOT) as it was introduced by Cachin in~\cite{Cachin98}, in that it gives the
receiver the most freedom in getting information on the string $X$.
Formally, for a finite set $\cX$ and parameters $\alpha > 1$
(allowing $\alpha = \infty$) and $r > 0$, an \UOT{\alpha}{r}{\cX} works
as follows: the sender inputs $x \in \cX$, and the receiver may choose
an arbitrary conditional probability distribution $P_{Y|X}$ with the
only restriction that for a uniformly distributed $X$ it must satisfy
$\H_{\alpha}(X|Y) \geq r$.
The receiver then gets as output $y$, sampled according to the
distribution $P_{Y|X}(\cdot|x)$, whereas the sender gets no
information on the receiver's choice for $P_{Y|X}$. Note that a \BUOT
is a limit case of this kind of \pUOT since ``\BUOT =
\UOT{1}{1}{\set{0,1}^2}''.

The crucial property of such an \pUOT is that the input is not restricted to two bits, but 
may be two bit-{\em strings}; this potentially allows to reduce \OT 
to {\em one} execution of a \pUOT, rather than to many independent
executions of the same primitive as for the 1-2 flavors of \pOT 
mentioned above. Indeed, following the design principle discussed in
Section~\ref{sec:reductions}, it is straightforward to come up with a
candidate protocol for \lStringOT which uses {\em one} execution
of a \UOT{\alpha}{r}{\cX} with $\cX = \set{0,1}^n \times \set{0,1}^n$.
The protocol is given in Figure~\ref{fig:otuot}, where $\UH$ is a
strongly \univ class of hash functions from $\set{0,1}^n$ to
$\set{0,1}^{\ell}$.

\begin{myfigure}{H}
 \begin{myprotocol}{\OTUOT$(c)$}
 \item $\S$ and $\R$ run \UOT{\alpha}{r}{\cX}: $\S$ inputs a random $x
   = (x_0,x_1) \in \cX = \set{0,1}^n \times \set{0,1}^n$, $\R$ inputs
   $P_{Y|X}$ with $P_{Y|X}(x_c'|(x_0',x_1')) = 1$ for any
   $(x'_0,x'_1)$, and as a result $\R$ obtains $y = x_c$.
  \item $\S$ samples independent random $\hf_0,\hf_1 \in \UH$, sends $\hf_0$ and $\hf_1$ to $\R$, and outputs $s_0 = \hf_0(x_0)$ and $s_1 = \hf_1(x_1)$. 
  \item $\R$ computes and outputs $s_c = \hf_c(y)$. 
 \end{myprotocol}
\caption{Protocol \OTUOT\ for 
\Rand\lStringOT. }\label{fig:otuot}
\end{myfigure}

In \cite{Cachin98} it is claimed that, for appropriate parameters,
protocol \OTUOT\ is a secure \Rand\: \lStringOT, respectively, the
resulting protocol for \OT is secure. However, we argue below that
the proof given is not correct and it is not obvious how to fix it.
In Theorem~\ref{thm:UOT} we then show that its security follows easily
from Theorem~\ref{thm:main}.

\subsubsection{A Flaw in the Security Proof}

In \cite{Cachin98} the security of protocol \OTUOT\ is argued as
follows.  Using rather complicated {\em spoiling-knowledge
  techniques}, it is \index{spoiling-knowledge} shown that,
conditioned on the receiver's output (which we suppress to simplify
the notation) at least one out of $\H_{\infty}(X_0)$ and
$\H_{\infty}(X_1|X_0 \= x_0)$ is ``large'' (for any~$x_0$), and,
similarly, at least one out of $\H_{\infty}(X_1)$ and
\mbox{$\H_{\infty}(X_0|X_1 \= x_1)$}.  Since collision entropy is
lower bounded by min-entropy, it then follows from the privacy
amplification theorem that at least one out of $\H(\Hf_0(X_0)|\Hf_0)$
and $\H(\Hf_1(X_1)|\Hf_1, X_0 \= x_0)$ is close to $\ell$, and
similarly,
one out of $\H(\Hf_1(X_1)|\Hf_1)$ and $\H(\Hf_0(X_0)|\Hf_0, X_1 \=
x_1)$. It is then claimed that this proves \OTUOT\ secure.

We argue that this very last implication is not correct. Indeed, what
is proven about the entropy of $\Hf_0(X_0)$ and $\Hf_1(X_1)$ does not
exclude the possibility that both entropies $\H(\Hf_0(X_0)|\Hf_0)$ and
$\H(\Hf_1(X_1)|\Hf_1)$ are maximal, but that
$\H(\Hf_0(X_0)\oplus\Hf_1(X_1)|\Hf_0,\Hf_1) = 0$. This would allow the
receiver to learn the bit wise XOR $S_0 \oplus S_1$, which is clearly
forbidden by the condition of sender-security.

Also note that the proof does not use the fact that the two functions
$\Hf_0$ and $\Hf_1$ are chosen {\em independently}. However, if they
are chosen to be the same, then the protocol is clearly insecure: if
the receiver asks for $Y = X_0 \oplus X_1$, and if $\UH$ is a class of
{\em linear} \univ hash functions, then $\dR$ obviously learns $S_0
\oplus S_1$.

\subsubsection{Reducing \lStringOT to \pUOT}

The following theorem guarantees the security of \OTUOT\ for an
appropriate choice of the parameters. The only restriction we have to
make is that $\UH$ needs to be a {\em strongly} \univ class of hash
function.

\begin{theorem}\label{thm:UOT}
  Let $\UH$ be a {\em strongly} \univ class of hash functions from
  $\set{0,1}^n$ to $\set{0,1}^{\ell}$. Then \OTUOT reduces a
  $2^{-\sp}$-secure \RandlStringOT to a perfect
  \UOT{2}{r}{\set{0,1}^{2n}} with $n \geq r \geq 4 \ell + 2 \sp + 1$.
\end{theorem}
Using the bounds from Lemma~\ref{lemma:bounds} on the different orders
of R\'enyi entropy, the reducibility of \lStringOT to
\UOT{\alpha}{r}{\cX} follows immediately for {\em any} $\alpha > 1$.

Informally, sender-security of the protocol \OTUOT\ is argued as for
the reduction of \OT to \RabinOT, \XOT etc., discussed in
Section~\ref{sec:reductions}, simply by using
Proposition~\ref{prop:balanced->u2} in combination with the privacy
amplification Theorem~\ref{thm:classicalpa}, and applying
Theorem~\ref{thm:main}.  The formal proof given below additionally
keeps track of the error term.

From this proof it also becomes clear that the exponential (in $\ell$)
overhead in Theorem~\ref{thm:main} is unavoidable. Indeed, a
sub-exponential overhead would allow $\ell$ in Theorem~\ref{thm:UOT}
to be super-linear in $n$, which of course is nonsense.

\begin{proof}
  By the definition of conditional collision entropy, we have that for
  all~$y$,
$\H_2(X|Y\=y) \geq r \geq 4 \ell + 2\sp +1$. Fix
  an arbitrary $y$ and consider any NDLF \index{non-degenerate linear function}
  $\bal:\set{0,1}^{\ell}\times\set{0,1}^{\ell}\rightarrow\set{0,1}$.
  Let $\Hf_0$ and $\Hf_1$ be the random variables that represent the
  random choices of $\hf_0$ and $\hf_1$, and set $B =
  \bal(\Hf_0(X_0),\Hf_1(X_1))$. In combination with
  Proposition~\ref{prop:balanced->u2}, privacy amplification
  (Corollary~\ref{thm:classicalpa}) guarantees that
$$
\dist{P_{B \Hf_0 \Hf_1|Y=y}, P_{\unif} P_{\Hf_0 \Hf_1|Y=y}} \leq
2^{-\frac12(\H_2(X|Y=y)+1)} \leq 2^{-\frac12(4\ell+2\sp+2)} =
2^{-2\ell-\sp-1}.
$$
It now follows that
\begin{align*}
  \dist{ P_{\bal(S_0,S_1) W} &, P_{\unif} \cdot P_{W} } = \dist{P_{B \Hf_0 \Hf_1 Y},P_{\unif} P_{\Hf_0 \Hf_1 Y}} \\
  &= \sum_y \dist{P_{B \Hf_0 \Hf_1|Y=y},P_{\unif} P_{\Hf_0 \Hf_1|Y=y}}
  P_Y(y) \leq 2^{-\sp} / 2^{2\ell+1} \, .
\end{align*}
Sender-security as claimed now follows from Theorem~\ref{thm:main}.
\end{proof}

The min-entropy splitting Lemma~\ref{lemma:ESL} and a larger (not
necessarily strongly) two-universal class of hash functions can
alternatively be used to show the security of the reduction protocol
\OTUOT\ without the use of \index{non-degenerate linear function}
NDLFs. We do this here for illustration
purposes because the same technique is used in the security proof of
\OT in the bounded-quantum-storage model in Chapter~\ref{chap:12OT}.
After the execution of a perfect \UOT{\infty}{r}{\set{0,1}^{2n}}, we
have $\hmin(X_0 X_1|Y) \geq r$ and Lemma~\ref{lemma:ESL} yields the
existence of a random variable $D \in \set{0,1}$ such that
$\hmin(X_{1-D} D | Y) \geq r/2$ and therefore also $\hmin(X_{1-D} D
  S_D | Y) \geq r/2$. By the chain rule (Lemma~\ref{lem:chain}) and
setting $\eps \assign 2^{-\sp-1}$, we get $\hiee{X_{1-D} | D S_D Y}
\geq r/2 -1 - \ell - \sp -1$. Hence to get a $2^{-\sp}$-secure
\RandlStringOT via the privacy amplification theorem
(Corollary~\ref{thm:pasmooth}), we need $r/2 - \ell - \sp -2 > 2
\kappa + \ell$ which gives slightly worse parameters than in
Theorem~\ref{thm:UOT}, namely $n \geq r \geq 4 \ell + 4 \sp +4$.

\subsection{Quantitative Comparisons To Related Work}\label{sec:comparison}
\index{entropy!splitting lemma|see {min-entropy splitting lemma}}
Subsequent to \cite{DFSS06}, Wullschleger improved the min-entropy
splitting technique \index{min-entropy splitting lemma} described in
the last paragraph. In \cite{Wullschleger07}, it is shown that the
protocol \OTUOT\ reduces a $2^{-\sp}$-secure \RandlStringOT to a
perfect \UOT{\infty}{r}{\set{0,1}^{2n}} if $n \geq r \geq 2 \ell + 6
\sp + 6 \log(3)$. So, \RandlStringOT of strings of length $\ell$
roughly half of the receivers min-entropy $r$ can be obtained, which
is asymptotically optimal for this reduction-protocol. Technically,
the result is essentially obtained by using the min-entropy splitting
approach sketched at the end of last section and a more careful case
distinction. The random variable $D \in \set{0,1}$ pointing to the
``known'' string $X_D$ is basically defined as in
Lemma~\ref{lemma:ESL}, but for the case when both $X_0,X_1$ have high
min-entropy, a new \emph{distributed} left-over hash lemma is used to show
that both $S_0$ and $S_1$ are close to uniform and therefore close to
independent (and hence, the pointer $D$ can be chosen arbitrarily in
this case).

\vspace{3mm}

\def\cl{c_{\text{len}}}
\def\ck{c_{\text{sec}}}
\def\c0{c_{\text{const}}}

In the following, we compare the simple reduction of \lStringOT to $n$
executions of \XOT, \GOT and \BUOT, respectively, using our analysis
based on Theorem~\ref{thm:main} together with the quantitative statement
given in Theorem~\ref{thm:UOT}, with the results achieved
in~\cite{BCW03}.\footnote{As mentioned earlier, these results are
  incomparable to the parameters achieved in \cite{CS06}, where
  \emph{interactive} reductions are used.} The quality of the analysis
of a reduction is given by the {\em reduction parameters} $\cl$, $\ck$
and $\c0$ such that the \lStringOT is guaranteed to be
$2^{-\sp}$-secure as long as $n \geq \cl \cdot \ell + \ck \cdot \sp +
\c0$. The smaller these constants are, the better is the analysis of
the reduction. The comparison of these parameters is given in
Figure~\ref{fig:comparison}. We focus on $\cl$ and $\ck$ since $\c0$
is not really relevant, unless very large.

\begin{myfigure}{H}
\renewcommand{\arraystretch}{1.5}
\setlength{\tabcolsep}{2ex}
\begin{tabular}{l|cc|cc|cc|}
            & \multicolumn{2}{c|}{\XOT} & \multicolumn{2}{c|}{\GOT} & \multicolumn{2}{c|}{\BUOT} \\[-1.5ex]
            & $\cl$ & $\ck$& $\cl$ & $\ck$& $\cl$ & $\ck$ \\ \hline
BCW \cite{BCW03}&  2 &  2 & 4.8 & 4.8 & 14.6 & 14.6 \\ 
this work~\cite{DFSS06}   &  4 &  2 & 4   & 3   & 13.2 & 10.0 \\ 
subsequent \cite{Wullschleger07} & 2 & 6 & 2 & 7 & 6.7 & 23.3 \\ \hline
\end{tabular}
\caption{Comparison of the reduction parameters. }\label{fig:comparison}
\vspace{2ex}
\end{myfigure}

The parameters in the first line can easily be extracted from
Theorems~5,~7 and~9 of~\cite{BCW03}, where in Theorem~9 $p_e \approx
0.19$. The parameters in the second line corresponding to the
reduction to \XOT follow immediately from Theorem~\ref{thm:UOT},
using the fact that in {\em one} execution of a \XOT, the receiver's
conditional collision entropy on the sender's two input bits is at
least~$1$.

Determining the parameters of the reductions to \GOT and \BUOT requires
a little more work. We first determine the \emph{average} conditional
\index{entropy!average conditional min-}
min-entropy $\tH_{\infty}(X|Y)$ of one instance of \GOT and \BUOT. In
the case of \GOT, $\tH_{\infty}(X|Y)$ can easily be seen to be at
least 1 (for example by inspection of Table 2 in \cite{BCW03}). For
one execution of \BUOT, the receiver's average Shannon entropy is at
least $1$. Therefore, it follows from Fano's Inequality
(Lemma~\ref{lem:fano}) that his average guessing probability is at
most $1 - p_e$ with $p_e \approx 0.19$ as above, and thus his average
conditional min-entropy is at least $-\log (1-p_e) \approx
0.3$.

We use Lemma~\ref{lemma:average} to lower bound the (regular)
conditional min-entropy $\hmin(X|Y=y)$ except with probability
$2^{-\sp-1}$ and use Theorem~\ref{thm:UOT} with security parameter
$2^{-\sp-1}$ which together yields a $2^{-\sp}$ secure \RandlStringOT.
To apply Theorem~\ref{thm:UOT}, we require $\H_2(X|Y=y) \geq
\hmin(X|Y=y) \geq 4 \ell + 2 \sp + 3$ and to obtain this by
Lemma~\ref{lemma:average}, we need $\tH_{\infty}(X|Y) \geq 4 \ell + 3
\sp + 4$.

This yields $\cl=4, \ck=3$ for \GOT and $\cl \approx 4/0.3$ and $\ck
\approx 3/0.3$ for \BUOT. The derivation of the parameters for \cite{Wullschleger07} is analogous.

\section{Extension to \onenlStringOT}
In this section we extend our characterization of sender-security of
\RandOT to \onenRandOT.  We use the following notation. For a sequence
of random variables $S_0,S_1,\ldots,S_{n-1}$ and indices $i,j \in
\set{0,\ldots,n-1}$, we denote by $\overline{S_{i,j}}$ the sequence of
variables $\Set{S_k}{k \in \set{0,\ldots,n-1} \setminus \set{i,j}}$
with all indices except $i$ and $j$. Similarly, $\ol{S_i}$ denotes all
variables but the $i$th.

\begin{definition}[\onenRandlStringOT]\label{def:onenROT}
An $\varepsilon$-secure {\em \onenRandOT} is a protocol between $\S$ and $\R$, with $\R$ having
input $C \in \{0,1,\ldots,n-1\}$  (while $\S$ has no input), such that
for any distribution of $C$, the following properties hold:
\begin{description}
\item[\boldmath$\varepsilon$-Correctness:] For honest $\S$ and $\R$, $\S$
  has output $S_0,S_1,\ldots,S_{n-1} \in \{0,1\}^{\ell}$ and $\R$ outputs
  $S_C$, except with probability $\varepsilon$.
\index{correctness!of classical Rand@of classical \onenRandlStringOT}
\item[\boldmath$\varepsilon$-Receiver-security:] If $\R$ is honest then for any
  (possibly dishonest) $\dS$ with output $V$,
\[ \dist{ P_{CV} , P_C \cdot P_V } \leq \eps . 
\] 
\index{receiver-security!of classical Randl@of classical \onenRandlStringOT}
\item[\boldmath$\varepsilon$-Sender-security:] If $\S$ is honest then for any
  (possibly dishonest) $\dR$ with output $W$, there exists a random
  variable $D$ with range $\set{0,1,\ldots,n-1}$ such that
\[ \dist{ P_{\ol{S_D} W S_D D} , P_{\unif^{\ell}}^{n-1} \cdot P_{W S_D
  D} } \leq \eps.
\] 
\index{sender-security!of classical Randl@of classical \onenRandlStringOT}
\end{description}
\end{definition}
Analogous to the \OT-case we want for sender-security that there exists
a choice $D$, such that when given the corresponding string (or bit) $S_D$ all the
other strings (or bits) look completely random from $\R$'s point of view. 

Recall that for the characterization of sender-security in the case of
\OT, it is sufficient that $P_{\bal(S_0,S_1) W} = P_{\unif} \cdot
P_{W}$ for every \index{non-degenerate linear function} NDLF $\bal$.
In a first attempt one might try to characterize the sender-security
of \onenOT using linear functions $\bal$ that non-trivially depend on
$n$ arguments. In the case of \onethreeOT of bits, the only linear
function of this kind is the XOR of the three bits, but it can be
easily verified that the requirement that $B_0 \oplus B_1 \oplus B_2$
is uniform does \emph{not} imply sender-security in the sense defined
above. Instead, as we will see below, sufficient requirements are that
the XOR of \emph{every pair of bits} is uniform {\em when given the
  value of the third}.

\index{sender-security!characterization of}
\begin{theorem}\label{thm:onen}
  The condition for $\varepsilon$-sender-security  for a \onenRandlStringOT 
  is satisfied for a particular (possibly dishonest) receiver $\dR$
  with output $W$, if 
for all $i\neq j \in \set{0,\ldots,n-1}$ 
$$
\dist{ P_{\bal(S_i,S_j) W \ol{S_{i,j}}} , P_{\unif} \cdot P_{W
    \ol{S_{i,j}}} } \leq \nu 
$$
for every NDLF $\bal$, where $\nu = \varepsilon/(2^{2\ell}n(n-1))$. 
\end{theorem}
\begin{proof}
We first consider and prove the perfect case. 

\paragraph{\sc The Perfect Case: }

Like in the proof of Theorem~\ref{thm:main}, we fix an output $w$ of the
receiver and consider the non-normalized probability distribution
$P_{S_0 \ldots S_{n-1}W}(\cdot,\ldots,\cdot,w)$. We use the variable
$p_{s_0,\ldots,s_{n-1}}$ to refer to the value $P_{S_0 \ldots S_{n-1}
  W}(s_0,\ldots,s_{n-1},w)$ and $\zero$ for the all-zero string
$(0,\ldots,0) \in \set{0,1}^{\ell}$. We use bold font to
denote a collection of strings $\bs \assign
(s_0,s_1,\ldots,s_{n-1}) \in \set{0,1}^{\ell n}$, and we write 
$\overline{\bs_i}$ for $(s_0,\ldots,s_{i-1},s_{i+1},\ldots,s_{n-1})$, 
the collection $\bs$ without $s_i$. Finally, for a collection 
$\b{t} = (t_0,\ldots,t_{k-1}) \in \set{0,1}^{\ell k}$ of arbitrary
size $k$, we define sets of indices with one (respectively two) non-zero
substrings:
\begin{align*}
  \Sone[\b{t}] &\assign
  \Set{(\zero,\ldots,\zero,t_i,\zero,\ldots,\zero)}{i \in
    \set{0,\ldots,k-1}}\\ 
  \Stwo[\b{t}] &\assign
  \Set{(\zero,\ldots,\zero,t_i,\zero,\ldots,\zero,t_j,\zero,\ldots,\zero)}{i
    < j \in \set{0,\ldots,k-1}}
\end{align*}
where the $t_i$ (and $t_j$) are at $i$th (and $j$th) position. 
%
As in the proof of Theorem~\ref{thm:main}, we assume for the clarity
of exposition that for all $i \in \set{0,\ldots,n-1}$ and $s_i \in
\set{0,1}^{\ell}$, it holds that $p_{\zero,\ldots,\zero} \leq
p_{\zero,\ldots,\zero,s_i,\zero,\ldots,\zero}$ (where $s_i$ is at
position~$i$). For symmetry reasons, the general case can be handled
along the same lines.

We extend the distribution $P_{S_0 \ldots
  S_{n-1}W}(\cdot,\ldots,\cdot,w)$ similarly to \eqref{eq:extension}:
for every $\bs \in \set{0,1}^{\ell n}$, we set
\begin{align*}
P_{S_0 \ldots S_{n-1}DW}(s_0,\ldots,s_{n-1},0,w) &\assign p_{s_0,\zero,\ldots,\zero}, \\
P_{S_0 \ldots S_{n-1}DW}(s_0,\ldots,s_{n-1},1,w) &\assign
p_{\zero,s_1,\zero,\ldots,\zero} - p_{\zero,\ldots,\zero}, \\
& \;\;\,\vdots \\
P_{S_0 \ldots S_{n-1}DW}(s_0,\ldots,s_{n-1},n-2,w) &\assign p_{\zero,\ldots,s_{n-2},\zero}- p_{\zero,\ldots,\zero}, \\
P_{S_0 \ldots S_{n-1}DW}(s_0,\ldots,s_{n-1},n-1,w) &\assign p_{\zero,\ldots,\zero,s_{n-1}}- p_{\zero,\ldots,\zero}.
\end{align*}
In order to show that this is a valid extension, we have to show that
for every $\bs \in \set{0,1}^{\ell n}$
\begin{equation} \label{eq:ssum}
p_{\bs}= \sum_{\b{t} \in \Sone} \!\! p_{\b{t}} - (n-1)p_{\zero,\ldots,\zero}.
\end{equation}
If this holds, then the random variable $D$ is well defined, and  
the $\ol{S_D}$ are uniformly distributed given $D,S_D$ and $W$.

We now show that \eqref{eq:ssum} follows from the assumed uniformity
property that $P_{\bal(S_i,S_j) W | \ol{S_{i,j}} \!=\!\ol{\bs_{i,j}}} =
P_{\unif} \cdot P_{W | \ol{S_{i,j}}\!=\!\ol{\bs_{i,j}}}$ for every
non-degenerate linear function $\bal$ and any $i \neq j$. This is done
by induction on $n$. The case $n = 2$ is covered by the proof of
Theorem~\ref{thm:main}, and by induction assumption we may assume that
it also holds for $n-1$.  Let us fix some $\bs \in \set{0,1}^{\ell n}$
and $i \in \set{0,\ldots,n-1}$. It is easy to see that the assumed
uniformity property on $S_0,\ldots,S_{n-1},W$ implies the
corresponding uniformity property on $\overline{S_i},W$ when
conditioning on $S_i = s_i$, and therefore, by induction assumption
and ``multiplying out the conditioning",
\begin{equation}
p_{\bs} = \sum_{\b{t}} p_{\b{t}} -
(n-2)p_{\zero,\ldots,\zero,s_i,\zero,\ldots,\zero} \,. \label{eq:texample}
\end{equation}
where the sum is over all $\b{t} \in \set{0,1}^{\ell n}$ with $t_i =
s_i$ and $\overline{\b{t}_i} \in {\cal S}_1(\overline{\b{s}_i})$.
Summing all the equations over $i \in
\set{0,\ldots,n-1}$ yields
\begin{equation} \label{eq:firststep}
n \cdot p_{\bs} = 2 \!\! \sum_{\b{t} \in \Stwo} \!\! p_{\b{t}} - (n-2) \!\!
\sum_{\b{t} \in \Sone} \!\! p_{\b{t}} \,.
\end{equation}

By a similar reasoning we can also derive
from the case $n=2$ that equations of type \eqref{eq:original} hold
conditioned on the event that all but two of the $S_i$'s are zero. More
formally, we have that for all $i < j \in \set{0,\ldots,n-1}$,
\begin{equation}\label{eq:zeros}
p_{\zero,\ldots,\zero,s_i,\zero,\ldots,\zero,s_j,\zero,\ldots,\zero} =
p_{\zero,\ldots,\zero,s_i,\zero,\ldots,\zero} +
p_{\zero,\ldots,\zero,s_j,\zero,\ldots,\zero} - p_{\zero,\ldots,\zero}.
\end{equation}
Summing these equations over all $i < j \in \set{0,\ldots,n-1}$
yields
\begin{equation} \label{eq:twocase}
\sum_{\b{t} \in \Stwo} p_{\b{t}} = (n-1) \!\! \sum_{\b{t} \in \Sone} \!\! p_{\b{t}} - {n \choose 2} p_{\zero,\ldots,\zero}
\end{equation}
We conclude by substituting (\ref{eq:twocase}) into
(\ref{eq:firststep}) as follows
\begin{align*}
n \cdot p_{\bs} &= 2 \!\! \sum_{\b{t} \in \Stwo} \!\! p_{\b{t}} - (n-2) \!\! \sum_{\b{t} \in \Sone}
\!\! p_{\b{t}}\\
&= 2 \left( (n-1) \!\! \sum_{\b{t} \in \Sone} \!\! p_{\b{t}} - {n \choose 2}
  p_{\zero,\ldots,\zero} \right) - (n-2) \!\! \sum_{\b{t} \in \Sone} \!\! p_{\b{t}}\\
&= n \!\! \sum_{\b{t} \in \Sone} \!\! p_{\b{t}} - n (n-1) p_{\zero,\ldots,\zero},
\end{align*}
which is equation~(\ref{eq:ssum}) after dividing by $n$, and thus finishes
the induction step and the claim for $\varepsilon = 0$.

\paragraph{\sc The General Case: }

For the non-zero error case, we follow the above argument, but keep
track of the error.  For technical reasons, we assume that the $S_i$'s
are independent and uniformly distributed, and we assume that the
assumed uniformity property with respect to \index{non-degenerate linear function}
NDLFs holds conditioned on
$\overline{S_{i,j}}=\overline{\b{s}_{ij}}$ for {\em any}
$\overline{\b{s}_{ij}}$, not just on average, i.e., $ \dist{ P_{\bal(S_i,S_j) 
W | \ol{S_{i,j}} = \ol{\b{s}_{ij}}} , P_{\unif} \cdot P_{W |
  \ol{S_{i,j}}=\ol{\b{s}_{ij}} } } \leq \nu$ for any
$\overline{\b{s}_{ij}} \in \set{0,1}^{\ell(n-2)}$.  We show at the end
of the proof how to argue in general.  Write
$$
\delta_{\b{s}} = \big|\sum_{\b{t} \in \Sone} \!\! p_{\b{t}} - (n-1)p_{\zero,\ldots,\zero} - p_{\bs}\big|
$$ 
such that~\eqref{eq:ssum} holds up to the error $\delta_{\b{s}}$. Note that $\delta_{\bs}$ depends on $w$; we also write $\delta_{\bs}(w)$ to make this dependency explicit. We will argue, following the induction proof, that 
$$
\sum_{w,\bs} \delta_{\bs}(w) \leq n(n-1)\cdot2^{2\ell} \cdot \nu = \varepsilon \, .
$$ 
The proof can then be completed analogue to the proof of Theorem~\ref{thm:main} by ``correcting" the values for $P_{S_0 \ldots S_{n-1}DW}$'s appropriately. 

By the proof of Theorem~\ref{thm:main}, the claimed inequality holds in case $n = 2$. For the induction step, note that by induction assumption,~\eqref{eq:texample} holds up to $\delta_{\overline{\b{s}_i}}(w) P_{S_i}(s_i)$ where 
$$
\sum_{w,\overline{\b{s}_i}} \delta_{\overline{\b{s}_i}}(w) \leq (n-1)(n-2)\cdot2^{2\ell} \cdot \nu \, .
$$
Furthermore, from the case $n=2$ it follows that Equation~\eqref{eq:zeros} holds up to $\delta_{s_i,s_j}(w) P_{\overline{S_{ij}}}(\zero\cdots\zero)$, where 
$$
\sum_{w,s_i,s_j} \delta_{s_i,s_j}(w) \leq 2^{2\ell+1}\cdot\nu
$$ 
and, by the additional assumption posed on the $S_i$'s, $P_{\overline{S_{ij}}}(\zero\cdots\zero) = 2^{-(n-2)\ell}
$.
It follows that~\eqref{eq:ssum} holds up to 
$$
\delta_{\bs} = \frac{1}{n}\Big(\sum_i\delta_{\overline{\b{s}_i}}P_{S_i}(s_i)+2\sum_{i<j}\delta_{s_i,s_j}P_{\overline{S_{ij}}}(\zero\cdots\zero)\Big)
$$ 
such that 
\begin{align*}
\sum_{w,\bs} \delta_{\bs}(w) &= \frac{1}{n}\Big( \sum_i \sum_{w,\overline{\bs_i}}\delta_{\overline{\b{s}_i}}(w) \sum_{s_i}P_{S_i}(s_i) + 2\sum_{i<j} \sum_{\overline{\bs_{ij}}}\sum_{w,s_i,s_j}\delta_{s_i,s_j}(w) P_{\overline{S_{ij}}}(\zero\cdots\zero) \Big) \\
&\leq (n-1)(n-2)\cdot2^{2\ell} \cdot \nu + (n-1) \cdot 2^{(n-2)\ell} \cdot 2^{2\ell+1} \cdot 2^{-(n-2)\ell} \cdot \nu \\[1ex]
&= \big( (n-1)(n-2) \cdot2^{2\ell} + 2\cdot (n-1) \cdot 2^{2\ell} \big) \cdot \nu \\[1ex]
&\leq n(n-1) \cdot 2^{2\ell} \cdot \nu = \varepsilon \, .
\end{align*}

It remains to argue the case where the $S_i$'s are not independent
uniformly distributed and/or the assumed uniformity property holds
only on average over the $\overline{\b{s}_{ij}}$'s.  We first argue
that we may indeed assume without loss of generality that the $S_i$'s
are random: We consider $\tilde{S}_0,\ldots,\tilde{S}_{n-1},\tilde{W}$
defined as $\tilde{S}_i = S_i \oplus R_i$ and $\tilde{W} =
[W,R_0,\ldots,R_{n-1}]$ for independent and uniformly distributed
$R_i$'s in $\set{0,1}^{\ell}$.  It is easy to see that the assumed
uniformity condition with respect to \index{non-degenerate linear function}
NDLFs on $S_0,\ldots,S_{n-1},W$
implies the corresponding uniformity condition on
$\tilde{S}_0,\ldots,\tilde{S}_{n-1},\tilde{W}$ with the same ``error"
$\nu$, and it is obvious that the $\tilde{S_i}$'s are independent and
uniformly distributed.  Furthermore, it is easy to see that
$\varepsilon$-sender-security for
$\tilde{S}_0,\ldots,\tilde{S}_{n-1},\tilde{W}$ implies
$\varepsilon$-sender-security for $S_0,\ldots,S_{n-1},W$ with the same
$\varepsilon$. Thus it suffices to prove the claim for the case of
random $S_i$'s.

Finally, in order to reason that we may assume that the uniformity
property holds conditioned on every $\overline{\b{s}_{ij}}$, where we
now may already assume that the $S_i$'s are random due to the above
observation, we again consider
$\tilde{S}_0,\ldots,\tilde{S}_{n-1},\tilde{W}$ defined as above. It is
not hard to verify that due to this randomization and since the
$S_i$'s are random, the average near-uniformity of $\bal(S_i,S_j)$
translates to a ``worst-case" near-uniformity of
$\bal(\tilde{S}_i,\tilde{S}_j)$ with the same $\nu$.
\end{proof}

\section{\OT in a Quantum Setting} \label{sec:quantumdoesnotwork}
As briefly mentioned in the introductory
Section~\ref{sec:introtoNDLF}, the results of this chapter were
originally motivated by the idea of using them to prove
sender-security in the bounded-quantum-storage model of the
\OT-protocol presented later in Chapter~\ref{chap:12OT}.  For this
protocol, we can use a quantum \index{uncertainty relation}
uncertainty relation to show a lower bound on the min-entropy of the
$n$-bit string $X$ transmitted by the sender using a quantum encoding.

If we had a quantum version of Theorem~\ref{thm:main} at hand, we
could use privacy amplification against quantum adversaries
(Theorem~\ref{thm:pasmooth}) to prove sender-security against
quantum-memory-bounded receivers. Unfortunately, the example below
shows that such a quantum version of Theorem~\ref{thm:main} cannot
exist.

In the case of a dishonest quantum receiver \dR, the final state of a quantum
protocol for \RandOT is given by the ccq-state $\rho_{S_0 S_1 \dR}$.
The condition for $\varepsilon$-sender-security given in
Definition~\ref{def:Rl12OT} requires the existence of a random
variable $D \in \set{0,1}$ such that
$$
  \dist{ \rho_{S_{1-D} S_D D \dR} , \id \otimes \rho_{S_D
      D \dR} } \leq \eps \, .
$$
This coincides with the classical Definition~\ref{def:RandOT},
except that the dishonest receiver's output is a quantum state, and
closeness is measured in terms of the trace-norm distance.

A quantum analogue of Theorem~\ref{thm:main} would state that this
condition is fulfilled if for every \index{non-degenerate linear function}
NDLF $\bal$,
$$
  \dist{ \rho_{\bal(S_0,S_1) \dR} , \id \otimes \rho_{\dR} } \leq
  \eps' \, 
$$
where $\eps'$ is comparable to the classical parameter $\eps/2^{2
\ell+1}$.

Consider now the following example for \OT of bits $B_0,B_1$. We define the
ccq-state $\rho_{B_0 B_1 \dR}$ as follows: Let
\[ \begin{split}
\rho_{B_0 B_1 \dR} \assign \frac{1}{4} \bigl( &\proj{00} \otimes
  \proj{0} \, + \, \proj{11} \otimes \proj{1} \\ 
+ &\proj{01} \otimes \proj{+} \, + \, \proj{10}
  \otimes \proj{-} \bigr), 
\end{split} \]
where $\proj{+}$ and $\proj{-}$ are the projectors onto the states $\ket{+} \assign
\ket{0}_\times = \frac{\ket{0}+\ket{1}}{\sqrt{2}}$ and $\ket{-}
\assign \ket{1}_\times= \frac{\ket{0}-\ket{1}}{\sqrt{2}}$.

For this state, it is clear that the XOR $B_0 \oplus B_1$ is perfectly
hidden from the dishonest receiver holding $\rho_{\dR}$, i.e.
\[
\dist{ \rho_{(B_0 \oplus B_1) \dR} , \id \otimes \rho_{\dR} } = 0 \, .
\]
On the other hand, $\dR$ can determine the bit of his choice by
measuring in the Breitbart basis $\set{\cos(\pi/8) \ket{0} +
  \sin(\pi/8) \ket{1}, \, \sin(\pi/8) \ket{0} - \cos(\pi/8) \ket{1}}$
if he is interested in the first bit, or by measuring in the Breitbart
basis rotated by 45 degrees if he wants to obtain the second bit. It
is easy to see that such a measurement succeeds in yielding the
correct bit with probability $\cos(\pi/8)^2 \approx 0.85$. This
precludes the existence of a pointer variable $D \in \set{0,1}$ such
that perfect sender-security in the sense of
Definition~\ref{def:Rl12OT} holds.

It is unclear how that difficulty can be overcome, but it is clear
from the simple example above, that a statement like in
Theorem~\ref{thm:main} with comparable parameters cannot hold.
Therefore, the alternative approach via the 
\index{min-entropy splitting lemma}entropy-splitting
Lemma~\ref{lemma:ESL} (outlined at the end of Section~\ref{sec:UOT})
will be taken in Chapter~\ref{chap:12OT} to show sender-security.

\index{sender-security!of classical \OT|)}



%% file: uncert.tex
\clearemptydoublepage
\chapter{Quantum Uncertainty Relations} \label{chap:uncertrelations} 
Quantum \index{uncertainty relation }uncertainty relations are the
fundamental tool for the security analysis of protocols in the
\index{bounded-quantum-storage model}bonded-quantum-storage model
presented later in this thesis. We start off with some preliminary
tools in Section~\ref{uncert:sec:prelim} and proceed to the history of
uncertainty relations in Section~\ref{sec:uncerthistory}. Then, we
derive new high-order entropic uncertainty relations for two
(Section~\ref{sec:twounbiasedbases}) and more
(Section~\ref{sec:moreunbiasedbases}) mutually unbiased bases. In the
last Section~\ref{sec:morerelation}, we investigate the situation
where for each qubit, a basis is picked independently at random from a
set of bases.

The results in this chapter are based on joint work with Damg{\aa}rd,
Fehr, Salvail and Renner which appeared in \cite{DFSS08journal,DFRSS07}.

\section{Preliminaries} \label{uncert:sec:prelim}
\subsection{Operators and Norms} \label{sec:norms}
For a linear operator $A$ on the complex Hilbert space $\cH$, we
define the \emph{\index{operator norm}operator norm}
\[
\| A \| \assign \sup_{\braket{x}{x}=1} \|Ax\|
\]
for the \index{Euclidian norm}Euclidian norm $\|x\| \assign \sqrt{\braket{x}{x}}$ of the vector $\ket{x} \in \cH$. When $A$
is \index{Hermitian}Hermitian, i.e.~the complex conjugate transpose $H^*$ and $H$ coincide, we have
\[
\| A \| = \lambda_{\max}(A) \assign \max\{|\lambda_j| : \lambda_j
  \mbox{ an eigenvalue of } A \}.
\]
From an equivalent definition of the norm $\|A\|=
\!\!\!\sup\limits_{\braket{y}{y}=\braket{x}{x}=1}
\!\!\!|\bra{y}A\ket{x}|$, it is easy to see that $\|A^*\|=\|A\|$. For
two Hermitian matrices $A$ and $B$, we have that $\| AB \| =
\|(AB)^*\|=\|B^*A^*\|=\|BA\|$.  The operator norm is \emph{unitarily
  invariant}, i.e. for all unitary $U,V$, $\|A\|=\|UAV\|$ holds. It
is easy to show that
\[ \left\| \begin{pmatrix} A & 0 \\ 0 & B \end{pmatrix} \right\| =\max\left\{\|A\|,\|B\|\right\}.
\]

\begin{lemma} \label{lem:ineq}
Let $X$, $Y$ be any two $n \times n$ matrices such that the products $XY$ and $YX$
are Hermitian. Then, we have
\[ \|XY\| = \|YX\| \] 
\end{lemma}
\begin{proof}
For any two $n \times n$ matrices $X$ and $Y$, $XY$ and $YX$ have the
same eigenvalues, see e.g.\ \cite[Exercise I.3.7]{Bhatia97}. Therefore,
$\| XY \| = \lambda_{\max}(XY) = \lambda_{\max}(YX) = \|YX\|$.
\end{proof}
                                
A linear operator $P$ such that $P^2=P$ and $P^*=P$ is called an
\emph{\index{orthogonal projector}orthogonal projector}.
\index{projector!orthogonal|see {orthogonal projector}}
\begin{proposition} \label{prop:norm2}
Let $A$ and $B$ be two orthogonal projectors. Then it holds that $\|A+B\| \leq 1 + \|AB\|$.
\end{proposition}
\begin{proof}
We adapt a technique by Kittaneh \cite{Kittaneh97} to
our case. Define two $2 \times 2$-\index{block matrix}block matrices $X$ and $Y$ as follows
\[
X \assign \begin{pmatrix} A & B \\ 0 & 0 \end{pmatrix}
\quad \mbox{ and } \quad
Y \assign \begin{pmatrix} A & 0 \\ B & 0 \end{pmatrix} \, .
\]
Using $A^2=A$ and $B^2=B$, we compute
\[
XY \assign \begin{pmatrix} A+B & 0 \\ 0 & 0 \end{pmatrix}
\mbox{ and }
YX \assign \begin{pmatrix} A & AB \\ BA & B \end{pmatrix}
= \begin{pmatrix}  A & 0 \\ 0 & B \end{pmatrix}
+
\begin{pmatrix} 0 & AB \\ BA & 0 \end{pmatrix} \, .
\]
As $A$ and $B$ are \index{Hermitian}Hermitian, so are $A+B$, $AB$, $BA$, $XY$ and $YX$ as well.
We use Lemma~\ref{lem:ineq} and the triangle inequality to obtain
\[
\left\| \begin{pmatrix} A+B & 0 \\ 0 & 0 \end{pmatrix} \right\| 
= \left\| \begin{pmatrix} A & AB \\ BA & B \end{pmatrix} \right\| 
\leq \left\| \begin{pmatrix} A & 0 \\ 0 & B \end{pmatrix}  \right\| 
  + \left\| \begin{pmatrix} 0 & AB \\ BA & 0 \end{pmatrix} \right\|.
\]
Using the unitary invariance of the operator norm to permute the
columns in the rightmost matrix and the facts that $\|A\|=\|B\|=1$ as
well as $\|AB\|=\|BA\|$, we conclude that
\[
\| A+B \| \leq 1 + \|AB\|.
\]
\end{proof}

A nice feature of this block-matrix technique is that it generalizes
easily to more projectors.
\begin{proposition} \label{prop:morebases}
For orthogonal projectors $A_0, A_1, A_2, \ldots, A_M$, it holds that 
\begin{equation} \label{eq:multiproj} 
\left\| \sum_{i=0}^M A_i \right\| \leq 1 + M \cdot \max_{0\leq i< j
  \leq M}
\|A_iA_j\|.
\end{equation}
\end{proposition}
\begin{proof}
Defining
\[
X \assign \begin{pmatrix} A_0 & A_1 & \cdots & A_M
    \\ 0 & 0 & \cdots & 0 \\ \vdots & \vdots& & \vdots \\ 0 & 0 &
    \cdots & 0\end{pmatrix}
\quad \mbox{ and } \quad
Y \assign \begin{pmatrix} A_0 & 0 & \cdots & 0 \\
    A_1 & 0 & \cdots & 0 \\  \vdots & \vdots& & \vdots \\ A_M & 0 & \cdots & 0 \end{pmatrix}
\]
yields
\begin{align*}
XY &= \begin{pmatrix} A_0 +A_1 + \ldots + A_M & 0 & \cdots & 0
    \\ 0 & 0 & \cdots & 0 \\ \vdots & \vdots& & \vdots \\ 0 & 0 &
    \cdots & 0\end{pmatrix} \quad \mbox{ and}\\
YX &= \begin{pmatrix} A_0 & A_0 A_1 & \cdots &
    A_0 A_M \\  A_1 A_0 & A_1 & \cdots & A_1 A_M \\  \vdots & \vdots&
    \ddots& \vdots \\ A_M A_0 & A_M A_1 & \cdots & A_M \end{pmatrix}
\end{align*}
The matrix $YX$ can be additively decomposed into $M+1$ matrices
according to the following pattern
\[
YX= \begin{pmatrix} * &  &  &
     & \\   & * & &  &   \\   &  & \ddots & &  \\  &  & & * &  \\
    &  &  & & * \end{pmatrix}
+ \begin{pmatrix} 0 & * &  &
     & \\   & 0 & &  &   \\   &  & \ddots & \ddots
    &  \\  &  & & 0 & * \\
    * &  &  & & 0 \end{pmatrix}
+\;\ldots\;
+ \begin{pmatrix} 0 & &  &
     & * \\  * & 0 & &  &   \\   &\ddots  & \! \ddots &  &  \\  
    &  & & 0 &  \\
    &  &  & * & 0 \end{pmatrix}
\]
where the asterisk stand for entries of $YX$ and for $i=0,\ldots,M$
the $i$th asterisk-pattern after the diagonal pattern is obtained by
$i$ cyclic shifts of the columns of the diagonal pattern. Entries
without asterisk are zero.

As in the proof of Proposition~\ref{prop:norm2}, $XY$ and
$YX$ are \index{Hermitian}Hermitian and we use Lemma~\ref{lem:ineq}, the triangle
inequality, the unitary invariance of the operator norm and the facts
that for all $i \neq j : \|A_i\|=1$, $\|A_i A_j\|=\|A_j A_i\|$ to obtain the desired
statement \eqref{eq:multiproj}.
\end{proof}

\subsection{Azuma's Inequality}\label{sec:Azuma}
\index{Azuma's inequality} 
As we will exclusively use the concentration result at the end of this
section, we only give an informal definition of martingales. We refer
to \cite{AS00} or \cite{MR95} for a more detailed treatment.
\begin{definition} \index{martingale sequence} 
  A sequence of real random variables $X_0, X_1, \ldots$ is a
  martingale sequence, if for all $i=1,2,\ldots$, it holds $\E[X_i
  | X_0, \ldots, X_{i-1} ] = X_{i-1}$.
\end{definition}
\begin{theorem}[Azuma's inequality \cite{Azuma67}] \label{thm:azuma}
Let $X_0, X_1, \ldots$ be a martingale sequence such that for each
$k$, $|X_k - X_{k-1}| \leq c_k$, where $c_k$ may depend on $k$. Then,
for all $t \geq 0$ and any $\tau > 0$,
\[ \Pr[ X_t - X_0  \geq \tau ] \leq \exp \left(
  -\frac{\tau^2}{2 \sum_{k=1}^t c_k^2} \right) \, .
\]
\end{theorem}
The theorem is often stated as two-sided bound with absolute values:
\[ \Pr\big[ |X_t - X_0|  \geq \tau \big] \leq 2 \exp \left(
  -\frac{\tau^2}{2 \sum_{k=1}^t c_k^2} \right),
\]
but the one-sided version fits our purposes better.

\begin{definition} \index{martingale difference sequence}
A sequence of real-valued random variables $R_1,\ldots,R_n$ is called
a \emph{martingale difference sequence} if for every $i$ and every
$r_1,\ldots,r_{i-1} \in \reals$: 
\mbox{$
\E[R_i | R_1\!=\!r_1,\ldots,R_{i-1}\!=\!r_{i-1}] =0
$}.
\end{definition}
Note that for an arbitrary sequence of real random variables $S_0, S_1, \ldots
\in \reals$, defining $R_n \assign \sum_{i=1}^n S_i - \E[S_i | S^{i-1}]$ (with $R_0 \assign 0$) yields
a martingale difference sequence $R_0, R_1, \ldots$.

The following lemma follows directly from Azuma's Theorem~\ref{thm:azuma}.
\begin{corollary}\label{cor:azuma} \index{Azuma's inequality}
Let $R_1, \ldots, R_n$ be a martingale difference sequence such that
$|R_i| \leq c$ for every $1 \leq i \leq n$. Then, for any  $\lambda > 0$,
\[ \Pr\left[ \sum_i R_i \geq \lambda n \right] \leq
\exp{\bigl(-\frac{\lambda^2 n}{2 c^2}\bigr)}. \]
\end{corollary}
\begin{proof}
Set $\tau \assign \lambda n$, $X_0 \assign 0$, and for $n \geq 1$, $X_n \assign
\sum_{i=1}^n R_i$ in Theorem~\ref{thm:azuma}.
\end{proof}

\subsection{Mathematical Tools}
The following two purely analytical lemmas will be used to bound some error terms.
\begin{lemma} \label{lem:delta}
  For any $0 < x < 1/e$ such that $y \assign x \log(1/x) < 1/4$,
  it holds that $x > \frac{y}{4 \log(1/y)}$.
\end{lemma}
\begin{proof}
  Define the function $x \mapsto f(x) = x \log(1/x)$. It holds that
  $f'(x) = \frac{d}{dx}f(x) = \log(1/x)-\log e$, which shows that $f$
  is bijective in the interval $(0,1/e)$, and thus the inverse
  function $f^{-1}(y)$ is well defined for $y \in (0,\log(e)/e)$,
  which contains the interval $(0,1/4)$. We are going to show that
  $f^{-1}(y) > g(y)$ for all $y \in (0,1/4)$, where $g(y) = \frac{y}{4
    \log(1/y)}$. Since both $f^{-1}(y)$ and $g(y)$ converge to 0 for
  $y \rightarrow 0$, it suffices to show that $\frac{d}{dy} f^{-1}(y)
  > \frac{d}{dy} g(y)$; respectively, we will compare their
  reciprocals. For any $x \in (0,1/e)$ such that $y = f(x) = x
  \log(1/x) < 1/4$
$$
\frac{1}{\frac{d}{dy} f^{-1}(y)} = f'(f^{-1}(y)) = \log(1/x)-\log(e)
$$
and
$$
\frac{d}{dy} g(y) = \frac{1}{4} \bigg( \frac{1}{\log(1/y)} + \frac{1}{\ln(2) \log(1/y)^2} \bigg)
$$
such that 
\begin{align*}
\frac{1}{\frac{d}{dy} g(y)} &= 4 \, \frac{\ln(2)\log(1/y)^2}{\ln(2) \log(1/y) + 1} 
= 4 \, \frac{\log(1/y)}{1+\frac{1}{\ln(2) \log(1/y)}}\\
& > 2 \log\Big(\frac{1}{y}\Big) 
= 2 \log\Big(\frac{1}{x \log(1/x)}\Big) \\[0.7ex]
&= 2\big(\log(1/x) - \log\log(1/x)\big)
\end{align*} 
where for the inequality we are using that $y < 1/4$ so that $\ln(2) \log(1/y) > 2\ln(2) = \ln(4) > 1$. 
Defining the function
$$
h(z) \assign z - 2\log(z) + \log(e)
$$
and showing that $h(z) > 0$ for all $z>0$ finishes the proof, as then
$$
0 < h\big(\log(1/x)\big) \leq  \frac{1}{\frac{d}{dy} g(y)} - \frac{1}{\frac{d}{dy} f^{-1}(y)}
$$
which was to be shown. 
For this last claim, note that $h(z) \rightarrow \infty$ for $z
\rightarrow 0$ and for $z \rightarrow \infty$, and thus the global
minimum is at $z_0$ with $h'(z_0) = 0$. $h'(z) = 1 - 2/(\ln(2)z)$ and
thus $z_0 = 2/\ln(2) = 2\log(e)$, and hence the minimum of $h(z)$ equals
$h(z_0) = 3 \log(e) - 2\log\big(2\log(e)\big)$, which turns out to be positive. 
\end{proof}

\begin{lemma} \label{lem:epsilon}
 For any $0 < x < 1/4$, it holds that $\exp(-\frac{x^2 }{32
 (2-\log(x))^2}) < 2^{-x^4/32}$ .
\end{lemma}
\begin{proof}
  Note that $\exp(-\frac{x^2 }{32(2-\log(x))^2}) = 2^{-
    \frac{\log(e)}{32} \frac{x^2}{(2-\log(x))^2}}$. Therefore, it
  suffices to show that $x^4 \leq \frac{x^2}{(2-\log(x))^2}$ or
  equivalently that the function $x \mapsto f(x) \assign x^2
  (2-\log(x))^2$ is smaller than 1 for $0<x<1/4$. It holds that
  $f(0)=0$ and $f(1/4)=1$ and it is easy to see that $f$ is a
  continuous increasing function, e.g. by verifying that for the first
  derivative
\[ \frac{d}{dx}f(x)=2x \left(2-\log(x) \right) \left(2-\log(x) - \frac{1}{\ln(2)} \right) 
    >0  \] holds for $0<x<1/4$.
\end{proof}

\section{History and Previous Work} \label{sec:uncerthistory}
\subsection{Mutually Unbiased Bases}
\index{mutually unbiased bases}
\begin{definition}[Mutually Unbiased Bases (MUBs)]
Two orthonormal bases $\Ba{0} \assign \set{\ket{a_i}}_{i=1}^N$ and $\Ba{1}
  \assign \set{\ket{b_j}}_{j=1}^N$ of the complex Hilbert space
  $\cH_N$ of dimension $N\assign 2^n$ are called \emph{mutually unbiased} if
\[ \forall i, j \in \set{1, \ldots, N} : \left|
  \braket{a_i}{b_j} \right| = \frac{1}{\sqrt{N}}=2^{-n/2}.
\]
More $\Ba{0}, \Ba{1}, \ldots, \Ba{M}$ bases of this space
$\cH_N$ are called \emph{mutually unbiased}, if every pair of
them is mutually unbiased.
\end{definition}

Wiesner showed in 1970 in one of the first articles about quantum
cryptography \cite{Wiesner83} that there are at least $m$ mutually
unbiased bases in a Hilbert space of dimension $2^{(m-1)!/2}$. Later,
optimal constructions of $N+1$ mutually unbiased bases in a Hilbert
space of dimension $N$ were shown by Ivanovi\'c when $N$ is prime
\cite{Ivanovic81} and by Wootters and Fields for $N$ a prime power
\cite{WF89} (in particular, for $N=2^n$ in the case of $n$ qubits). A
nice construction based on the \index{stabilizer formalism}stabilizer
formalism can be found in
the article by Lawrence, Brukner, and Zeilinger \cite{LBZ02}. It
turned out to be an intriguing question to determine the maximal
number of mutually unbiased bases in other dimensions, already the
case $N=6$ is still open \cite{Englert03}.

For a density matrix $\rho$ describing the state of $n$ qubits, let
$Q_\rho^0(\cdot), Q_\rho^1(\cdot), \ldots, Q_\rho^M(\cdot)$ be the
probability distributions over $n$-bit strings when measuring $\rho$
in bases $\Ba{0}, \Ba{1}, \ldots, \Ba{M}$, respectively. For instance,
for basis $\Ba{0} = \set{\ket{a_i}}_{i=1}^N$ and basis $\Ba{1}
\set{\ket{b_j}}_{j=1}^N$, we have $Q_\rho^0(i) =
\bra{a_i}\rho\ket{a_i}$ and $Q_\rho^1(j) = \bra{b_j}\rho\ket{b_j}$. We
leave out the state $\rho$ in the subscript when it is clear from the
context.

\subsection{Uncertainty Relations Using Shannon Entropy}
\index{uncertainty relation|(}
The history of uncertainty relations starts with Heisenberg who showed
that the outcomes of two non-commuting observables applied to a
quantum state are not easy to predict simultaneously
\cite{Heisenberg27}. However, Heisenberg only speaks about the
variance of the measurement results, and his result was shown to have
several shortcomings by Deutsch \cite{Deutsch83} and Hilgevood and Uffink
\cite{HU88}. More general forms of uncertainty relations were proposed
by Bialynicki-Birula and Mycielski in \cite{BM75} and by Deutsch
\cite{Deutsch83} to resolve these problems.  The new relations were
called {\em entropic uncertainty relations}, because they are
expressed using \index{entropy!Shannon}Shannon entropy instead of the statistical variance.

For mutually unbiased bases, \index{Deutsch's relation}Deutsch's relation reads
\[
\H(Q^0) + \H(Q^1) \geq -2 \log{ \frac12 (1+\frac{1}{\sqrt{N}}) }.
\]
A much stronger bound was first conjectured by Kraus~\cite{Kraus87} and later proved
by \index{Maassen and Uffink's relation}Maassen and Uffink~\cite{MU88}
\begin{equation} \label{eq:maassenuffink}
\H(Q^0) + \H(Q^1) \geq \log{N}=n. 
\end{equation}
Intuitively, these bounds assure that if you know the outcome of
measuring $\rho$ in basis $\Ba{0}$ pretty well, you have large
uncertainty when measuring in the other basis $\Ba{1}$. 

Note that for entropic bounds using \emph{Shannon entropy}, it is
sufficient to state them for pure states. They then automatically hold
for mixed state by \index{concave function}concavity.
\begin{lemma}
  If $\H(Q_{\ket{\varphi}}^0) + \H(Q_{\ket{\varphi}}^1) \geq k$ holds
  for all pure states $\ket{\varphi} \in \cH$, then $\H(Q_{\rho}^0) +
  \H(Q_{\rho}^1) \geq k$ holds for all (possibly mixed) states $\rho \in \dens{\cH}$.
\end{lemma}
\begin{proof}
Let $\rho=\sum_x \lambda_x \proj{\varphi_x}$ the spectral composition
of a mixed state. We then have for $i=0,1$ that $Q_{\rho}^i=\sum_x
\lambda_x Q_{\ket{\varphi_x}}^i$ and therefore by concavity of the
Shannon entropy (Lemma~\ref{lem:concavity})
\[
\H(Q_{\rho}^0) + \H(Q_{\rho}^1) \geq \sum_x \lambda_x \left( \H(Q_{\ket{\varphi_x}}^0)
+ \H(Q_{\ket{\varphi_x}}^1) \right) \geq k.\]
\end{proof}

Although a bound on Shannon \index{entropy!Shannon}entropy can be
helpful in some cases, it is usually not good enough in cryptographic
applications.  The main tool to reduce the adversary's
information---\index{privacy amplification}privacy amplification by two-universal
hashing---requires a bound on the adversary's min-\index{entropy!min-}entropy (in fact
collision entropy), see Section~\ref{sec:pa}. As $H(Q) \geq
H_\alpha(Q)$ for $\alpha > 1$, higher-order entropic bounds are
generally weaker, but imply bounds for Shannon entropy as well.

\subsection{Higher-Order Entropic Uncertainty Relations}
Different results are known for
\emph{complete sets} of $N+1$ mutually unbiased bases of
$\cH_N$. All of them are based on the following surprising
geometrical result by Larsen.
\begin{theorem}[\cite{Larsen90}] \label{thm:larsen}
 Let $Q_\rho^0, \ldots, Q_\rho^N$ be the $N+1$ distributions obtained by
 measuring state $\rho$ in mutually unbiased bases $\Ba{0},
 \ldots, \Ba{N}$ of the Hilbert space $\cH_N$. Then, 
 \begin{equation} \label{eq:larsen}
    \sum_{i=0}^{N} \pi_2(Q_\rho^i) = 1 + \tr(\rho^2),
  \end{equation}
  where $\pi_2(Q)=\sum_x Q(x)^2$ denotes the \index{collision probability}
collision probability of a distribution $Q$ (cf.
  Definition~\ref{def:ordersum}).
\end{theorem}

For a pure state $\rho = \proj{\psi}$, $\tr(\rho^2)=1$ holds and the
right hand side of~\eqref{eq:larsen} equals 2. In this case, using
that $x \mapsto -\log(x)$ is a convex function,
S{\'a}nchez-Ruiz~\cite{Ruiz95} applies \index{Jensen's inequality}Jensen's inequality
(Lemma~\ref{lem:jensen}) to derive the following lower-bound on the
sum of the collision entropies
\begin{align*}
  \sum_{i=0}^N \H_2(Q^i) &= \sum_{i=0}^N -\log(\pi_2(Q^i))\\
&\geq -(N+1) \log\left(\frac{\sum_{i=0}^N \pi_2(Q^i)}{N+1}\right) = (N+1)
  \log \left( \frac{N+1}{2} \right).
\end{align*}
Because of the lack of convexity of higher-order R\'enyi entropy, we
cannot immediately extend an uncertainty relation for pure states to
mixed states. On the other hand, the following lemma shows that
uncertainty relations based on upper bounds of high-order \emph{probability
sums} for pure states also hold for mixed states and therefore
translate to entropy lower bounds for mixed states. 
\begin{lemma} \label{lem:sumbound}
Let $\alpha \in (1,\infty]$. If $\sum_{i=0}^M
\pi_\alpha(Q^i_{\ket{\varphi}}) \leq c$ for all pure states
$\ket{\varphi}$, then for all mixed states $\rho$,
$$\sum_{i=0}^M \H_\alpha(Q^i_\rho) \geq
(M+1) \log\left(\frac{M+1}{c} \right).$$
Equality holds for a state $\rho$ for which
$\pi_\alpha(Q^i_\rho)=\frac{c}{M+1}$ for all $i$.
\end{lemma}
\begin{proof}
As $x \mapsto x^{\alpha}$ is convex for $\alpha >1$,
$\pi_\alpha(\cdot)$ is a convex functional. Therefore, for a mixed state
$\rho = \sum_x \lambda_x \proj{\varphi_x}$, we have $Q^i_\rho=\sum_x
\lambda_x Q^i_{\ket{\varphi_x}}$ and
\[
\sum_{i=0}^M \pi_{\alpha}(Q^i_{\rho}) \leq \sum_{i=0}^M \sum_x
\lambda_x \pi_\alpha(Q^i_{\ket{\varphi_x}}) \leq \sum_x \lambda_x
\sum_{i=0}^M \pi_\alpha(Q^i_{\ket{\varphi_x}}) \leq c.
\]
Just as above follows by \index{Jensen's inequality}Jensen's
 inequality (Lemma~\ref{lem:jensen}) that
\begin{align*}
  \sum_{i=0}^M \H_\alpha(Q_\rho^i) &= \sum_{i=0}^M
 -\log(\pi_\alpha(Q^i_\rho))\\
&\geq -(M+1)  \log\left(\frac{\sum_{i=0}^M \pi_\alpha(Q^i_\rho)}{M+1}\right) \geq (M+1)
  \log \left( \frac{M+1}{c} \right).
\end{align*}
Jensen's inequality is tight if the values $\pi_\alpha(Q^i_\rho)$ are
all equal.
\end{proof}

For incomplete sets of bases $\Ba{0},\ldots,\Ba{M}$ with $1 \leq M
\leq N$, the current state-of-the-art bound was independently obtained
by Damg{\aa}rd, Salvail and Pedersen~\cite{DPS04} and
Azarchs~\cite{Azarchs04} by subtracting the minimal amount of
collision probability ($1/N$) in the bases not included in the sum:
\begin{equation}
\sum_{i=0}^{M} \pi_2(Q^i_{\ket{\varphi}}) \leq 2 - \frac{(N + 1 - (M+1))}{N} 
  = \frac{N + M}{N}. \label{eq:sumcollision}
\end{equation}
By Lemma~\ref{lem:sumbound}, this yields
\begin{equation}
\sum_{i=0}^{M} \H_2(Q^i_\rho) \geq (M+1) \log \left( \frac{N (M+1)}{N+M} \right).    
\end{equation}

As mentioned above, all lower bounds on the collision entropy from
this section imply bounds on the Shannon \index{entropy!min-}
entropy because $H(Q) \geq H_2(Q)$, but do not tell us anything about
the min-entropy $H_\infty(Q)$.  In the rest of this chapter, we derive
entropic uncertainty relations involving \emph{min-entropy}.

Uncertainty relations in terms of R\'enyi entropy have also been studied in a
different context by Bialynicki-Birula \cite{Bialynicki06}.

\section{Two Mutually Unbiased Bases}\label{sec:twounbiasedbases}
In this section, we consider the situation where a $n$-qubit state is
measured in one out of two \index{mutually unbiased bases}mutually unbiased bases of
$\cH_{2^n}$. Without loss of generality, we assume these two bases to
be the $n$-fold tensor product of the computational basis $+^{\otimes
  n}$ and of the diagonal basis $\times^{\otimes n}$, in this section
simply called $+$- and $\times$-basis. 

We show that two distributions obtained by measuring in two mutually
unbiased bases cannot \emph{both} be ``very far from uniform''. One
way to characterize non-uniformity of a distribution is to identify a
subset of outcomes that has much higher probability than for a uniform
choice. Intuitively, the theorem below says that such sets cannot be
found simultaneously for \emph{both} measurements. 

\index{uncertainty relation!two mutually unbiased bases}
\begin{theorem} \label{thm:hadamard}
  Let $\rho$ be an arbitrary state of $n$ qubits, and let $\Qp(\cdot)$ and
  $\Qt(\cdot)$ be the respective distributions of the outcome when
  $\rho$ is measured in the $+$-basis respectively the $\times$-basis.
  Then, for any two sets $L^+ \subset \set{0,1}^n$ and $L^{\times}
  \subset \set{0,1}^n$ it holds that
\[ \Qp(L^+)+\Qt(L^{\times}) \leq 1 + 2^{-n/2} \sqrt{|L^+|
   |L^{\times}|}. \]
\end{theorem}
\begin{proof}
We define the two orthogonal projectors
\[
A \assign \sum_{x \in L^+} \proj{x} \quad \mbox{and} \quad B \assign \sum_{y \in
  L^{\times}} H^{\otimes n}\proj{y}H^{\otimes n}.
\]
Using the spectral decomposition of $\rho = \sum_w \lambda_w
\proj{\varphi_w}$, we have
\begin{align*}
\Qp(L^+)+\Qt(L^{\times}) &= \trace{A\rho}  + \trace{B\rho}\\
 &= \sum_w
\lambda_w \left( \trace{A \proj{\varphi_w}} + \trace{B
    \proj{\varphi_w}} \right)\\
&= \sum_w \lambda_w \left( \bra{\varphi_w}A\ket{\varphi_w} +
  \bra{\varphi_w}B\ket{\varphi_w} \right)\\
&= \sum_w \lambda_w \bra{\varphi_w}(A+B)\ket{\varphi_w}\\
&\leq \|A+B\| \leq 1 + \|AB\|,
\end{align*}
where the last line is Proposition~\ref{prop:norm2}.
To conclude, we show that $\|AB\| \leq 2^{-n/2} \sqrt{|L^+|
  |L^{\times}|}$. Note that an arbitrary state $\ket{\psi} = \sum_z
\lambda_z H^{\otimes n} \ket{z}$ can be expressed with coordinates
$\lambda_z$ in the diagonal basis. Then, with the sums over $x$ and $y$ understood as over $x \in L^+$ and $y \in L^{\times}$, respectively, 
\begin{align*}
\big\| AB &\ket{\psi} \big\| = \bigg\| \sum_{x,y} \proj{x} H^{\otimes
  n}\proj{y}H^{\otimes n}  
\ket{\psi}
\bigg\| 
= 2^{-n/2} \bigg\| \sum_{x,y} \ket{x}\bra{y}H^{\otimes n}  
\ket{\psi}
\bigg\|
  \\
&= 2^{-n/2} \bigg\| \sum_{x} \ket{x} \bigg\| \cdot \bigg|  \sum_{y} \lambda_y
  \bigg| 
\leq 2^{-n/2} \sqrt{|L^+|} \sum_{y} |\lambda_y| 
\leq 2^{-n/2} \sqrt{|L^+| |L^{\times}|},
\end{align*}
The second equality holds since $\bra{x}H^{\otimes n}\ket{y}=2^{-n/2}$
are mutually unbiased, the first inequality follows from Pythagoras
and the triangle inequality, and the last inequality follows from
\index{Cauchy-Schwarz}Cauchy-Schwarz (Lemma~\ref{lem:cauchyschwarz}).  
This implies $\|AB\| \leq 2^{-n/2} \sqrt{|L^+|
  |L^{\times}|}$ and finishes the proof.
\end{proof}

This theorem yields a meaningful bound as long as $|L^+| \cdot
|L^{\times}| < 2^n$, for instance if $L^+$ and
$L^{\times}$ both contain less than $2^{n/2}$ elements. The relation
is tight in the sense that for the Hadamard-invariant state
$$\ket{\varphi} = \left(\ket{0}^{\otimes n} + (H\ket{0})^{\otimes n}
\right)/\sqrt{2(1+2^{-n/2})}$$
and $L^+ = L^{\times} = \set{0^n}$, it
is straightforward to verify that $\Qp(L^+) = \Qt(L^{\times}) = (1 +
2^{-n/2})/2$ and therefore $\Qp(L^+) + \Qt(L^{\times}) = 1 +
2^{-n/2}$. Another state that achieves equality (for $n$ even) is $\ket{\varphi} =
\ket{0}^{\otimes n/2} \otimes ( H \ket{0} )^{\otimes n/2}$ with $L^+ =
\set{0^{n/2}x | x \in \set{0,1}^{n/2}}$ and $L^{\times} = \set{x
  0^{n/2} | x \in \set{0,1}^{n/2}}$. We get that $\Qp(L^+) =
\Qt(L^{\times}) = 1$ and thus $\Qp(L^+) + \Qt(L^{\times}) = 2 = 1 + 2^{-n/2}
\sqrt{2^n}$.

If for $r \in
\{+,\times \}$, $L^{r}$ contains only the $n$-bit string with the
maximal probability of $Q^{r}$, we obtain a known tight relation (see (9) in \cite{MU88}).
\begin{corollary}\label{cor:pmax}
Let $q_{\infty}^+$ and $q_{\infty}^{\times}$ be the maximal
probabilities of the distributions $Q^+$ and $Q^{\times}$ from
above. It then holds that $q_{\infty}^+ + q_{\infty}^{\times} \leq
1+c$ and therefore also $q_{\infty}^+ \cdot q_{\infty}^{\times} \leq
\frac{1}{4} (1+c)^2$ where $c=2^{-n/2}$.
\end{corollary}
Equality is achieved for the same state $\ket{\varphi} = \left(\ket{0}^{\otimes n} + (H\ket{0})^{\otimes n}
\right)/\sqrt{2(1+2^{-n/2})}$ as above. 

Using Lemma~\ref{lem:sumbound}, the following corollary is obtained.
\begin{corollary}\label{cor:pmax2} For all quantum states $\rho$ of
  $n$ qubits, it holds that
\[\H_\infty(Q_\rho^+) + \H_\infty(Q_\rho^\times) \geq 2 (1-\log(1+2^{-n/2})).
\]
There exists a quantum state achieving equality.
\end{corollary}

The following corollary plays the crucial role in the security proofs of
protocols in the \index{bounded-quantum-storage model}
bounded-quantum-storage model presented in the
following chapters of this thesis.
\begin{corollary} \label{cor:hadamard}
  Let $R$ be a random variable over $\set{+,\times}$, and let $X$ be
  the outcome when $\rho$ is measured in basis $R$, such that
  $P_{X| R}(x|r) = Q^r(x)$. Then, for any $\lambda < \frac12$ there
  exists $\sp>0$ and an event $\ev$ such that
$$
\P[\ev | R\!=\!+] + \P[\ev | R\!=\!\times] \geq 1 - 2^{-\sp n}
$$
and thus $\P[\ev] \geq \frac12 - 2^{-\sp n}$ in case $R$ is uniform, and such that 
$$
\H_{\infty}(X | R\!=\!r,\ev) \geq \lambda n 
$$
for $r \in \set{+,\times}$ with $P_{R| \ev}(r) > 0$. 
\end{corollary}

\begin{proof}
Choose $\sp > 0$ such that $\lambda + 2 \sp < \frac12$, and 
define
\begin{align*}
S^+ \assign \big\{ x \in \nbit &: \Qp(x) \leq  2^{-(\lambda + \sp )n}
\big\}\; \mbox{ and}\\
S^{\times} \assign \big\{ z \in \nbit &: \Qt(z) \leq  2^{-(\lambda
  +\sp )n} \big\}
\end{align*} 
to be the sets of strings with small probabilities and denote by $L^+
\assign \ol{S}^+$ and $L^{\times} \assign \ol{S}^{\times}$ their
complements\footnote{Here's the mnemonic: $S$ for the strings with
  \emph{S}mall probabilities, $L$ for \emph{L}arge.}. Note that for
all $x \in L^+$, we have that $\Qp(x) > 2^{-(\lambda + \sp )n}$ and
therefore $|L^+| < 2^{(\lambda + \sp)n}$. Analogously, we have
$|L^{\times}| < 2^{(\lambda + \sp)n}$. For ease of notation, we
abbreviate the probabilities that strings with small probabilities
occur with $\qp \assign \Qp(S^+)$ and $\qt \assign \Qt(S^{\times})$.
It follows immediately from the choice of $\sp$ and
Theorem~\ref{thm:hadamard} that $$\qp+\qt \geq 1 - 2^{-n/2} \cdot
2^{(\lambda + \sp)n} \geq 1- 2^{-\sp n} \, .$$

We define $\ev$ to be the event $X \in S^R$.  Then $\P[\ev|R\!=\!+]
= \P[X\in S^+|R\!=\!+] = \qp$ and similarly $\P[\ev|R\!=\!\times] =
\qt$, and thus the first claim follows immediately. Furthermore, if
$R$ is uniformly distributed, then
\begin{align*}
\P[\ev] &= \P[\ev|R\!=\!+] P_R(+) + \P[\ev|R\!=\!\times] P_R(\times)\\
 &= \frac{1}{2} (\qp + \qt) \geq \frac{1}{2} - 2^{-\sp n}/2 \geq
 \frac{1}{2} - 2^{-\sp n}.
\end{align*}
Regarding the second claim, in case $R=+$, we have
\begin{align*}
  \H_{\infty}(X|R\!=\!+, \ev) 
&= -\log\left(\max_{x \in S^+} \frac{\Qp(x)}{\qp}\right) \nonumber\\
  &\geq -\log\left(\frac{2^{-(\lambda + \sp)n}}{\qp}\right) = \lambda
  n + \sp n + \log(\qp). 
\end{align*}
Thus, if $\qp \geq 2^{-\sp n}$ then indeed $\H_{\infty}(X|R\!=\!+, X
\in S^+) \geq \lambda n$. The corresponding holds for the case $R =
\times$.

Finally, if $\qp < 2^{-\sp n}$ (or similarly \mbox{$\qt < 2^{-\sp
    n}$}) then instead of the above, we define $\ev$ as the {\em
  empty event} if $R = +$ and as the event $X \in S^{\times}$ if $R =
\times$. It follows that $\P[\ev|R\!=\!+] = 0$ and
$\P[\ev|R\!=\!\times] = \qt \geq 1 - 2^{-\sp n}$, as well as
$\H_{\infty}(X|R\!=\!\times, \ev) = \H_{\infty}(X|R\!=\!\times, X \in
S^{\times}) \geq \lambda n + \sp n + \log(\qt) \geq \lambda n$ (for
$n$ large enough), both by the bound on $\qp+\qt$ and on~$\qp$,
whereas $P_{R | \ev}(+) = 0$.
\end{proof}

\section{More Mutually Unbiased Bases}
\label{sec:moreunbiasedbases}
In this section, we generalize the uncertainty relation derived in
Section~\ref{sec:twounbiasedbases} to more than two mutually unbiased
bases. Such uncertainty relations over more than two, but not all mutually
unbiased bases in terms of  min-entropy may be of independent
interest, see the discussion at the end of Section~\ref{sec:uncerthistory}.

\index{uncertainty relation!more mutually unbiased bases}
\begin{theorem} \label{thm:mub}
  Let the density matrix $\rho$ describe the state of $n$ qubits
  and let $\Ba{0}, \Ba{1}, \ldots, \Ba{M}$ be mutually
  unbiased bases of $\cH_{2^n}$. Let $Q^0(\cdot), Q^1(\cdot), \ldots,
  Q^M(\cdot)$ be the distributions of the outcome when $\rho$ is
  measured in bases $\Ba{0}, \Ba{1}, \ldots, \Ba{M}$, respectively.
  Then, for any sets $L^0, L^1, \ldots, L^M \subset \nbit$, it holds that
\begin{align*}
\sum_{i=0}^M& Q^i(L^i)
\leq \: 1 + M \cdot 2^{-n/2} \max_{0 \leq i<j \leq M} \sqrt{|L^i| |L^j|}.
\end{align*}
\end{theorem}
\begin{proof}
  Except of using Proposition~\ref{prop:morebases} instead of
  Proposition~\ref{prop:norm2}, the proof is analogous to the one of
  Theorem~\ref{thm:hadamard}.
\end{proof}

As in Corollary~\ref{cor:pmax2}, we derive an uncertainty relation about
the sum of the min-entropies of up to $2^{n/2}$ distributions.
\begin{corollary} \label{cor:generalhmax}
  For an $\varepsilon>0$, let $0< M < 2^{\frac{n}{2}-\varepsilon n}$.
  For $i=0,\ldots,M$, let $\H_{\infty}(Q^i)$ be the min-entropies of the
  distributions $Q^i$ from the theorem above. Then,
\[ \sum_{i=0}^M \H_{\infty}(Q^i) \geq (M+1) \big(\log (M+1) - \negl{n}\big). \]
\end{corollary}
\begin{proof}
  For $i=0,\ldots,M$, we denote by $q_{\infty}^i$ the maximal
  probability of $Q^i$ and let $L^i$ be the set containing only the
  $n$-bit string $x$ with this maximal probability $q_{\infty}^i$.
  Theorem~\ref{thm:mub} together with the assumption about $M$ assures
  $\sum_{i=0}^M q_{\infty}^i \leq 1 + \negl{n}$. By
  Lemma~\ref{lem:sumbound} follows
\begin{align*}
\sum_{i=0}^M \H_{\infty}(Q^i) &\geq (M+1) \big(\log(M+1) - \negl{n}\big).
\end{align*}
\end{proof}

\section{Independent Bases for Each Subsystem} \label{sec:morerelation}
So far, we have focused on the case of an $n$-qubit state $\rho \in
\dens{\cH_{2^n}}$ measured in two or more mutually unbiased bases of
$\cH_{2^n}$. In this section, we investigate the case when each of
the $n$ qubits is measured in an individual basis, picked
independently and uniformly from $\set{+,\times}$, i.e. $\rho$ is
measured in basis \mbox{$\Theta \in_R \set{+,\times}^n$}.

More generally, our result holds for a state $\rho \in \cH_d^{\otimes
  n}$ of $n$ quantum systems---each $d$-dimensional---which are
measured in an individual basis, picked independently and uniformly
from a set $\cB$ of basis of $\cH_d$, see Theorem~\ref{thm:genrel}.

\subsection{A Classical Tool}
We start our derivation with a classical information-theoretic tool
which itself might be of independent interest.
\begin{theorem}\label{thm:hmin}
Let $Z_1,\ldots, Z_n$ be $n$ random variables (not necessarily
independent) over alphabet $\cZ$. If there exists a real number
$h>0$ such that for all $1\leq i\leq n$ and $z_1,\ldots, z_{i-1} \in \cZ$:
\begin{equation*} 
\H(Z_i | Z_1=z_1,\ldots,Z_{i-1}=z_{i-1})\geq h,
\end{equation*}
then for any $0<\lambda<\frac12$ 
\[
\hiee{Z_1,\ldots,Z_n} \geq (h-2\lambda) n,
\]
where $\varepsilon =
\exp{\bigl(-\frac{\lambda^2 n}{32\log(|\mathcal{Z}|/\lambda)^2}\bigr)}$.
\end{theorem}
If the $Z_i$'s are independent and have Shannon
\index{entropy!Shannon}entropy at least $h$,
it is known (see Lemma~\ref{lem:asymptshannon}) that the smooth min-entropy of
$Z_1,\ldots,Z_n$ is at least $n h$ for large enough~$n$. Informally,
Theorem~\ref{thm:hmin} guarantees that when the independence-condition
is relaxed to a lower bound on the Shannon entropy of $Z_i$ \emph{given
  any previous history}, then we still have (almost) $n h$ bits of
min-entropy except with negligible probability~$\varepsilon$. \index{entropy!min-}

The proof idea is to use \index{Azuma's inequality}Azuma's inequality in the form of
Corollary~\ref{cor:azuma} for cleverly chosen $R_i$'s. The main trick
is that for a random variable $Z$ over $\cZ$, we can define another
random variable $S \assign \log P_Z(Z)$ over $\reals$ with expected
value $\E[S] = \sum_{z \in \cZ} P_Z(z) \cdot \log P_Z(z) = \H(Z)$
equal to the Shannon entropy of $Z$, which allows us to make the
connection with the assumption about the Shannon entropy.

\begin{proof}
Recall that the superscript means $\prei[i]{Z} \assign (Z_1,\ldots,Z_i)$ for any $i \in
\set{1,\ldots,n}$, and similarly for other sequences.  We want to show
that $$\Pr\big[P_{\prei[n]{Z}}(\prei[n]{Z}) \geq
  2^{-(h-2\lambda)n}\big] \leq \varepsilon$$ for $\varepsilon$ as
claimed in Theorem~\ref{thm:hmin}.  This means that
$P_{\prei[n]{Z}}(\prei[n]{z})$ is smaller than $2^{-(h-2\lambda)n}$
except with probability at most $\varepsilon$ (over the choice of
$\prei[n]{z}$), and therefore implies the claim
\mbox{$H_{\infty}^{\varepsilon}(Z^n) \geq (h-2\lambda)n$}
by the definition of smooth min-entropy from Section~\ref{sec:defsmoothrenyientropy}.
Note that
$P_{\prei[n]{Z}}(\prei[n]{Z}) \geq 2^{-(h-2\lambda)n}$ is equivalent to
\begin{equation}\label{eq:bound1}
  \sum_{i=1}^n \Big( \log\big( P_{Z_i | \prei{Z}}(Z_i | \prei{Z}) \big) + h \Big) \geq 2 \lambda n
\end{equation}
which is of suitable form to apply \index{Azuma's inequality}Azuma's
 inequality (Corollary~\ref{cor:azuma}). 

Consider first an arbitrary sequence $S_1,\ldots,S_n$ of real-valued
random variables. We assume the $S_i$'s to be either all positive or
all negative. Define a new sequence $R_1,\ldots,R_n$ of random
variables by putting $R_i := S_i - \E[S_i | \prei{S}]$. It is
straightforward to verify that $\E[R_i | \prei{R}] = 0$, i.e.,
$R_1,\ldots,R_n$ forms a martingale difference sequence. Thus if for any $i$,
$|S_i| \leq c$ for some $c$, and thus $|R_i| \leq c$,
Azuma's inequality guarantees that
\begin{equation}\label{eq:bound2}
\Pr\left[\sum_{i=1}^n \Big(S_i - \E\big[S_i | \prei{S}\big]\Big) \geq \lambda n\right] \leq \exp\left(-\frac{\lambda^2 n}{2 c^2}\right) \, .
\end{equation}
We now put $S_i := \log P_{Z_i | \prei{Z}}(Z_i | \prei{Z})$ for
$i=1,\ldots,n$. Note that $S_1,\ldots,S_n \leq 0$. It is easy to see that
the bound on the conditional entropy of $Z_i$ from
Theorem~\ref{thm:hmin} implies that $\E[S_i | \prei{S}] \leq -h$.  Indeed,
for any $\prei{z} \in \cZ^{i-1}$, we have $\E\big[\log
P_{Z_i | \prei{Z}}(Z_i | \prei{Z}) | \prei{Z}\!=\!\prei{z}\big] = -
\H(Z_i | \prei{Z}\!=\!\prei{z}) \leq - h$, and thus for any subset $\ev$
of $\cZ^{i-1}$, and in particular for the set of $\prei{z}$'s
which map to a given $\prei{s}$, it holds that
\begin{align} \nonumber
\E\big[S_i | \prei{Z}\!\in\!\ev\big] &=
\sum_{\prei{z} \in \ev}\!\! P_{\prei{Z} | \prei{Z}\in\ev}(\prei{z}) \cdot
  \E\big[\log P_{Z_i | \prei{Z}}(Z_i | \prei{Z}) |
  \prei{Z}\!=\!\prei{z}\big]\\
 &\leq - h \, . \label{eq:bound3}
\end{align}
As a consequence, the bound on the probability of (\ref{eq:bound2}) in
particular bounds the probability of the event (\ref{eq:bound1}), even
with $\lambda n$ instead of $2 \lambda n$. A problem though is that we
have no upper bound $c$ on the $|S_i|$'s.  Because of that, we now
consider a modified sequence $\tilde{S}_1,\ldots,\tilde{S}_n$ defined
by $\tilde{S}_i := \log P_{Z_i | \prei{Z}}(Z_i | \prei{Z})$ if
$P_{Z_i | \prei{Z}}(Z_i | \prei{Z}) \geq \delta$ and
$\tilde{S}_i := 0$ otherwise, where $\delta > 0$ will be determined
later.  This gives us a bound like (\ref{eq:bound2}) but with an
explicit $c$, namely $c = \log(1/\delta)$. Below, we will argue that
$\E\big[\tilde{S}_i | \prei{\tilde{S}}\big]-\E\big[S_i |
\prei{\tilde{S}}\big] \leq \lambda$ by the right choice of $\delta$;
the claim then follows from observing that
\begin{align*}
\tilde{S}_i - \E\big[\tilde{S}_i | \prei{\tilde{S}}\big] &\geq S_i -
\E\big[\tilde{S}_i | \prei{\tilde{S}}\big]\\ 
&\geq S_i - \E\big[S_i | \prei{\tilde{S}}\big] - \lambda\\
&\geq S_i + h - \lambda,
\end{align*}
where the last inequality follows from (\ref{eq:bound3}).  Regarding
the claim $\E\big[\tilde{S}_i | \prei{\tilde{S}}\big]-\E\big[S_i
| \prei{\tilde{S}}\big] \leq \lambda$, using a similar argument as
for (\ref{eq:bound3}), it suffices to show that $\E\big[\tilde{S}_i
| \prei{\tilde{Z}}\!=\!\prei{z}\big]-\E\big[S_i |
\prei{\tilde{Z}}\!=\!\prei{z}\big] \leq \lambda$ for any $\prei{z}$:
\begin{align*}
\E&\big[\tilde{S}_i | \prei{\tilde{Z}}\!=\!\prei{z}\big]-
\E\big[S_i | \prei{\tilde{Z}}\!=\!\prei{z}\big]\\
  &= - \sum_{z_i} P_{Z_i |\prei{Z}}(z_i | \prei{z})
  \log(P_{Z_i | \prei{Z}}(z_i | \prei{z}))\\
  &\leq |\mathcal{Z}| \delta \log ( 1/ \delta)
\end{align*}
where the summation is over all $z_i \in \cZ$ with $P_{Z_i | \prei{Z}}(z_i | \prei{z}) <
    \delta$, and 
where the inequality holds as long as $\delta \leq 1/e$, as can
easily be verified. Thus, we let $0<\delta<1/e$ be such that
$|\mathcal{Z}| \delta \log(1/\delta) = \lambda$.  Using the mathematical
Lemma~\ref{lem:delta}, we have that $\delta > \frac{\lambda /
  |\mathcal{Z}|}{4 \log{(|\mathcal{Z}| / \lambda})}$ and derive
that $c^2 = \log(1/\delta)^2 = \lambda^2/(\delta|\cZ|)^2 < 16
\log(|\mathcal{Z}|/\lambda)^2$, which gives us the claimed bound
$\varepsilon$ on the probability.
\end{proof}

\subsection{Quantum Uncertainty Relations}
We now state and prove the new entropic uncertainty relation
in its most general form. A special case will then be introduced
(Corollary~\ref{cor:uncertainty})
and used in the security analysis of the \OT-protocols we consider
in Chapter~\ref{chap:12OT}. 
\begin{definition}\label{def:aeub}
  Let $\bset$ be a finite set of orthonormal bases in the $d$-dimensional
  Hilbert space~$\cH_d$.  We call $h \geq 0$ an {\em average
    entropic uncertainty bound} for $\bset$ if every state in $\cH_d$ satisfies $\frac{1}{|\bset|} \sum_{\vartheta \in \bset}\H(P_{\vartheta}) \geq h$, where $P_{\vartheta}$ is the
  distribution obtained by measuring the state in basis $\vartheta$.
\end{definition}
\index{average entropic uncertainty bound}
Note that by the convexity of the Shannon \index{entropy!Shannon}entropy $\H$, a lower bound
for all \emph{pure} states in $\cH_d$ suffices to imply the bound
for all (possibly mixed) states.

\index{uncertainty relation!individual bases}
\begin{theorem}\label{thm:genrel}
  Let $\bset$ be a set of orthonormal bases in $\cH_d$ with an
  average entropic uncertainty bound $h$, and let 
  $\rho \in \dens{\cH^{\otimes n}_d}$ be an arbitrary quantum state.
Let $\Theta = (\Theta_1,\ldots,\Theta_n)$ be uniformly distributed
over $\bset^n$ and let $X = (X_1,\ldots,X_n)$ be the outcome when
measuring $\rho$ in basis $\Theta$, distributed over
\mbox{$\set{0,\ldots, d-1}^n$}.  Then for any $0 < \lambda < \frac12$
$$
\hie{\varepsilon}{X | \Theta} \geq \left(h-
  2\lambda \right)n 
$$
with 
$\varepsilon = \exp \!\left( - \frac{\lambda^2 n}{32
    \left(\log(|\bset|\cdot d / \lambda) \right)^2} \right)$.
\end{theorem}

\begin{proof}
Define $Z_i \assign (X_i,\Theta_i)$ and $\prei[i]{Z} \assign (Z_1,\ldots,Z_i)$. Let 
$\prei[i-1]{z}\in \bset^{i-1}$ be arbitrary. Then
\begin{align*} 
\H(Z_i | \prei[i-1]{Z}\!=\!\prei[i-1]{z}) &= \H(X_i | \Theta_i, \prei[i-1]{Z}\!=\!\prei[i-1]{z})
+ \H(\Theta_i | \prei[i-1]{Z}\!=\!\prei[i-1]{z}) \geq h + \log{|\bset|},
\end{align*}
where the inequality follows from the fact that $\Theta_i$ is chosen
uniformly at random and from the definition of $h$. Note that $h$
lower bounds the average entropy for any system in $\cH_d$, 
and thus in particular for the $i$th subsystem 
of $\rho$, with all
previous $d$-dimensional subsystems measured.  
Theorem~\ref{thm:hmin} thus implies that $\hie{\varepsilon}{X \Theta} \geq (h+\log|{\cal B}| - 2\lambda) n$ for any $0 < \lambda < \frac12$ and for $\eps$ as claimed. We conclude that 
\begin{align*}
\hie{\varepsilon}{X\mid \Theta}
&\geq \hie{\varepsilon}{X \Theta} - n\log|{\cal B}| \geq (h - 2 \lambda)n \enspace ,
\end{align*}
where the first inequality follows from the equality 
$$
P_{X\ev|\Theta}(x|\theta) = P_{X \Theta \ev}(x, \theta)/P_{\Theta}(\theta) = |{\cal B}|^n \cdot P_{X \Theta \ev}(x, \theta)
$$
for all $x$ and $\theta$ and any event $\ev$, and from the definition of (conditional) smooth entropy.
\end{proof}


For the special case where $\bset=\{+,\times\}$ is the set of BB84
\index{BB84 coding scheme}bases, we can use the uncertainty relation of Maassen and
Uffink~\cite{MU88} (see Equation~\eqref{eq:maassenuffink}) which, using our
terminology, states that \index{Maassen and Uffink's relation}
$\bset$ has average entropic uncertainty bound $h =
\frac12$.  Theorem~\ref{thm:genrel} together with
Lemma~\ref{lem:epsilon} then immediately gives the
following corollary.

\index{uncertainty relation!individual bases}
\begin{corollary}\label{cor:uncertainty}
  Let $\rho \in \dens{\cH_2^{\otimes n}}$
  be an arbitrary $n$-qubit quantum state. Let $\Theta =
  (\Theta_1,\ldots,\Theta_n)$ be uniformly distributed over
  $\set{+,\times}^n$ and $X = (X_1,\ldots,X_n)$ be the outcome
  when measuring $\rho$ in basis $\Theta$. Then for any $0 <
  \lambda < \frac{1}{4}$
$$
\hie{\varepsilon}{X | \Theta} \geq \left(\textstyle\frac{1}{2}
  - 2\lambda \right)n 
$$
where $\varepsilon = 2^{-\frac{\lambda^4}{32}n }$.
\end{corollary}

Maassen and Uffink's relation being optimal means there exists a
quantum state $\rho$|namely the product state of eigenstates of the
subsystems, e.g.~$\rho=\proj{0}^{\otimes n}$|for which $\H(X |
\Theta) = \frac{n}{2}$. On the other hand, we have shown that
$(\frac12 - \lambda)n \leq \hiee{X | \Theta}$ for $\lambda>0$
arbitrarily close to $0$. For the product state $\rho$, the $X$ are
independent and we know from Lemma~\ref{lem:asymptshannon} that
$\hiee{X | \Theta}$ approaches $\H(X | \Theta) = \frac{n}{2}$.
It follows that the relation cannot be significantly improved even
when considering R\'enyi entropy of order $1 < \alpha < \infty$. 

Another tight corollary is obtained if we consider the set of
measurements $\bset=\{+,\times,\circlearrowleft\}$ (see Section~\ref{sec:qit} 
for the definition of the circular basis $\circlearrowleft$). In \cite{Ruiz93},
S\'anchez-Ruiz shows that for this $\bset$, the average entropic
uncertainty bound
\begin{equation} \label{eq:hthreebases}
h=\frac{2}{3}
\end{equation}
is optimal. It implies that
$\hie{\varepsilon}{X| \Theta} \gtrsim \H(X| \Theta) =
\frac{2n}{3}$ for negligible $\varepsilon$. 

\subsection{The Overall Average Entropic Uncertainty Bound}\label{sec:uncertbound}
\index{average entropic uncertainty bound!overall}
In the this section, we compute the average uncertainty bound for the
set of \emph{all bases} of a $d$-dimensional Hilbert space. Let
$\cU(d)$ be the set of unitaries on $\cH_d$.  Moreover, let $d U$ be
the normalized \index{Haar measure}Haar measure on $\cU(d)$, i.e.,
\[
  \int_{\cU(d)} f(V U) d U 
= 
  \int_{\cU(d)} f(U V) d U 
= 
  \int_{\cU(d)} f(U) d U \ ,
\]
for any $V \in \cU(d)$ and any integrable function $f$, and $\int_{\cU(d)} dU = 1$. (Note that the
normalized Haar measure $d U$ exists and is unique.)

\def\all{\text{\rm all}}

Let $\{\omega_1, \ldots, \omega_d\}$ be a fixed orthonormal basis of
$\cH_d$, and let $\bset_{\all} = \{\vartheta_U\}_{U \in \cU(d)}$ be the
family of bases $\vartheta_U = \{U \omega_1, \ldots, U \omega_d\}$
with $U \in \cU(d)$. The set $\bset_{\all}$ consist of {\em all}
orthonormal basis of $\cH_d$. We generalize Definition~\ref{def:aeub},
the average entropic uncertainty bound for a finite set of bases, to
the {\em infinite} set $\bset_{\all}$.
\begin{definition}
We call $h_d$ an {\em overall average entropic uncertainty bound} in $\cH_d$ if every state in $\cH_d$ satisfies
  \[
     \int_{\cU(d)} \H(P_{\vartheta_U}) d U \geq h_d \ ,
  \]
  where $P_{\vartheta_U}$ is the distribution obtained by measuring the state in basis $\vartheta_U \in \bset_{\all} $. 
  \end{definition}
  \index{average entropic uncertainty bound!overall}

\begin{proposition} \label{prop:comph} 
For any positive integer $d$, 
  \[
  h_d = \left( \sum_{i=2}^d \frac{1}{i} \right) / \ln(2)
  \]
  is the overall average entropic uncertainty bound in $\cH_d$. It is
  attained for any pure state in $\cH_d$.
\end{proposition}
The proposition follows immediately from Formula~(14) in~\cite{JRW94}
for a pure state, i.e. $(\lambda_1,\ldots,\lambda_n)=(1,0,\ldots,0)$.
The result was originally shown by S\'ykora~\cite{Sykora74} and by
Jones~\cite{Jones91}, another proof can be found in the appendix of
an article by Jozsa, Robb, and Wootters~\cite{JRW94}. An elementary
proof suggested by Harremo{\"es} based on recent results by
Harremo{\"e}s and Vignat~\cite{HV06} is given below.
\begin{proof}
  Let $\ket{\varphi}$ be a pure state in $\cH_d$. For the probability
  distribution $P_{\vartheta_U}=(p_1,\ldots,p_d)$ holds $p_i =
  |\bra{\varphi}U\ket{\omega_i}|^2$. We want to compute the integral
\[ \int_{\cU(d)} -\sum_{i=1}^d p_i \log(p_i) \, d U =
-\sum_{i=1}^d \int_{\cU(d)} |\bra{\varphi}U\ket{\omega_i}|^2
\log(|\bra{\varphi}U\ket{\omega_i}|^2) \, d U.
\]
 Note that by the invariance of the \index{Haar measure}Haar measure, all summands on the
 right-hand side are equal and it suffices to compute
\begin{equation} \label{eq:summand}
-d \int_{\cU(d)} |\bra{\varphi}U\ket{e_1}|^2
 \log(|\bra{\varphi}U\ket{e_1}|^2) \, d U,
\end{equation}
where $\ket{e_1}$ is the first vector in the computational basis,
i.e.~$|\bra{\varphi}U\ket{e_1}|^2$ is the length of the projection onto the first
coordinate of $U^*\ket{\varphi}$.

The Haar measure over $\cU(d)$ is the uniform distribution over the
$d$-dimensional complex sphere which can be seen as the uniform
distribution over the $2d$-dimensional real sphere $S_{2d} =
\set{(X,Y) \in \mathbb{R}^{2d} | \sum_{i=1}^{2d} X_i^2 +Y_i^2= 1}$
where the complex coordinates are given by
$(X_1+iY_1,\ldots,X_d+iY_d)$. Setting $Z_i=X_i^2+Y_i^2$ and
$Z=(Z_1,\ldots,Z_d)$ and using a result from \cite{HV06} about the
projection of the uniform distribution over $S_{2d}$ to the first
coordinate, we obtain that the density of $Z_1$ is
$f(z)=(d-1)(1-z)^{d-2} dz$ for $z \in [0,1]$. Therefore,
\eqref{eq:summand} equals
\[ -d \int_0^1 z \log(z) \cdot (d-1)(1-z)^{d-2} dz = \left(
\sum_{i=2}^d \frac{1}{i} \right) / \ln(2),
\]
where the evaluation of this integral follows from standard
calculus. By convexity of the Shannon entropy, the bound also holds
for mixed states and the claim follows.
\end{proof}



The following table gives some numerical values of $h_d$ for small
values of $d$.
\begin{center}
\begin{tabular}{c|cccc}
  $d$                    & $2$    & $4$    & $8$    & $16$ \\
  \hline
  $h_d$                  & $0.72$ & $1.56$ & $2.48$ & $3.43$ \\
  $\frac{h_d}{\log(d)}$ & $0.72$ & $0.78$ & $0.83$ & $0.86$ 
\end{tabular}
\end{center}

It is well-known that the harmonic series in
Proposition~\ref{prop:comph} diverges in the same way as $\log(d)$
and therefore, $\frac{h_d}{\log(d)}$ goes to 1 for large dimensions
$d$.

\index{uncertainty relation|)}

%% file: rabinot.tex
\chapter[\RabinOT in the Bounded-Quantum-Storage Model]{\RabinOT in the Bounded-Quantum-Storage
  Model} \label{chap:RabinOT}
\index{oblivious transfer!\RabinOT|(}
In this chapter, we present an efficient protocol for Rabin Oblivious Transfer
which is secure in the bounded-quantum-storage model. It first
appeared in \cite{DFSS05}, a journal version of this paper is in preparation
\cite{DFSS08journal}.

\section{The Definition} \label{sec:def-rabin-obliv-transf}
A protocol for Rabin Oblivious Transfer (\RabinOT) between sender
Alice and receiver Bob allows for Alice to send a bit $b$ through an
erasure \index{erasure channel} channel to Bob. Each transmission
delivers $b$ or an erasure with probability $\frac12$.  Intuitively, a
protocol for \RabinOT is secure if
\begin{itemize}
\item the sender Alice gets no information on whether $b$ was
                  received or not, no matter what she does, and
\item the receiver Bob gets no information about $b$ with
  probability at least~$\frac{1}{2}$, 
   no matter what he does. 
\end{itemize}
In this chapter, we are considering quantum protocols for \RabinOT. This
means that while the inputs and outputs of the honest senders are classical,
described by random variables, the protocol may contain quantum
computation and quantum communication, and the view of a dishonest
player is quantum, and is thus described by a 
quantum state. \index{quantum protocol}

Any such (two-party) protocol is specified by a family
$\{(\A_n,\B_n)\}_{n>0}$ of pairs of interactive quantum circuits (i.e.
interacting through a quantum channel). Each pair is indexed by a
security parameter $n>0$, where $\A_n$ and $\B_n$ denote the circuits
for sender Alice and receiver Bob, respectively.  In order to simplify
the notation, we often omit the index $n$, leaving the dependency on it
implicit.


For the formal definition of the security requirements of a \RabinOT
protocol, let us fix the following notation. Let $B$ denote the binary
random variable describing \A's input bit $b$, and let $A$ and $Y$
denote the binary random variables describing \B's two output bits,
where the meaning is that $A$ indicates whether the bit was received
or not. 
Furthermore, for a dishonest sender~\dA, the final state of a fixed
candidate protocol for \RandOT can be described by the ccq-state
\smash{$\rho_{AY\dA}$} where (by slight abuse of notation) we also denote by
\smash{$\dA$} the quantum register that the sender outputs. Its state may
depend on $A$ and $Y$. Similarly, for a dishonest receiver \dB, we
have the cq-state $\rho_{B \dB}$.

\begin{definition}\label{def:ROT}
  A two-party (quantum) protocol $(\A,\B)$ is a \emph{$\eps$-secure \RabinOT} if the following holds:
\begin{description}
\item[\boldmath$\varepsilon$-Correctness:] For honest \A\ and \B, 
$$\P[B=Y|A=1] \geq 1 - \eps \, .$$
\item[\boldmath$\varepsilon$-Receiver-security:] For honest \B\ and any dishonest \dA\ there
  exists\footnote{Recall from Section~\ref{sec:qit}: Given a cq-state
    $\rho_{X \regE}$, by saying that there exists a random
    variable $Y$ such that $\rho_{XY\regE}$ satisfies some condition,
    we mean that $\rho_{X \regE}$ can be understood as $\rho_{X\regE}
    = \tr_Y(\rho_{XY\regE})$ for a ccq-state $\rho_{XY\regE}$ that
    satisfies the required condition.} a binary random variable $B'$
  such that
\[ \P[B'=Y|A=1] \geq 1 - \eps, \quad\mbox{ and }\quad \dist{\rho_{AB'\dA}, \I \otimes \rho_{B'\dA} } \leq \eps \, . \]
\item[\boldmath$\varepsilon$-Sender-security:] For any \dB\ there exists an event $\ev$ with
  $\P[\ev] \geq \frac{1}{2} - \eps$ such that
\[
\dist{ \rho_{B\dB|\ev}, \rho_B \otimes \rho_{\dB|\ev} } \leq \eps \, .
\]
\end{description}
\index{correctness!of quantum \RabinOT}
\index{receiver-security!of quantum \RabinOT}
\index{sender-security!of quantum \RabinOT}
If any of the above holds for $\eps=0$,  
then the corresponding property is said to hold {\bf perfectly}. 
If one of the properties only holds with respect to a restricted class
$\mathfrak{S}$ of \dA's respectively $\mathfrak{R}$ of \dB's, then this property
is said to hold (and the protocol is said to be secure) {\bf against}
$\mathfrak{S}$ respectively~$\mathfrak{R}$.
\end{definition}

Receiver-security requires that the joint quantum state is essentially
the same as when the dishonest sender chooses a bit $B'$ according to
some distribution and a (possibly dependent) quantum state, and gives
$B'$ to an ideal functionality which passes it on to the receiver with
probability $\frac12$. Sender-security requires that the joint quantum
state is essentially the same as when the dishonest receiver gets the
sender's bit $B$ with probability $\frac12$ and prepares some state
that may depend on $B$ in case he receives it, and prepares some state
that does not depend on $B$ otherwise. In other words, security
requires that the dishonest party cannot do more than when attacking
an ideal functionality.  From such a strong security guarantee we
expect nice composition behavior, for instance like
in~\cite{CSSW06}.

Note that the original definition given in~\cite{DFSS05} does not
guarantee that the distribution of the input bit is determined at the
end the execution of \RabinOT. This is a strictly weaker definition
and does not fully capture what is expected from a \RabinOT: it is
easy to see that if the dishonest sender can still influence his input
bit after the execution of the protocol, then known schemes based on
\RabinOT, like bit commitments, are not secure anymore. The security
definition given here is in the spirit of the security definition
from~\cite{DFRSS07} for 1-2~OT, described in the next
Chapter~\ref{chap:12OT}.

\section{The Protocol}\label{sec:otprot}
We present a quantum protocol for \RabinOT that will be shown
perfectly correct and perfectly receiver-secure (against any sender)
and statistically sender-secure against any quantum-memory-bounded
receiver. Our protocol exhibits some similarity with quantum conjugate
coding introduced by Wiesner~\cite{Wiesner83}. \index{conjugate coding}

\begin{myfigure}{h}
 \begin{myprotocol}{\qot$(b)$}
      \item $\A$ picks $x\in_R\{0,1\}^n$, and $r\in_R\{+,\times \}$
        and sends $\ket{\psi} \assign \ket{x}_r$ 
            to $\B$ (i.e. the string $x$ encoded in basis $r$).
      \item $\B$ picks $r'\in_R\{+,\times \}$ and measures all qubits
       of $\ket{\psi}$ in basis $r'$. Let $x'\in\{0,1\}^n$ be the
       result.
      \item $\A$ announces $r$, $\hf\in_R \chf{n}$, and $e \assign b\oplus \hf(x)$.\label{bound}
      \item $\B$ outputs $a \assign 1$ and $y \assign e\oplus \hf(x')$ if $r'=r$ and else $a \assign 0$ and $y \assign 0$.
 \end{myprotocol}
\caption{Quantum Protocol for \RabinOT}\label{fig:ot}
\end{myfigure}

The protocol given in~Figure~\ref{fig:ot} is very simple: $\A$ picks
$x\in_R\{0,1\}^n$ and sends to
$\B$ $n$ qubits in state either $\ket{x}_+$ 
or $\ket{x}_{\times}$ 
each chosen with probability~$\frac{1}{2}$.  $\B$ then measures all
received qubits either in the rectilinear or in the diagonal basis.
\index{basis!rectilinear} \index{basis!diagonal} With probability
\smash{$\frac{1}{2}$}, $\B$ picked the right basis and gets $x$, while
any \dB\ that is forced to measure part of the state (due to a memory
bound) can only have full information on $x$ in case the $+$-basis was
used {\em or} in case the $\times$-basis was used (but not in both
cases). \index{privacy amplification}Privacy amplification based on
any \univ\ class of hashing functions $\chf{n}$ is then used to
eliminate partial information (as explained in Section~\ref{sec:pa}).
For simplicity, we focus on the case where the output size of the
family $\chf{n}$ is just one bit, i.e. $\ell=1$, but all results of
this chapter can easily be extended to \RabinlStringOT of $\ell$-bit
strings, by using an output size $\ell > 1$ and adjusting the memory
bounds accordingly, see Section~\ref{sec:extensionrot}.

In order to avoid \index{abort}aborting, we specify that if a dishonest \dA\ 
refuses to participate, or sends data in incorrect format, then \B\ 
samples its output bits $a$ and $y$ both at random in $\set{0,1}$.

We first consider receiver-security. 
\index{receiver-security!of quantum \RabinOT}
\begin{proposition}\label{prop:sec:receiverprivate}
\qot\ is perfectly receiver-secure.
\end{proposition}
It is obvious that no information about whether $\B$ has received the
bit is leaked to any sender \smash{$\dA$}, since $\B$ does not send
anything. However, one needs to show the existence of a random
variable $B'$ as required by receiver-security.
\begin{proof}
  Recall, the quantum state \smash{$\rho_{AY\dA}$} is defined by the
  experiment where the dishonest sender $\dA$ interacts with the
  honest memory-bounded $\B$. Consider a modification of the
  experiment where we allow $\B$ to be {\em unbounded} in memory and
  where $\B$ waits to receive $r$ and then measures all qubits in
  basis~$r$. Let $X'$ be the resulting string. Nevertheless, $\B$
  picks $r' \in_R \set{+,\times}$ at random and outputs $(A,Y) =
  (0,0)$ if $r' \neq r$ and $(A,Y) = (1,e \oplus f(X'))$ if $r' = r$.
  Since the only difference between the two experiments is {\em when}
  $\B$ measures the qubits and {\em in what basis} $\B$ measures them
  when $r \neq r'$, in which case his final output is independent of
  the measurement outcome, the two experiments result in the same
  $\rho_{AY\dA}$. However, in the modified experiment we can choose
  $B'$ to be $e \oplus f(X')$, such that by construction $B' = Y$ if
  $A=1$ and $A$ is uniformly distributed, independent of anything, and
  thus $\rho_{AB'\dA} = \I \otimes \rho_{B'\dA}$.
\end{proof}

As we shall see in Section~\ref{sec:otsecurity}, the security of the
\qot\ protocol against receivers with bounded-size quantum memory
holds as long as the bound applies before Step~\ref{bound} is reached.
An equivalent protocol is obtained by purifying the sender's actions.
Although \qot\ is easy to implement, the purified or EPR-based version
depicted in Figure~\ref{fig:eprot} is easier to prove secure. This
technique was pioneered by Ekert~\cite{Ekert91} in the scenario of
quantum key distribution. A similar approach was taken in the
Shor-Preskill proof of security for the BB84 quantum-key-distribution
scheme~\cite{SP00}. \index{purification} \index{EPR-based version|see {purification}}

\begin{myfigure}{h}
\begin{myprotocol}{\eprqot$(b)$}
\item $\A$ prepares $n$ EPR pairs each in state
  $\ket{\Omega}=\frac{1}{\sqrt{2}}(\ket{00}+\ket{11})$ and sends one
  half of each pair to $\B$ and keeps the other halves.\label{rec}
\item $\B$ picks $r'\in_R\{+,\times \}$ and measures all received qubits
      in basis $r'$. Let $x'\in\{0,1\}^n$ be the result.
\item $\A$ picks $r\in_R\{+,\times \}$, and measures all kept
      qubits in basis $r$. Let $x\in\{0,1\}^n$ be the outcome.  $\A$
      announces $r$, $\hf\in_R \chf{n}$, and $e \assign b\oplus \hf(x)$.\label{it:measure}
\item $\B$ outputs $a \assign 1$ and $y \assign e \oplus \hf(x')$ if $r'=r$ and else $a \assign 0$ and $y \assign 0$.
\end{myprotocol}
\caption{Protocol for EPR-based \RabinOT}\label{fig:eprot}
\end{myfigure}

Notice that while \qot\ requires no quantum memory for honest players,
quantum memory for $\A$ seems to be required in \eprqot. The following
Lemma shows the strict security equivalence between \qot\ and \eprqot.

\begin{lemma}\label{lem:seqequiv}
  \qot\ is $\eps$-sender-secure if and only if \eprqot\ is.
\end{lemma}
\begin{proof}
The proof follows easily after observing that $\A$'s
choices of $r$ and $\hf$, together with 
the measurements all commute with $\dB$'s actions.
Therefore, they can be performed right after Step 1
with no change for $\dB$'s view. Modifying \eprqot\
that way results in \qot.
\end{proof}

Note that for a dishonest receiver it is not only irrelevant whether he tries to attack \qot\ or \eprqot, but in fact there is no difference in the two protocols from his point of view. 

\section{Modeling Dishonest Receivers} \label{sec:modeldishonestreceiversrabin}
We model dishonest receivers in \qot, respectively \eprqot, under the assumption that 
the maximum size of their quantum storage is bounded.
These adversaries are only required to have bounded quantum storage
when they reach Step~\ref{bound} in (\epr-)\qot. 
Before (and after) that, \index{bounded-quantum-storage model}
the adversary can store and carry out quantum computations involving any number
of qubits. Apart from the restriction on the size of the quantum
memory available to the adversary, no other assumption is made. In 
particular, the adversary is not assumed to be computationally
bounded and the size of its classical memory is not restricted. 
\begin{definition}\label{boundedstorage}
The set $\mathfrak{R}_{\gamma}$ denotes all
possible quantum dishonest receivers $\{\dB_n\}_{n>0}$ in \qot\ or \eprqot\ 
where for each $n>0$, $\dB_n$ has quantum memory of size 
at most $\gamma n$ when Step \ref{bound} is reached.  
\end{definition}
\index{dishonest receiver!of \RabinOT}
In general, the adversary $\dB$ is allowed to perform any quantum
computation compressing the $n$ qubits received from $\A$ into a
quantum register $M$ of size at most $\gamma n$ when Step \ref{bound}
is reached. More precisely, the \index{compression}compression function is implemented by
some unitary transform $T$ acting upon the quantum state received and
an \index{ancilla}ancilla register of arbitrary size (initially in the state
$\ket{0}$). The compression is performed by a measurement that we
assume in the computational basis without loss of generality.  Before
starting Step \ref{bound}, the adversary first applies a unitary
transform~$T$:
\begin{eqnarray*}
2^{-n/2} \!\! \sum_{x\in\{0,1\}^n}\ket{x}\otimes T\ket{x}\ket{0}
  \mapsto  2^{-n/2} \!\! \sum_{x\in\{0,1\}^n} \ket{x} \otimes
 \sum_{y}\alpha_{x,y}\ket{\varphi_{x,y}}^{M}\ket{y}^{Y}, 
\end{eqnarray*}
where for all $x$, $\sum_y |\alpha_{x,y}|^2=1$.
Then, a measurement in the computational basis is applied
to register $Y$ providing classical outcome $y$. The result
is a quantum state in register $M$ of size $\gamma n$ qubits.
Ignoring the value of $y$ to ease the notation,
the re-normalized state of the system in its most general
form when Step~\ref{bound} in \eprqot\ is reached is thus of the form
\[ \ket{\psi}= 
\sum_{x\in\{0,1\}^n} \alpha_x \ket{x}\otimes\ket{\varphi_{x}}^M,
\]
where $\sum_{x} |\alpha_x|^2=1$. We will prove security for any such state $\ket{\psi}$ and thus conditioned on any value $y$ that may be observed. It is therefore safe to leave the dependency on $y$ implicit.

\section{Security Against Dishonest Receivers}\label{sec:otsecurity}
In this section, we use the uncertainty relation derived in
Section~\ref{sec:twounbiasedbases} to show that \eprqot\ is secure
against any dishonest receiver having access to a quantum storage
device of size strictly smaller than half the number of qubits
received at Step \ref{rec}.

\index{sender-security!of quantum \RabinOT}
\begin{theorem}\label{thm:privacyeprqot}
For all $\gamma < \frac12$, \qot\ is $\eps$-secure for a negligible (in $n$) $\eps$ 
against $\mathfrak{R}_{\gamma}$.
\end{theorem}
\begin{proof}
  After Lemmas~\ref{lem:seqequiv} and~\ref{prop:sec:receiverprivate},
  it remains to show that \eprqot\ is $\eps$-sender-secure against
  $\mathfrak{R}_{\gamma}$.  Since \mbox{$\gamma < \frac{1}{2}$}, we
  can find $\sp > 0$ with \mbox{$\gamma + \sp < \frac{1}{2}$}.
  Consider a dishonest receiver \smash{$\dB$} in \eprqot\ with quantum
  memory of size~$\gamma n$.  Let $R$ and $X$ denote the random
  variables describing the basis $r$ and the outcome $x$ of \A's
  measurement (in basis~$r$) in Step~\ref{it:measure} of \eprqot,
  respectively.  We implicitly understand the distribution of $X$
  given $R$ to be conditioned on the classical outcome $y$ of the
  measurement \smash{$\dB$} performed when the memory bound applies,
  as described in Section~\ref{sec:modeldishonestreceiversrabin}; the
  following analysis works no matter what $y$ is.
  Corollary~\ref{cor:hadamard} with $\lambda = \gamma + \sp$ implies
  the existence of $\eps$ negligible in $n$ and an event $\ev$ such
  that $\P[\ev] \geq \frac12 - \eps$ and such that
  $\H_{\infty}(X|R\!=\!r,\ev) \geq \gamma n + \sp n$ for any relevant
  $r$.  Note that by construction, the random variables $X$ and $R$,
  and thus also the event $\ev$, are independent of the sender's input
  bit $B$, and hence $\rho_{B|\ev} = \rho_B$.  It remains to show that
  $\dist{ \rho_{B\dB|\ev},\rho_{B|\ev}\!\otimes\!\rho_{\dB|\ev} } \leq
  \eps$.  As the bit $B$ is masked by the output of the
  two-universal hash function $F(X)$ in Step 4 of \eprqot\ (where the
  random variable $F$ represents the random choice for~$f$), it
  suffices to show that $F(X)$ is close to uniform and essentially
  independent from $\dB$'s view, conditioned on $\ev$.  But this is
  guaranteed by the above bound on $\H_{\infty}(X|R\!=\!r,\ev)$ and by
  the \index{privacy amplification} privacy-amplification theorem
  (Corollary~\ref{thm:pasmooth} with $\eps \assign 0, \ell \assign 1,
  q \assign \gamma n$ and $U$ constant).
\end{proof}

\section{On the Necessity of Privacy Amplification}
\index{privacy amplification!randomized}
In this section, we show that randomized privacy amplification is
needed for protocol \qot\ to be secure.
For instance, it is tempting to believe that the sender could use the
XOR $\bigoplus_i x_i$ in order to mask the bit $b$, rather than $f(x)$
for a randomly sampled $f \in \chf{n}$. This would reduce the
communication complexity as well as the number of random coins needed.
However, we argue in this section that this is not secure (against an
adversary as we model it). Indeed, somewhat surprisingly, this variant
can be broken by a dishonest receiver that has {\em no quantum memory
 at all} (but that can do coherent measurements on pairs of qubits) in
the case $n$ is even. For odd $n$, the dishonest receiver needs to
store \emph{a single qubit}.

Clearly, a dishonest receiver can break the modified scheme \qot\ and
learn the bit $b$ with probability $1$ if he can compute $\bigoplus_i
x_i$ with probability~$1$. Note that, using the equivalence between
\qot\ and \eprqot, $x_i$ can be understood as the outcome of the
\index{basis!computational} \index{basis!diagonal}
measurement in either the $+$- or the $\times$-basis, performed by the
sender on one part of an EPR pair while the other is handed over
to the receiver.  The following proposition shows that indeed the
receiver can learn $\bigoplus_i x_i$ by a suitable measurement of his
parts of the \index{EPR pair}EPR pairs. Concretely, he measures the qubits he receives
pair-wise by a suitable measurement which allows him to learn the XOR
of the two corresponding $x_i$'s, no matter what the basis is (and he
needs to store one single qubit in case $n$ is odd). This obviously
allows him to learn the XOR of all $x_i$'s in all cases.

\begin{proposition}
  Consider two EPR pairs, i.e., $\ket{\psi} = \frac{1}{2} \sum_{x}
  \ket{x}^S\ket{x}^R$ where $x$ ranges over $\set{0,1}^2$.  Let $r \in
  \set{+,\times}$, and let $x_1$ and $x_2$ be the result when
  measuring the two qubits in register $S$ in basis $r$.  There exists
  a fixed measurement for register $R$ so that the outcome together
  with $r$ uniquely determines $x_1 \oplus x_2$.
\end{proposition}

\begin{proof}
\index{Bell measurement} \index{Bell basis}
The measurement that does the job is the {\em Bell measurement}, i.e., the measurement in the Bell basis $\set{\ket{\Phi^+},\ket{\Psi^+},\ket{\Phi^-},\ket{\Psi^-}}$. Recall, 
\begin{align*}
\ket{\Phi^+} &= \frac{1}{\sqrt{2}} \big(\ket{00}_+ + \ket{11}_+\big) =  \frac{1}{\sqrt{2}} \big(\ket{00}_{\times} + \ket{11}_{\times}\big) \\
\ket{\Psi^+} &= \frac{1}{\sqrt{2}} \big(\ket{01}_+ + \ket{10}_+\big) =  \frac{1}{\sqrt{2}} \big(\ket{00}_{\times} - \ket{11}_{\times}\big) \\
\ket{\Phi^-} &= \frac{1}{\sqrt{2}} \big(\ket{00}_+ - \ket{11}_+\big) =  \frac{1}{\sqrt{2}} \big(\ket{01}_{\times} + \ket{10}_{\times}\big) \\
\ket{\Psi^-} &= \frac{1}{\sqrt{2}} \big(\ket{01}_+ - \ket{10}_+\big) =  \frac{1}{\sqrt{2}} \big(\ket{10}_{\times} - \ket{01}_{\times}\big)  \, .
\end{align*}
Due to the special form of the Bell basis, when register $R$ is
measured and, as a consequence, one of the four Bell states is
observed, the state in register $S$ collapses to that {\em same} Bell
state. Indeed, when doing the basis transformation, all cross-products
cancel each other out. It now follows by inspection that knowledge of
the Bell state and the basis $r$ allows to predict the XOR of the two
bits observed when measuring the Bell state in basis $r$. For
instance, for the Bell state $\ket{\Psi^+}$, the XOR is $1$ if $r = +$
and it is $0$ if $r = \times$.
\end{proof}

Note that from the proof above, one can see that the receiver's attack,
respectively his measurement on each pair of qubits, can be understood
as teleporting one of the two entangled qubits from the receiver to
the sender using the other as EPR pair. However, the receiver does not send
the outcome of his measurement to the sender, but keeps it in order to
predict the XOR. \index{teleportation}

Clearly, the same strategy also works against any fixed linear
function. Therefore, the only hope for doing deterministic privacy
amplification is by using a non-linear function. However, it has been
shown recently by Ballester, Wehner, and Winter \cite{BWW06}, that
also this approach is doomed to fail in our scenario, because the
outcome of \emph{any fixed Boolean function} can be perfectly predicted by a
dishonest receiver who can store a single qubit and later learns the
correct basis $r \in \set{+,\times}$. \index{privacy amplification!deterministic}


\section{Weakening the Assumptions}\label{sec:weakass}
Observe that \qot\ requires error-free quantum communication, in that
a transmitted bit $b$, that is encoded by the sender and measured by
the receiver using the same basis, is always received as $b$. In
addition, it also requires a perfect quantum source which on request
produces {\em one and only one} qubit in the right state, e.g.\ {\em
  one} photon with the right polarization.  Indeed, in case of noisy
quantum communication, an honest receiver in \qot\ is likely to
receive an incorrect bit, and the sender-security of \qot\ is
vulnerable to imperfect sources that once in a while transmit more
than one qubit in the same state: a malicious receiver $\dB$ can
easily determine the basis $r \in \{+,\times \}$ and measure all the
following qubits in the right basis.  However, current technology only
allows to approximate the behavior of single-photon sources and of
noise-free quantum communication. It would be preferable to find a
variant of \qot\ that allows to weaken the technological requirements
put upon the honest parties. \index{single-photon source}
\index{imperfect source} \index{multi-qubit emission}

In this section, we present such a protocol based on BB84 states
\cite{BB84}, \BBqot\ (see Figure~\ref{fig:BB84ot}). The security proof
follows essentially by adapting the security analysis of \qot\ in a
rather straightforward way, as will be discussed later.
\index{BB84 coding scheme}

\subsection{Weak Quantum Model}
\index{weak quantum model} \index{$(\phi,\eta)$-weak quantum model|see
{weak quantum model}}
Let us consider a quantum channel with an error probability $\phi <
\frac{1}{2}$, i.e., $\phi$ denotes the probability that a transmitted
bit $b$, that is encoded by the sender and measured by the receiver
using the same basis, is received as $1-b$. In order not to have the
security rely on any level of noise, we assume 
the error probability to be zero when considering a {\em dishonest} receiver.
Also, let us consider a quantum source which produces two or more
qubits (in the same state), rather than just one, with probability
$\eta < 1 - \phi$.
We call this the \emph{$(\phi,\eta)$-weak quantum model}. By adjusting the
parameters, this model can also cope with dark counts and empty
pulses, see Section~\ref{sec:moreimperfect}.
\index{dark count} \index{empty pulse}

In order to deal with noisy quantum communication, we need to do
error-correction without giving the adversary too much information. 
\remove{ 
For this, we use {\em secure sketches}, as introduced in~\cite{DRS04}.
A $(\ell,m,\phi)$-secure sketch\footnote{Note that our definition of a
  secure sketch differs slightly from the one given in~\cite{DRS04}. }
is a randomized function $S:\{0,1\}^{\ell} \rightarrow \{0,1\}^*$ such
that (1) for any $w \in \{0,1\}^{\ell}$ and for $w'$ received from $w$
by flipping each bit (independently) with probability $\phi$, the
string $w$ can be recovered from $w'$ and $S(w)$ except with
negligible probability (in $\ell$), and (2) for all random variables
$W$ over $\{0,1\}^{\ell}$, the ``average min-entropy'' of $W$ given
$S(W)$ is at least $\H_{\infty}(W) - m$.  We would like to point out
that the notion of average min-entropy used in \cite{DRS04} and here differs
slightly from the standard notion $\H_{\infty}(W|S(W))$, but it implies
that for any $\Delta > 0$, the probability that $S(W)$ takes on a
value $y$ such that $\H_{\infty}(W|S(W)=y) \geq \H_{\infty}(W)-m-\Delta$
is at least $1-2^{-\Delta}$ (which is sufficient for our purpose).
}
Techniques to solve this problem are known as {\em information
  reconciliation} (as introduced for instance by Brassard and 
Salvail~\cite{BS93}) or as {\em secure sketches} introduced by Dodis,
Reyzin, Smith~\cite{DRS04}.  Let $x \in \{0,1\}^\ell$ be an
arbitrary string, and let $x' \in\{0,1\}^\ell$ be the result of
flipping every bit in $x$ (independently) with probability $\phi$. It
is well known that learning the syndrome $S(x)$ of $x$, with respect
to a suitable efficiently-decodable linear error-correcting code $C$
of length $\ell$, allows to recover $x$ from $x'$, except with negligible probability
in~$\ell$ (see,e.g.,~\cite{Maurer91,Crepeau97,DRS04}).  Furthermore,
it is known from coding theory that, for large enough $\ell$, such a
code can be chosen with rate $R$ arbitrarily close to but smaller than
$1 - h(\phi)$, i.e., such that the syndrome length $s$ is bounded by
$s < (h(\phi) + \varepsilon) \ell$ where $\varepsilon > 0$ (see
e.g.~\cite{Crepeau97} or the full version of~\cite{DRS04} and the
references therein). \index{error correction} \index{syndrome} \index{information reconciliation} \index{secure sketch} \index{error probability} \index{code}  

Regarding the loss of information, we can use the
\index{privacy amplification} privacy-amplification statement in form of
Corollary~\ref{thm:pasmooth} with $\eps \assign 0$ and constant $U$ in
a similar way as before, just by appending the classical syndrome
$S(x)$ (of length $s$) to the quantum register $\regE$, which results
in
%
\begin{align}  \label{ddbound}
  \dist{ \rho_{F(X) F S(X) \regE},\I \otimes \rho_{F S(X) \regE} }
&\leq \frac{1}{2} 2^{-\frac{1}{2}({\H_{\infty}(X)-q-s-1})}.
\end{align}

Consider the protocol \BBqot\ shown in Figure~\ref{fig:BB84ot} in the
$(\phi,\eta)$-weak quantum model. The protocol uses an efficiently
decodable linear code $C_{\ell}$, parametrized in $\ell \in
\naturals$, with codeword length $\ell$,
rate $R = 1 - h(\phi) - \varepsilon$ for some small $\varepsilon > 0$,
and being able to correct errors occurring with probability $\phi$
(except with negligible probability). Let $S_{\ell}$ be the
corresponding syndrome function.  Like before, the memory bound in
\BBqot\ applies before Step~\ref{BBbound}.

\begin{myfigure}{h}
\begin{myprotocol}{\BBqot$(b)$}
 \item $\A$ picks $x \in_R \nbit$ and $\theta \in_R \{+,\times \}^n$
 and sends $x_i$ in the corresponding bases $\ket{x_1}_{\theta_1}, \ldots, \ket{x_n}_{\theta_n}$ to~$\B$.
\item $\B$ picks $r'\in_R\{+,\times \}$ and measures all qubits in
    basis $r'$. Let $x'\in\{0,1\}^n$ be the result.
\item $\A$ picks $r \in_R \{+,\times \}$, sets $I \assign \Set{i}{\theta_i \!=\!\set{+,\times}_{[r]}}$ and $\ell \assign |I|$, and announces $r$, $I$, $syn \assign S_{\ell}(x|_{I})$, $\hf\in_R \chf{\ell}$, and $e \assign b\oplus \hf(x|_{I})$.\label{BBbound}
\item $\B$ recovers $x|_{I}$ from $x'|_{I}$ and $syn$, and outputs
  $a \assign 1$ and $b' \assign e \oplus \hf(x|_{I})$ if $r'=r$ and else $a
  \assign 0$ and $b' \assign 0$.
\end{myprotocol}
\caption{Protocol for the BB84 version of \RabinOT}\label{fig:BB84ot}
\end{myfigure}

By the above mentioned properties of the code $C_{\ell}$, it is
obvious that $\B$ receives the correct bit $b$ if $r'=r$, except with
negligible probability. \index{correctness!of quantum \RabinOT} (The
error probability is negligible in $\ell$, but by Chernoff's
\index{Chernoff's inequality} inequality (Lemma~\ref{lem:chernoff}),
$\ell$ is linear in $n$ except with negligible probability.)  Also,
since there is no communication from $\B$ to $\A$, a dishonest sender
$\dA$ cannot learn whether $\B$ received the bit. In fact, \BBqot\ can
\index{receiver-security!of quantum \RabinOT} \index{purification}
be shown perfectly receiver-secure in the same way as in
Proposition~\ref{prop:sec:receiverprivate}.  Similar as for protocol
\qot, in order to argue about sender-security we compare \BBqot\ with
a purified version shown in Figure~\ref{fig:BB84eprot}.  \BBeprqot\ 
runs in the $(\phi,0)$-weak quantum model, and the imperfectness of
the quantum source assumed in \BBqot\ is simulated by $\A$ in
\BBeprqot\ so that there is no difference from $\B$'s point of view.
\index{weak quantum model}

\begin{myfigure}{h}
\begin{myprotocol}{\BBeprqot$(b)$}
\item $\A$ prepares $n$ EPR pairs each in state
  $\ket{\Omega}=\frac{1}{\sqrt{2}}(\ket{00}+\ket{11})$. Additionally,
  $\A$ initializes $I'_{+}:= \emptyset$ and $I'_{\times}:= \emptyset$.
  For every $i \in \{1,\ldots,n\}$, $\A$ does the following. With
  probability $1-\eta$, $\A$ sends one half of the $i$-th pair to $\B$
  and keeps the other half. While with probability $\eta$, $\A$ picks
  $\theta_i \in_R \set{+,\times}$, replaces $I'_{\theta_i}$ by
  $I'_{\theta_i} \cup \{i\}$ and sends two or more qubits in the same
  state $\ket{x_i}_{\theta_i}$ to $\B$ where $x_i \in_R \{0,1\}$.
\item $\B$ picks $r'\in_R\{+,\times \}$ and measures all received
  qubits in basis $r'$. Let $x'\in\{0,1\}^n$ be the result.
\item\label{it:abort} $\A$ picks a random index set $J \subset_R
  \{1,\ldots,n\} \setminus (I'_{+}\cup I'_{\times})$. Then, it picks $r \in_R \{+,\times \}$, sets $I \assign J \cup I'_r$ and $\ell \assign |I|$, and for each $i \in J$ it measures the corresponding qubit
  in basis $r$. Let $x_i$ be the corresponding outcome, and let
  $x|_{I}$ be the collection of all $x_i$'s with $i \in
  I$. $\A$ announces $r$, $I$, $syn =
  S_{\ell}(x|_{I})$, $\hf\in_R \chf{\ell}$, and $e = b\oplus
  \hf(x|_{I})$.
\item $\B$ recovers $x|_{I}$ from $x'|_{I}$ and
  $syn$, and outputs $a \assign 1$ and $b' \assign e \oplus \hf(x|_{I})$, if
  $r'=r$ and else $a \assign 0$ and $b' \assign 0$.
\end{myprotocol}
\caption{Protocol for EPR-based \RabinOT, BB84 version}\label{fig:BB84eprot}
\end{myfigure}
The security equivalence between \BBqot\ (in the $(\phi,\eta)$-weak
quantum model) and \BBeprqot\ (in the $(\phi,0)$-weak quantum model)
follows along the same lines as in Section~\ref{sec:otprot}. 

\index{sender-security!of quantum \RabinOT}
\begin{theorem}\label{thm:BBqotsec}
  In the $(\phi,\eta)$-weak quantum model, \BBqot\ is $\eps$-secure
  with $\eps$ negligible in $n$ against $\mathfrak{R}_{\gamma}$ for
  any $\gamma < \frac{1-\eta}{4} - \frac{h(\phi)}{2}$ and $n$ large
  enough.
\end{theorem}
\begin{sketch}
  It remains to show that \BBeprqot\ is sender-secure against
  $\mathfrak{R}_{\gamma}$ (in the $(\phi,0)$-weak quantum model).  The
  reasoning goes analogous to the proof of
  Theorem~\ref{thm:privacyeprqot}, except that we restrict our
  attention to those $i$'s which are in $J$.  By Chernoff's inequality
  (Lemma~\ref{lem:chernoff}), $\ell$ lies within $(1\pm\varepsilon)n/2$ and $|J|$
  within $(1-\eta\pm\varepsilon)n/2$ except with negligible
  probability. In order to make the proof easier to read, we assume
  that $\ell = n/2$ and $|J| = (1-\eta)n/2$, and we also treat the
  $\varepsilon$ occurring in the rate of the code $C_{\ell}$ as zero.
  For the full proof, we simply need to carry the $\varepsilon$'s
  along, and then choose them small enough at the end of the proof.

Write $n' = |J| = (1-\eta)n/2$, and let $\gamma'$ be such that $\gamma n =
  \gamma' n'$, i.e., $\gamma' = 2\gamma/(1-\eta)$. 
Assume $\kappa > 0$ such that $\gamma' + \kappa < \frac12$, where we make sure later that such $\kappa$ exists. 
It then follows from Corollary~\ref{cor:hadamard} that there exists an 
event $\ev$ such that $\P[\ev] \geq \frac12 - \negl{n'} = \frac12 - \negl{n}$ and 
$$
\H_{\infty}\big(X|_J \big| R\!=\!r,\ev\big) \geq (\gamma' +
 \kappa) n' 
 = \gamma n + \kappa (1-\eta)n/2 \, .
$$
By Inequality~(\ref{ddbound}), it remains to argue that this is larger than $q + s = \gamma n + h(\phi) n/2$, i.e., 
$$
\kappa (1-\eta) > h(\phi) \, ,
$$ 
where $\kappa$ has to satisfy 
$$
\kappa < \frac12 -
\gamma' = \frac12 - 2\gamma/(1-\eta) \, .
$$
This can obviously be achieved (by choosing $\kappa$ appropriately) if and only if
the claimed bound on $\gamma$ holds. 
\end{sketch}

\section{\RabinOT of Strings} \label{sec:extensionrot}
In this chapter, we only considered \RabinOT of one bit per invocation.
Our technique can easily be extended to deal with \RabinlStringOT of
$\ell$-bit strings, essentially by using a class of two-universal
functions \index{two-universal hashing} with range $\set{0,1}^{\ell
  n}$ rather than $\set{0,1}$, for some $\ell$ with $\gamma + \ell <
\frac{1}{2}$ (respectively $<\frac{1-\eta}{4}-\frac{h(\phi)}{2}$ for
\BBqot).


\index{oblivious transfer!\RabinOT|)}


%% file: 12ot.tex
\chapter[\OT in the Bounded-Quantum-Storage Model]
{\OT in the Bounded-Quantum-Storage Model}\label{chap:12OT}
\index{oblivious transfer!\OT|(}
In the last chapter, we have shown how to construct \RabinOT securely
in the \index{bounded-quantum-storage model}bounded-quantum-storage
model. Although other flavors of \pOT
can be constructed from \RabinOT using standard reductions, a more
direct approach gives a better ratio between storage-bound and
communication-complexity.

In this chapter, we present an efficient protocol for \onetwo\
Oblivious Transfer secure in the bounded-quantum-storage model. The protocol
is very close to Wiesner original \index{conjugate coding}''conjugate-coding'' protocol
\cite{Wiesner83} from the early 70's. The uncertainty relation from
Section~\ref{sec:morerelation} will be extensively used for proving the
security. 

The results of this section appeared in~\cite{DFRSS07}.

\section{The Definition} 
In \lStringOT, the sender Alice sends two $\ell$-bit strings $S_0,
S_1$ to the receiver Bob in such a way that Bob can choose which
string he wants to receive, but does not learn anything about the
other. Alice does not get to know which string Bob has chosen. As
explained in Chapter~\ref{chap:ClassicalOT}, the common way to build
\lStringOT is by constructing a protocol for
\mbox{(Sender-)}Randomized \lStringOT, which then can easily be
converted into an ordinary \lStringOT. \RandlStringOT essentially
coincides with ordinary \lStringOT, except that the two strings $S_0$
and $S_1$ are not \emph{input} by the sender but generated uniformly
at random during the protocol and \emph{output} to the sender.



For the formal definition of the security requirements for a quantum
protocol for \RandlStringOT, we translate the classical
Definition~\ref{def:RandOT} to the quantum setting using a similar
notation as for the definition of \RabinOT in
Section~\ref{sec:def-rabin-obliv-transf}: Let $C$ denote the binary
random variable describing receiver \R's choice bit, let $S_0, S_1$
denote the $\ell$-bit long random variables describing sender $\S$'s
output strings, and let $Y$ denote the $\ell$-bit long random variable
describing $\R$'s output string (supposed to be $S_C$).  Furthermore,
for a fixed candidate protocol for \RandlStringOT, and for a fixed
input distribution for $C$, the overall quantum state in case of a
dishonest sender $\dS$ is given by the ccq-state $\rho_{C Y \dS}$.
Analogously, in the case of a dishonest receiver $\dR$, we have the
ccq-state $\rho_{S_0 S_1 \dR}$.


\begin{definition}[\RandlStringOT] \label{def:Rl12OT}
  An $\varepsilon$-secure {\em \RandlStringOT} is a quantum protocol
  between $\S$ and $\R$, with $\R$ having input $C \in \{0,1\}$ while
  $\S$ has no input, such that for any distribution of $C$, the
  following holds:
\begin{description}
\item[\boldmath$\eps$-Correctness:] If\ \S\ and\ \R\ follow the
  protocol, then $\S$ gets output strings $S_0,S_1 \in \{0,1\}^{\ell}$ and
  $\R$ gets $Y = S_C$ except with probability~$\varepsilon$.
\index{correctness!of quantum \RandlStringOT}
\item[\boldmath$\varepsilon$-Receiver-security:] If $\R$ is honest, then
  for any $\dS$, there exist\footnote{Recall from
    Section~\ref{sec:qit}: Given a cq-state $\rho_{X \regE}$,
    by saying that there exists a random variable $Y$ such that
    $\rho_{XY\regE}$ satisfies some condition, we mean that $\rho_{X
      \regE}$ can be understood as $\rho_{X\regE} =
    \tr_Y(\rho_{XY\regE})$ for a ccq-state $\rho_{XY\regE}$ that
    satisfies the required condition.} random variables $S'_0$ and
  $S'_1$ such that $\Pr\big[Y=S'_C\big] \geq 1-\varepsilon$ and
$$
  \dist{ \rho_{C S'_0 S'_1 \dS} , \rho_C \otimes \rho_{S'_0 S'_1 \dS}
  } \leq \eps \, .
$$
\item[\boldmath$\varepsilon$-Sender-security:] If $\S$ is honest, then for
  any $\dR$, there exists a random variable $D \in \set{0,1}$ such that
$$
  \dist{ \rho_{S_{1-D} S_{D} D \dR} , \id \otimes \rho_{S_{D}
      D \dR} } \leq \eps \, .
$$
\end{description}
\index{sender-security!of quantum \RandlStringOT}
\index{receiver-security!of quantum \RandlStringOT}
If any of the above holds for $\varepsilon = 0$, 
then the corresponding property is said to hold {\em perfectly}. 
If one of the properties only holds with respect to a restricted class
$\mathfrak{S}$ of \dS's respectively $\mathfrak{R}$ of \dR's, then this property
is said to hold and the protocol is said to be secure {\em against}
$\mathfrak{S}$ respectively~$\mathfrak{R}$.
\end{definition}

Receiver-security, as defined here, implies that whatever a dishonest
sender does is as good as the following: generate the ccq-state
$\rho_{S_0' S'_1 \dS}$ independently of $C$, let $\R$ know $S'_C$, and
output $\rho_{\dS}$.  On the other hand, sender-security implies that
whatever a dishonest receiver does is as good as the following:
generate the ccq-state $\rho_{S_{D} D \dR}$ arbitrarily, let $\S$
know $S_{D}$ and an independent uniformly distributed $S_{1-D}$, and
output $\rho_{\dR}$.  In other words, a protocol satisfying
Definition~\ref{def:Rl12OT} is a secure implementation of the natural
\RandlStringOT ideal functionality, except that it allows a dishonest
sender to influence the distribution of $S_0$ and $S_1$, and the
dishonest receiver to influence the distribution of the string of his
choice. This is in particular good enough for constructing a standard
\lStringOT in the straightforward way.

We would like to point out the importance of requiring the existence
of $S_0'$ and $S_1'$ in the formulation of receiver-security in a
quantum setting: requiring only that the sender learns no information
on $C$, as is sufficient in the classical setting
(see~e.g.~\cite{CSSW06}), does not prevent a dishonest sender from
obtaining $S_0,S_1$ by a suitable measurement {\em after} the
execution of the protocol in such a way that he can choose $S_0 \oplus
S_1$ at will, and $S_C$ is the string the receiver has obtained in the
protocol. This would for instance make the straightforward
construction of a \index{bit commitment}bit commitment
\footnote{The committer sends two random bits of parity
  equal to the bit he wants to commit to, the verifier chooses to
  receive at random one of those bits.} based on
\OT[2] insecure.

\section{The Protocol}\label{sec:otprot12}
We present a quantum protocol for \RandlStringOT that will be shown
perfectly receiver-secure against any sender and statistically
sender-secure against any quantum-memory-bounded receiver. The first
two steps of the protocol are identical to Wiesner's ``conjugate
coding'' \index{conjugate coding} protocol~\cite{Wiesner83} from circa 1970 for
\emph{``transmitting two messages either but not both of which may be
  received''}.

The simple protocol is described in Figure~\ref{fig:Randlqot}.  
The sender $\S$ sends random \index{BB84 coding scheme}BB84 states to the receiver $\R$, who
measures all received qubits according to his choice bit $C$. $\S$
then picks randomly two functions from a fixed 
\index{two-universal hashing} \univ class of hash
functions $\chf{n}$ from $\set{0,1}^n$ to $\set{0,1}^{\ell}$, where $\ell$
is to be determined later, and applies them to the bits encoded in the
$+$-basis respectively the bits encoded in $\times$-basis to obtain the
output strings $S_0$ and $S_1$. Note that we may apply a function $\hf
\in \chf{n}$ to a $n'$-bit string with $n' < n$ by padding it with
zeros\footnote{Recall the notation for padding $x \pad_I$ introduced
  in Section~\ref{sec:notation}.}
(which does not decrease its entropy).
$\S$ announces the encoding bases and the hash
functions to the receiver who then can compute $S_C$. Intuitively, a
dishonest receiver who cannot store all the qubits until the right
bases are announced will measure some qubits in the wrong basis and thus cannot learn both strings simultaneously.  

\begin{myfigure}{h}
\begin{myprotocol}[Let $c$ be $\R$'s choice bit.]{\Randlqot}
\item $\S$ picks $x \in_R \nbit$ and $\theta \in_R \{+,\times \}^n$ 
and sends $\ket{x_1}_{\theta_1},
   \ket{x_2}_{\theta_2}, \ldots, \ket{x_n}_{\theta_n}$ to $\R$. 
\item $\R$ measures all qubits in basis $[+,\times]_{c}$. 
Let $x' \in \{0,1\}^n$ be the result.
\item $\S$ picks two hash functions $\hf_0,\hf_1 \in_R \chf{n}$,
  announces $\theta$ and $\hf_0, \hf_1$ to $\R$, and outputs $s_0
  \assign \hf_0(x \pad_{I_0})$ and $s_1 \assign \hf_1(x \pad_{I_1})$ where
  $I_b \assign \Set{i}{\theta_i \!=\! [+,\times]_b}$. \label{bound12}
\item $\R$ outputs $s_{c} = \hf_{c}(x' \pad_{I_{c}})$.
\end{myprotocol}
\caption{Quantum Protocol for \boldmath\RandlStringOT.}\label{fig:Randlqot}
\end{myfigure}

We would like to stress that although protocol description and
analysis are designed for an ideal setting with perfect noiseless
quantum communication and with perfect sources and detectors, all our
results can easily be extended to a more realistic noisy setting along the
same lines as in the previous Chapter~\ref{chap:RabinOT}.
\index{weak quantum model}

It is clear by the non-interactivity of \Randlqot\ that a
dishonest sender cannot learn anything about the receiver's choice
bit. Below, we show \Randlqot\ perfectly receiver-secure
according to Definition~\ref{def:Rl12OT}.
\begin{proposition}\label{prop:sec:oblivious}
\Randlqot\ is perfectly receiver-secure.
\end{proposition}
\begin{proof}
  Recall that the ccq-state $\rho_{C Y \dS}$ is defined by the experiment
  where $\dS$ interacts with the honest memory-bounded $\R$. We now
  define (in a new Hilbert space) the ccccq-state $\hat{\rho}_{\hat{C}
    \hat{Y} \hat{S}'_0 \hat{S}'_1 \dS}$ by a slightly different
  experiment: We let $\dS$ interact with a receiver with {\em
    unbounded} quantum memory, which waits to receive $\theta$ and
  then measures the $i$-th qubit in basis $\theta_i$ for
  $i=1,\ldots,n$. Let $X$ be the resulting string, and define
  $\hat{S}'_0 = f_0(X \pad_{I_0})$ and $\hat{S}'_1 = f_1(X
  \pad_{I_1})$.  Finally, sample $\hat{C}$ according to $P_C$ and set
  $\hat{Y} = \hat{S}'_C$.  It follows by construction that
  \smash{$\Pr\!\big[\hat{Y} \!\neq\! \hat{S}'_{\hat{C}}\big] = 0$} and
  $\hat{\rho}_{\hat{C}}$ is independent of $\hat{\rho}_{\hat{S}'_0
      \hat{S}'_1 \dS}$.  It remains to argue that $\hat{\rho}_{\hat{C}
      \hat{Y} \dS} = \rho_{C Y \dS}$, so that corresponding $S'_0$ and
    $S'_1$ also exist in the original experiment. But this is
    obviously satisfied since the only difference between the two
    experiments is when and in what basis the qubits at position $i
    \in I_{1-C}$ are measured, which, once $C$ is fixed, cannot
    influence $\rho_{Y \dS}$ respectively $\hat{\rho}_{\hat{Y} \dS}$.
\end{proof}


\section{Security Against Dishonest Receivers} \label{sec:modeldishonestreceivers}
As in Section~\ref{sec:modeldishonestreceiversrabin}, we model
\index{dishonest receiver!in \Randlqot}dishonest receivers 
in \Randlqot\ under the assumption that the
maximum size of their quantum storage is bounded.  Such adversaries
are only required to have bounded quantum storage when Step
\ref{bound} in \Randlqot\ is reached. Before and after that, the
adversary can store and carry out arbitrary quantum computations
involving any number of qubits. Apart from the restriction on the size
of the quantum memory available to the adversary, no other assumption
is made. In particular, the adversary is not assumed to be
computationally bounded and the size of its classical memory is not
restricted.

\begin{definition}\label{boundedstorage12}
  The set $\mathfrak{R}_{\gamma}$ denotes all possible quantum
  dishonest receivers $\dR$ in \Randlqot\ which
  have quantum memory of size at most $\gamma n$ when Step~\ref{bound12}
  is reached.
\end{definition}

First, we consider a purified version of \Randlqot, \eprRandlqot\ in
Figure~\ref{fig:eprRandlqot}, where $\S$ prepares an EPR pair
\smash{$\ket{\Phi} = \frac{1}{\sqrt{2}}(\ket{00}+\ket{11})$} instead
of $\ket{x_i}_{\theta_i}$ and sends one part to the receiver while
keeping the other. Only when Step~\ref{step:afterbound} is reached and $\dR$'s
quantum memory is bound to $\gamma n$ qubits, $\S$ measures her qubits
in basis $\theta \in_R \set{+,\times}^n$. It is easy to see that for
any $\dR$, \eprRandlqot\ is equivalent to the original \Randlqot, and
it suffices to prove sender-security for the former.  Indeed, $\S$'s
choices of $\theta$ and $\hf_0,\hf_1$, together with the measurements
all commute with $\R$'s actions.  Therefore, they can be performed
right after Step 1 with no change for $\R$'s view. Modifying
\eprRandlqot\ that way results in \Randlqot.

\begin{myfigure}{h}
\begin{myprotocol}{\eprRandlqot}
\item $\S$ prepares $n$ EPR pairs each in state 
      $\ket{\Omega}=\frac{1}{\sqrt{2}}(\ket{00}+\ket{11})$ 
and sends one half of each pair to $\R$ and keeps the other
  halves.
\item $\R$ measures all qubits in basis $[+,\times]_{c}$. Let $x' \in
  \{0,1\}^n$ be the result. 
\item $\S$ picks random $\theta \in_R \{+,\times \}^n$, and she
  measures the $i$th qubit in basis $\theta_i$. Let $x\in\{0,1\}^n$ be
  the outcome. $\S$ picks two hash functions $\hf_0, \hf_1 \in_R \chf{n}$,
  announces $\theta$ and $\hf_0, \hf_1$ to $\R$ and outputs $s_0
  \assign \hf_0(x \pad_{I_0})$ and $s_1 \assign \hf_1(x \pad_{I_1})$ where
  $I_b \assign \Set{i}{\theta_i \!=\! [+,\times]_b}$. \label{step:afterbound}
\item $\R$ outputs $s_{c} = \hf_{c}(x' \pad_{I_{c}})$.
\end{myprotocol}
\caption{Protocol for EPR-based \boldmath\RandlStringOT.}\label{fig:eprRandlqot}
\end{myfigure}

\begin{theorem}\label{thm:OT}
  \Randqot$^{\ell}$ is $\varepsilon$-secure against
  $\mathfrak{R_\gamma}$ for a negligible (in $n$) $\varepsilon$ if
  there exists $\delta>0$ such that $\gamma n \leq n/4 - 2 \ell - \delta n$.
\end{theorem}
The proof has the same structure as the security-proof for the
reduction \OTUOT\ described at the end of Section~\ref{sec:UOT}. The
\index{uncertainty relation}uncertainty relation
 from Section~\ref{sec:morerelation} lower bounds
the dishonest receiver's (smooth) min-entropy about the sender's $X$.
Hence, we have an (imperfect) \UOT{\infty}{\frac{n}{2}}{\set{0,1}^n}
from which we get an ordinary \RandlStringOT via the 
\index{min-entropy splitting lemma} min-entropy
splitting lemma and \index{privacy amplification}privacy amplification
 against quantum adversaries.
\begin{proof}
  Consider the ccq-state $\rho_{X \Theta \dR}$ in \eprRandlqot\ after
  $\dR$ has measured all but $\gamma n$ of his qubits, where $X$
  describes the outcome of the sender measuring her part of the state
  in random basis $\Theta$. Also, let $\Hf_0$ and $\Hf_1$ be the
  random variables that describe the random and independent choices of
  $\hf_0,\hf_1 \in \chf{n}$.  Finally, let $X_b$ be $X_b =
  X\pad_{\Set{i}{\Theta_i = [+,\times]_b}}$ (padded with zeros so it
  makes sense to apply $\Hf_b$).
  
  Choose $\lambda, \kappa$ all positive, but small enough
  such that (for large enough $n$) $$\gamma n \leq (1/4 - \lambda - \lambda' - \kappa) n -
  1 -2\ell .$$ From the uncertainty relation
  (Corollary~\ref{cor:uncertainty}), we know that
  $\hie{\eps}{X_0 X_1 | \Theta} \geq (1/2 - 2\lambda )n$ for $\eps$ exponentially small in $n$.
  Therefore, by the Min-Entropy Splitting Lemma~\ref{lemma:ESL}, there
  exists a binary random variable $D$ such that
\[\hie{\eps}{X_{1-D} D | \Theta} \geq (1/4 - \lambda)n 
.\] We denote by the
random variables $F_0,F_1$ Alice's choices of hash functions. It is
clear that we can condition (for free) on the independent $F_{D}$. We
write $S_{D} = F_{D}(X_{D})$, set $\eps'=2^{-\lambda' n}$, and use the
chain rule (Lemma~\ref{lem:chain}) to condition on $D,S_D$ as well.
\begin{align*}
  \hie{\eps+\eps'}{X_{1-D} &| \Theta F_D D S_{D}}\\
  & \geq \hie{\eps}{X_{1-D} D S_{D} | \Theta F_{D} } -
  \hmax(D S_{D} | \Theta F_{D} ) - \lambda'n \\
&\geq (1/4 -\lambda - \lambda')n - 1 -\ell \\
&\geq \gamma n +\ell +\kappa n,
\end{align*}
by the choice of $\lambda, \lambda', \kappa$.

\hspace{5mm}
We can now apply \index{privacy amplification}privacy amplification
 in form of Corollary~\ref{thm:pasmooth}
to obtain
\begin{align*}
\dist{ &\rho_{S_{1-D} F_{1-D}  \Theta F_{D} D S_D \dR}, 
\id \otimes \rho_{F_{1-D}  \Theta F_{D}  D S_D \dR} } \\
&\quad \leq \frac12 2^{-\frac12 \left( \hie{\eps+\eps'}{ X_{1-D} |
 \Theta S_{D} F_{D}D} - \gamma n - \ell \right)} + (\eps+\eps')\\
&\quad \leq \frac12 2^{-\frac12 \kappa n} + \eps + \eps',
\end{align*} 
which is negligible. This shows $\eps$-sender-security according to
Definition~\ref{def:Rl12OT}.
\end{proof}

\section{Extensions}
\subsection{\lStringOT with Longer Strings}
It is possible to extend recent techniques by Wullschleger
\cite{Wullschleger07} described in Section~\ref{sec:comparison} to the
quantum case and hence, the security of \Randqot$^{\ell}$ can be
proven against $\mathfrak{R_\gamma}$ if there exists $\delta>0$ such
that $\gamma n \leq n/4 - \ell -\delta n$.

\subsection{Weakening the Assumptions} \label{sec:weakmodel12ot}
As described in Section~\ref{sec:weakass} for \RabinOT, we can extend
protocol \Randqot to work in the $(\phi,\eta)$-weak quantum model.
\index{weak quantum model} \index{error correction} To enable the
receiver to recover from errors in the transmission, the sender \S\ 
additionally sends error-correcting information in
Step~\ref{step:afterbound}. The players agree beforehand on an
efficiently decodable error-correcting code of length $n/2$ with
syndrome length $s$ roughly $h(\phi)n/2$ as in
Section~\ref{sec:weakass}. Then, \S\ sends along the two syndromes of
$S(x|_{I_0})$ and $S(x|_{I_1})$ (where the $x|_{I_b}$ are padded with
0s or truncated to length $n/2$). It can be argued as for \RabinOT
that this will reduce the min-entropy by the length $s$ of the
syndrome and hence, we can show 
\index{sender-security!of quantum \Randqot} sender-security
 of this protocol
against the class of receivers $\mathfrak{R_\gamma}$ with $\gamma$
such that there exists $\delta >0$ with
\[ \gamma n \leq \left( \frac{1-\eta}{4} - \frac{h(\phi)}{2} \right)n - 2 \ell -
\delta n \, .
\]

\subsection{Reversing the Quantum Communication}
In order to illustrate the versatility of our security analysis, we
show that the proofs carry easily over to a protocol where the
direction of the quantum communication is reversed. In the protocol
described in Figure~\ref{fig:RandlqotReversed}, the receiver \R\ of
the \RandOT sends $n$ qubits, encoded in the basis determined by his
choice bit.  The sender of the \RandOT $\S$ measures them in a random
basis.  The players then proceed as in \Randqot.
\begin{myfigure}{h}
\begin{myprotocol}[Let $c$ be $\R$'s choice bit.]{\Randlqot}
\item $\R$ picks $x' \in \set{0,1}^n$ at random and sends
  $\ket{x'}_{\theta'}$ to $\R$ where $\theta'=[+,\times]_{c}$.
\item $\S$ picks $\theta \in_R \{+,\times \}^n$ 
and measures the received qubits in basis $\theta$. Let $x \in
\set{0,1}^n$ be the result.
\item $\S$ picks two hash functions $\hf_0,\hf_1 \in_R \chf{n}$,
  announces $\theta$ and $\hf_0, \hf_1$ to $\R$, and outputs $s_0
  \assign \hf_0(x \pad_{I_0})$ and $s_1 \assign \hf_1(x \pad_{I_1})$ where
  $I_b \assign \Set{i}{\theta_i \!=\! [+,\times]_b}$. 
\item $\R$ outputs $s_{c} = \hf_{c}(x' \pad_{I_{c}})$.
\end{myprotocol}
\caption{\Randlqot\ with Reversed Quantum Communication.}\label{fig:RandlqotReversed}
\end{myfigure}

It is clear by construction that the protocol is perfectly
correct. $\eps$-Sender-security against dishonest receivers in $\mathfrak{R}_\gamma$ can
be argued as in Theorem~\ref{thm:OT} above by observing that the
uncertainty relation applies to any $n$-qubit state of the honest
sender which is measured in a random basis and about which the
dishonest receiver holds at most $\gamma n$ qubits of information.

For the security of an honest receiver against a dishonest sender, we
can show the existence of the two input strings as in
Proposition~\ref{prop:sec:oblivious} above by letting the sender
interact with an unbounded receiver. In an error-free model, it
further holds that the sender cannot infer the basis in which the
qubits are encoded and therefore does not learn any information about
the receiver's choice bit. However, in a more realistic setting with
multi-pulse emissions, this coding scheme with reversed communication
is highly insecure, as a malicious sender can determine the encoding
basis from a multi-pulse qubit. The same problem occurred for the
\RabinOT-protocol \qot\ from the last chapter.
\index{weak quantum model} \index{reversed quantum communication}
\index{multi-qubit emission}

\index{oblivious transfer!\OT|)}


%% file: commit.tex
\chapter{Quantum Bit Commitment} \label{chap:qbc}
\index{bit commitment|(}
This chapter is about quantum Bit Commitment (\BC) schemes. In \BC, a
committer $\C$ commits himself to a choice of a bit $b \in \set{0,1}$
by exchanging information with a verifier $\V$. We want that $\V$ does
not learn $b$ (we say the commitment is \emph{hiding}), yet $\C$ can later
choose to reveal $b$ in a convincing way, i.e., only the value fixed at
commitment time will be accepted by $\V$ (we say the commitment is
\emph{binding}).

In the next section, we present a \BC scheme from a committer $\C$ with
bounded quantum memory to an unbounded receiver $\V$. The scheme is
peculiar since in order to commit to a bit, the committer does not
send anything. During the committing stage, information only goes from
$\V$ to $\C$.  Therefore, there is no way for the verifier to get
information about the committed bit, i.e.~the scheme is perfectly
hiding.

In Section~\ref{sec:defbinding}, we define two notions of the binding
property and show our scheme secure against quantum-memory-bounded
committer in both of these senses. Similar techniques as in the two previous
chapters for the analysis of the oblivious-transfer protocols are
used.

The results in this chapter appeared in \cite{DFSS05, DFRSS07}.

\section{The Protocol} \label{sec:commprotocol}
The protocol is given in Figure~\ref{fig:comm}. Intuitively, a
commitment to a bit $b$ is made by measuring random BB84-states in
basis $\{+,\times\}_{[b]}$.
  
\begin{myfigure}{h}
\begin{myprotocol}{\comm$(b)$}
\item $\V$ picks $x \in_R \nbit$ and $\theta \in_R \{+,\times \}^n$ and
  sends $x_i$ in the corresponding bases $\ket{x_1}_{\theta_1},
  \ket{x_2}_{\theta_2}, \ldots, \ket{x_n}_{\theta_n}$ to~$\C$.
 \item $\C$ commits to the bit $b$ by measuring all qubits in basis
   $\{+,\times \}_{[b]}$. Let $x' \in \nbit$ be the result.
 \item\label{step:open} To open the commitment, $\C$ sends $b$ and $x'$ to $\V$.
 \item $\V$ verifies that $x_i= x_i'$ for those $i$ where $\theta_i =
   \{+,\times \}_{[b]}$. $\V$ accepts if and only if this is the case.
\end{myprotocol}
\caption{Protocol for quantum bit commitment}\label{fig:comm}
\end{myfigure}

As for the oblivious-transfer protocols in the two previous chapters, we present an
equivalent EPR-version of the protocol that is easier to analyze (see
Figure~\ref{fig:eprcomm}). \index{purification}
\begin{myfigure}{h}
\begin{myprotocol}{\eprcomm$(b)$}
\item $\V$ prepares $n$ EPR pairs each in state
  $\ket{\Omega}=\frac{1}{\sqrt{2}}(\ket{00}+\ket{11})$. $\V$ sends one half of each pair to $\C$ and keeps the other halves.
\item $\C$ commits to the bit $b$ by measuring all received qubits in basis
  $\{+,\times \}_{[b]}$. Let $x' \in \nbit$ be the result.
\item To open the commitment, $\C$ sends $b$ and $x'$ to $\V$.\label{step:epropen}
\item $\V$ measures all his qubits in basis $\{+,\times \}_{[b]}$ and
  obtains $x \in \nbit$.  He chooses a random subset $I \subseteq \{1,
  \ldots ,n\}$. $\V$ verifies that
  $x_i= x_i'$ for all $i \in I$ and accepts if and only if this is the
  case.\label{step:last}
\end{myprotocol}
\caption{Protocol for EPR-based quantum bit commitment}\label{fig:eprcomm}
\end{myfigure}

\begin{lemma}\label{lem:commeprcomm}
  \comm\ is secure against dishonest committers $\dC$ if and only if \eprcomm\ is.
\end{lemma}
\begin{proof}
The proof uses similar reasoning as the one for Lemma~\ref{lem:seqequiv}. 
First, it clearly makes no difference, if we change Step~\ref{step:last} to the
following:
\begin{itemize}
\item[\ref{step:last}'.] $\V$ chooses the subset $I$, measures all
  qubits with index in $I$ in basis $\{+,\times \}_{[b]}$ and all
  qubits not in $I$ in basis $\{+,\times \}_{[1-b]}$. $\V$ verifies
  that $x_i= x_i'$ for all $i\in I$ and accepts if and only if this is
  the case.
\end{itemize}
Finally, we can observe that the view of $\dC$ does not change if $\V$
would have done his choice of $I$ and his measurement already in
Step~1. Doing the measurements at this point means that the qubits to
be sent to $\dC$ collapse to a state that is distributed identically to
the state prepared in the original scheme. The EPR-version is
therefore equivalent to the original commitment scheme from $\dC$'s
point of view.
\end{proof}

It is clear that \eprcomm\ is \index{bit commitment!hiding}hiding, i.e., that the commit phase
reveals no information on the committed bit, since no information is
transmitted to $\V$ at all. Hence we have
\begin{lemma}\label{lem:sec:hiding}
\eprcomm\ is perfectly hiding.
\end{lemma}

\section{Modeling Dishonest Committers}\label{sec:dishonestcomm}
\index{dishonest committer|(} A dishonest committer $\dC$ with bounded
memory of at most $\gamma n$ qubits in \eprcomm\ can be modeled very
similarly to the dishonest oblivious-transfer receivers $\dB$ from
Section~\ref{sec:modeldishonestreceiversrabin} and
\ref{sec:modeldishonestreceivers}: $\dC$ consists first of a circuit
acting on all $n$ qubits received, then of a measurement of all but at
most $\gamma n$ qubits, and finally of a circuit that takes the
following input: a bit $b$ that $\dC$ will attempt to open, the
$\gamma n$ qubits in memory, and some ancilla in a fixed state. The
output is a string $x' \in \nbit$ to be sent to $\V$ at the opening
stage.
\begin{definition}
  We define $\mathfrak{C}_{\gamma}$ to be the class of all committers
  $\{\dC_n\}_{n>0}$ in \comm\ or \eprcomm\ that, at the start of the opening phase
  (i.e. at Step \ref{step:epropen}), have a quantum memory of size at most
  $\gamma n$ qubits.
\end{definition}
\index{dishonest committer|)}

\section{Defining the Binding Property} \label{sec:defbinding}
\index{bit commitment!binding|(}
\subsection{The ``Standard'' Binding Condition}
In the context of unconditionally secure \emph{quantum} bit
commitment, it is widely accepted that ``the right way'' of defining
the \emph{binding property} is to require that the probability of
opening a commitment successfully to 0 plus the probability of opening
it successfully to 1 is essentially upper bounded by one, put forward by
Dumais, Mayers, and Salvail \cite{DMS00}. We call this notion
\emph{weakly binding}, as opposed to the new notion of \emph{strongly
  binding} defined in the next section below.
\begin{definition} \label{def:weakbinding}
  A (quantum) bit-commitment scheme is \emph{weakly binding}
  against $\mathfrak{C}$ if for all $\{\dC_n\}_{n>0}\in \mathfrak{C}$,
  the probability $p_b(n)$ that $\dC_n$ opens $b\in\{0,1\}$ with
  success satisfies
\[ p_0(n)+p_1(n) \leq 1+\negl{n}.
\] 
\end{definition} \index{bit commitment!weak binding}
In the next Section~\ref{sec:weakbinding}, we show that \eprcomm\ is
weakly binding against $\mathfrak{C}_{\gamma}$ for any
$\gamma<\frac{1}{2}$.

Note that the binding condition given here in
Definition~\ref{def:weakbinding} is weaker than the classical one,
where one would require that a bit $b$ exists such that $p_b(n)$ is
negligible.  For a general quantum adversary though who can always
commit to 0 and 1 in superposition, this is a too strong requirement;
thus, it is typically argued that Definition~\ref{def:weakbinding} is the best
one can hope for. 

However, we argue now that this weaker notion is not really
satisfactory, and we show that there exists a stronger notion, which
still allows the committer to commit to a superposition and thus is
not necessarily impossible to achieve in a quantum setting, but which
is closer to the classical standard way of defining the
binding property. 

\subsection{A Stronger Binding Condition}\label{sec:strongerbinding}
A shortcoming of Definition~\ref{def:weakbinding} is that committing
bit by bit is not guaranteed to yield a secure string commitment---the
argument that one is tempted to use requires independence of the
$p_{b}$'s between the different executions, which in general does not
hold.

We now argue that this notion is {\em unnecessarily} weak, at least in
some cases, and in particular in the case of commitments in the
bounded-quantum-storage model where the dishonest committer is forced
to do some partial measurement and where we assume honest parties to
produce only classical output (by measuring their entire quantum
state).  Technically, this means that for any dishonest committer
$\tilde{\sf C}$, the joint state of the honest verifier and of
$\tilde{\sf C}$ after the commit phase is a ccq-state \smash{$\rho_{V
    Z \tilde{\sf C}} = \sum_{v,z} P_{VZ}(v,z) \proj{v} \otimes
  \proj{z} \otimes \rho_{\tilde{\sf C}}^{v,z}$}, where the first
register contains the verifier's (classical) output and the remaining
two registers contain $\tilde{\sf C}$'s (partially classical) output.
We propose the following definition.

\begin{definition}\label{def:strongbinding}
  A commitment scheme in the bounded-quantum-storage model is called
  {\em $\eps$-binding}, if for every (dishonest) committer $\tilde{\sf
    C}$, inducing a joint state $\rho_{V Z \tilde{\sf C}}$ after the
  commit phase, there exists a classical binary random variable $D$,
  given by its conditional distribution $P_{D|VZ}$, such that for
  $b=0$ and $b=1$ the state \smash{$\rho_{V Z \tilde{\sf C}}^{b} =
    \sum_v P_{VZ|D}(v,z|b) \proj{v} \otimes \proj{z} \otimes
    \rho_{\tilde{\sf C}}^{v,z}$} satisfies the following condition.
  When executing the opening phase on the state $\rho_{V \tilde{\sf
      C}}^{b}$, for any strategy of \smash{$\tilde{\sf C}$}, the
  honest verifier accepts an opening to $1-b$ with probability at most
  $\eps$.
\end{definition}
It is easy to see that the binding property as defined here implies
the above discussed weak version, namely $p_b \leq P_{D}(b) +
P_{D}(1-b)\eps$ and thus $p_0 + p_1 \leq 1 + \eps$.  Furthermore, it
is straightforward to see that this stronger notion allows for a
formal proof of the obvious reduction of a string to a bit commitment
by committing bit-wise: the $i$-th execution of the bit commitment
scheme guarantees a random variable $D_i$, defined by $P_{D_i|V_i
  Z}$, such that the committer cannot open the $i$-th bit commitment
to $1-D_i$, and thus there exists a random variable $S$, namely $S
= (D_1,\ldots,D_m)$ defined by $P_{D_1\cdots D_m|V_1\cdots V_m Z}
= \prod_i P_{D_i|V_i Z}$, such that for any opening strategy, the
committer cannot open the list of commitments to any other string than
$S$.

In Section~\ref{sec:strongbinding}, we show that the bit commitment
\comm\ from Figure~\ref{fig:comm} as a matter of fact satisfies this
stronger and more useful notion of security. This turns out to be a
rather straightforward consequence of the security of the \OT\ scheme
from Chapter~\ref{chap:12OT}.

\section{Weak Binding of the Commitment Scheme} \label{sec:weakbinding}
\index{bit commitment!weak binding|(} In this section, we use the
techniques from the analysis of the \RabinOT protocol from
Chapter~\ref{chap:RabinOT} to prove our commitment scheme \comm\ (or
rather its purified version \eprcomm) weakly binding against
quantum-memory-bounded adversarial committers.

Note that the first two steps of \eprqot\ (from
Figure~\ref{fig:eprot}) and \eprcomm\ (i.e.~before the memory bound applies) are exactly the
same!  This allows us to reuse Corollary~\ref{cor:hadamard} and the
analysis of Section~\ref{sec:otsecurity} to prove the weakly binding
property of \eprcomm.
\begin{theorem}\label{thm:weakbinding}
For any $\gamma<\frac{1}{2}$,
\comm\ is perfectly hiding and weakly binding against $\mathfrak{C}_{\gamma}$.  
\end{theorem}
The proof is given below. It boils down to showing that essentially
$p_0(n) \leq 1 - \qp$ and $p_1(n) \leq 1 - \qt$. The weak binding property
then follows immediately from Corollary~\ref{cor:hadamard}. The
intuition behind $p_0(n) \leq 1 - \qp = 1 - \Qp(S^+)$ is that a committer has only a
fair chance in opening to $0$ if $x$ measured in the $+$-basis has large
probability, i.e., $x \not\in S^+$. The following proof makes this
intuition precise by choosing the $\varepsilon$ and $\delta$'s
correctly.
\begin{proof}
  It remains to show that \eprcomm\ is binding against
  $\mathfrak{C}_{\gamma}$. Let $\eps, \delta > 0$ be such that $\gamma
  + 2h(\delta) + 2\eps < 1/2$, where $h$ is the binary entropy
  function. Recall that the number $\ball{\delta n}$ of $n$-bit strings of Hamming-distance at
  most $\delta n$ from a fixed string is at most $2^{h(\delta) n}$. Let
  $R$ be the basis, determined by the bit that $\dC$ claims in
  Step~\ref{step:open}, and in which $\V$ measures the quantum state
  in Step~\ref{step:last}, and let $X$ be the outcome.
  Corollary~\ref{cor:hadamard} implies the existence of an event $\cal
  E$ such that $\P[\ev|R\!=\!+] + \P[\ev|R\!=\!\times] \geq 1 -
  \negl{n}$ and $\H_{\infty}(X|R\!=\!r,\ev) \geq (\gamma+2h(\delta) +
  2\eps) n$. Applying Corollary~\ref{cor:guess} (with constant $U$ and
  $\eps=0$), it follows that any guess $\hat{X}$ for $X$ satisfies
\begin{align*}
\P\big[ \hat{X} \in \ball{\delta n}(X) \,|\, R\!=\!r,\ev \big]
 &\leq 2^{-\frac{1}{2} (\H_{\infty}(X|X \in S^+)-\gamma n-1) + \log(\ball{\delta
 n}) } \leq 2^{-\eps n + \frac{1}{2}}.
\end{align*}
However, if $\hat{X} \not\in \ball{\delta n}(X)$ then sampling a
random subset of the positions will detect an error except with
probability at most $2^{-\delta n}$. Hence, writing $\qp \assign
\P[\ev|R\!=\!+]$ and $\qt \assign \P[\ev|R\!=\!\times]$,
$$
p_0(n) \leq (1-\qp) + \qp\cdot (2^{-\eps n + \frac12} + 2^{-\delta
  n}) \leq 1-\qp + \negl{n}
$$
and analogously $p_1(n) \leq 1-\qt + \negl{n} $. We conclude that
$$
p_0(n) + p_1(n) \leq  2 - \qp - \qt + \negl{n} \leq  1+ \negl{n} \, .
$$
\end{proof}
\index{bit commitment!weak binding|)}

\section{Strong Binding of the Commitment Scheme} \label{sec:strongbinding}
\index{bit commitment!strong binding|(} In this section, we reuse the
analysis of the \OT-protocol from Chapter~\ref{chap:12OT} to prove the
strong binding condition.

\begin{theorem}
  The quantum bit-commitment scheme \comm\ is $\eps$-binding
  according to Definition~\ref{def:strongbinding} against $\mathfrak{C}_{\gamma}$
  for a negligible (in $n$) $\eps$ if $\gamma < \frac14$.
\end{theorem}
Intuitively, one can argue that $X$ has (smooth) min-entropy about $n/2$ given
$\Theta$.  The \index{min-entropy splitting lemma}
Min-Entropy Splitting Lemma implies that there exists
$D$ such that $X_{1-D}$ has smooth min-entropy about $n/4$ given
$\Theta$ and $D$. Privacy amplification implies that $F(X_{1-D})$ is
close to random given $\Theta, D, F$ and $\dC$'s quantum register of
size $\gamma n$, where $F$ is a \univ one-bit-output hash function,
which in particular implies that $\dC$ cannot guess $X_{1-D}$. The
formal proof is given below.
\begin{proof}
  It remains to show that \eprcomm\ is strongly binding against
  $\mathfrak{C}_{\gamma}$.  Let $\Theta \in \set{+,\times}^n$ be the
  random basis that would correspond to the choice of basis in the
  first step of \comm, i.e. $\theta_i = \set{+,\times}_{[b]}$ for $i
  \in I$ and $\theta_i = \set{+,\times}_{[1-b]}$ for $i \not\in I$. Let
  $X$ be the measurement outcome when $\V$ measures his halves of the
  EPR-pairs in basis $\Theta$.
  
Recall that $h(\cdot)$ denotes the binary Shannon entropy. Choose $\lambda, \lambda', \kappa$ and $\delta$ all positive, but
  small enough such that $\gamma \leq 1/4 - \lambda - \lambda' -
  2h(\delta) - 2 \kappa$, $h(\delta) \leq \lambda'-\kappa$, and
  $h(\delta) \leq \frac{\lambda^4}{32} -
  \kappa$. Before Step~\ref{step:open}, the overall state is given by
  the ccq-state $\rho_{X \Theta \dC}$ after $\dC$ has measured all but
  $\gamma n$ of his qubits, where $X$ describes the outcome of the
  verifier $\V$ measuring his part of the state in random basis
  $\Theta$.  From the uncertainty relation
  (Corollary~\ref{cor:uncertainty}), we know that $\hie{\eps}{X
    \mid \Theta} \geq (1/2 - 2\lambda )n$ for
  $\eps=2^{-\frac{\lambda^4}{32}n}$ 
exponentially small in $n$.  Therefore, by
  Corollary~\ref{cor:ESL}, there exists a binary random variable $D
  \in \set{0,1}$ such that for $\eps'=2^{-\lambda' n}$, it holds that
\begin{align*}
\hie{\eps+\eps'}{X_{1-D} \mid \Theta D} &\geq (1/4 - \lambda - \lambda')n -1\\ 
&\geq (1/4 -\lambda - \lambda')n -1\\
&\geq \gamma n + 2h(\delta)n + 2 \kappa n -1 \, .
\end{align*}

Recall that $\ball{\delta n} \leq 2^{h(\delta) n}$. Applying
Corollary~\ref{cor:guess}, it follows that any guess $\hat{X}$ for $X_{1-D}$
satisfies
\begin{align*}
  \P\big[ \hat{X} \in \ball{\delta n}(X_{1-D}) \big] &\leq 2^{-\frac{1}{2}
    (\hie{\eps+\eps'}{X_{1-D}|\Theta D}-\gamma n-1) + \log(\ball{\delta
      n}) } + (2 \eps + 2 \eps') \ball{\delta n}\\
  &\leq 2^{-\frac{1}{2} (2 \kappa n - 2)} + 2 \cdot 2^{
    -\frac{\lambda^4}{32}n + h(\delta) n } + 2 \cdot 2^{ -
    \lambda' n + h(\delta) n} \\
  &\leq \frac12 2^{-\kappa n } + 2 \cdot 2^{ -\kappa n } + 2 \cdot 2^{
    -\kappa n} \, ,
\end{align*}
which is negligible by the choice of the parameters.
\end{proof}

\index{bit commitment!strong binding|)}
\index{bit commitment!binding|)}

\section{Weakening the Assumptions} \label{sec:weakassumptioncomm}
\index{weak quantum model|(}As argued earlier, assuming that a party can
produce single qubits (with probability~1) is not reasonable given
current technology. Also the assumption that there is no noise on the
quantum channel is impractical.  It can be shown that a
straightforward modification of \comm\ remains secure in the
$(\phi,\eta)$-weak quantum model as introduced in
Section~\ref{sec:weakass} (see also Section~\ref{sec:moreimperfect}),
with $\phi < \frac{1}{2}$ and $\eta < 1- \phi$.

\begin{myfigure}{h}
\begin{myprotocol}{\comm'$(b,\phi)$}
\item $\V$ picks $x \in_R \nbit$ and $\theta \in_R \{+,\times \}^n$ and
  sends $x_i$ in the corresponding bases $\ket{x_1}_{\theta_1},
  \ket{x_2}_{\theta_2}, \ldots, \ket{x_n}_{\theta_n}$ to~$\C$.
 \item $\C$ commits to the bit $b$ by measuring all qubits in basis
   $\{+,\times \}_{[b]}$. Let $x' \in \nbit$ be the result.
 \item To open the commitment, $\C$ sends $b$ and $x'$ to $\V$.
 \item $\V$ verifies that $x_i= x_i'$ for $i$ where $\theta_i = \{+,\times
   \}_{[b]}$. $\V$ accepts if and only if this is the case \emph{for all but
   a $\phi$-fraction of these positions}.
\end{myprotocol}
\caption{Protocol for noise-tolerant quantum bit commitment}\label{fig:commnoisy}
\end{myfigure}

The protocol \comm'\ in Figure~\ref{fig:commnoisy} is the same as
\comm\ from Figure~\ref{fig:comm} except that in the last
Step~\ref{step:last}, $\V$ accepts if and only if $x_i = x'_i$ for all
{\em but about a $\phi$-fraction} of the $i$ where $r_i = \{+,\times
\}_{[b]}$. More precisely, for all but a
$(\phi+\varepsilon)$-fraction, where $\varepsilon > 0$ is sufficiently
small.
\begin{theorem}\label{thm:weakcommsec}
  In the $(\phi,\eta)$-weak quantum model, \comm'\ is perfectly hiding
  and it is weakly binding against $\mathfrak{C}_{\gamma}$ for any $\gamma$
  satisfying $\gamma < \frac{1}{2}(1-\eta) - 2 h(\phi)$.
\end{theorem}
\begin{sketch} 
  Using \index{Chernoff's inequality} Chernoff's inequality (Lemma~\ref{lem:chernoff}), one can
  argue that for {\em honest} $\C$ and $\V$, the opening of a
  commitment is accepted except with negligible probability.  The
  hiding property holds using the same reasoning as in
  Lemma~\ref{lem:sec:hiding}. And the binding property can be argued
  essentially along the lines of Theorem~\ref{thm:weakbinding}, with
  the following modifications. Let $J$ denote the set of indices $i$
  where $\V$ succeeds in sending a single qubit. We restrict the
  analysis to those $i$'s which are in $J$. By 
  \index{Chernoff's inequality} Chernoff's inequality
  (Lemma~\ref{lem:chernoff}), the cardinality of $J$ is about
  $(1-\eta)n$ (meaning within $(1-\eta\pm \varepsilon)n$), except with
  negligible probability.  Thus, restricting to these $i$'s has the
  same effect as replacing $\gamma$ by $\gamma/(1-\eta)$ (neglecting
  the $\pm \varepsilon$ to simplify notation). Assuming that $\dC$
  knows every $x_i$ for $i \not\in J$, for all $x_i$'s with $i \in J$,
  he has to be able to guess all but about a $\phi/(1-\eta)$-fraction
  correctly, in order to be successful in the opening. Using
  Corollary~\ref{cor:guess}, we can show that for a correctly chosen
  $\delta >0$, the probability of guessing $\hat{X}$ within Hamming
  distance $\delta n$ to the real $X$ is negligible. Therefore, $\dC$
  succeeds with only negligible probability if the fraction of allowed
  errors $\phi/(1-\eta)$ is smaller than $\delta$, i.e.
$$
\phi/(1-\eta) < \delta \, ,$$ 
Additionally, in order for the machinery from Theorem~\ref{thm:weakbinding} to work, $\delta$ must be such that 
$$ \frac{\gamma}{1-\eta} + 2h(\delta) < \frac{1}{2} \, .
$$ 
$\delta$ can be chosen that way if 
$$
\frac{\gamma}{1-\eta} + 2 \, h\!\left( \frac{\phi}{1-\eta} \right) < \frac{1}{2} \, .
$$
Using the fact that $h(\nu p) \leq \nu h(p)$ for any $\nu \geq 1$
and $0 \leq p \leq \frac12$ such that $\nu p \leq 1$, this is clearly
satisfied if
$\gamma + 2 h(\phi) < \frac{1}{2}(1-\eta)$. 
\end{sketch}

\begin{theorem}\label{thm:weakcommsecstrong}
  In the $(\phi,\eta)$-weak quantum model, \comm'\ is perfectly hiding
  and it is strongly binding against $\mathfrak{C}_{\gamma}$ for any $\gamma$
  satisfying $\gamma < \frac{1}{4}(1-\eta) - 3 h(\phi) - \sqrt[4]{32
  \, h(\phi)}$.
\end{theorem}
\begin{sketch}
  The proof goes like the proof of Theorem~\ref{thm:weakcommsec}, but
  uses the techniques from Section~\ref{sec:strongbinding}. In order
  for those to work, we need to choose $\lambda, \lambda'$,
  and $\delta$ all positive and such that
\begin{align} \label{eq:conditions} \begin{split}
\frac{\phi}{1 - \eta} &< \delta,  \\
\frac{\gamma}{1 - \eta} + 2 h(\delta) + \lambda' + \lambda &< 1/4 \,
, \\
h(\delta) &< \lambda' \, , \\
h(\delta) &< \frac{\lambda^4}{32} \, .
\end{split} \end{align}
We verify that the assumption $\gamma < \frac{1}{4}(1-\eta) - 3 h(\phi) - \sqrt[4]{32
  \, h(\phi)}$ on $\gamma$ allows for that. Rearranging
the terms and using that $x < \sqrt[4]{x}$ for $0<x<1$ yields
\[
 \frac{\gamma}{1 - \eta} + 3 \frac{h(\phi)}{1 - \eta} +
 \sqrt[4]{32\frac{h(\phi)}{1- \eta}} < 1/4 \, .
\]
Using as in the previous proof the fact that $h(\nu p) \leq \nu h(p)$ for any
$\nu \geq 1$ and $0 \leq p \leq \frac12$ such that $\nu p \leq 1$, we
get that
\[
 \frac{\gamma}{1 - \eta} + 3 h\left(\frac{\phi}{1-\eta}\right) +
 \sqrt[4]{32 \, h\left(\frac{\phi}{1-\eta}\right)} < 1/4.
\]
That allows to choose $\delta > \frac{\phi}{1-\eta}$ such that
\[
 \frac{\gamma}{1 - \eta} + 2 h(\delta) +  h(\delta) + 
 \sqrt[4]{32 \, h(\delta)} < 1/4,
\]
and therefore, also $\lambda$ and $\lambda'$ can be chosen such that
the conditions \eqref{eq:conditions} are fulfilled.
\end{sketch}

\index{weak quantum model|)}
\index{bit commitment|)}


%% file: rest.tex
\chapter[\QKD Against Bounded Eavesdroppers]
{\QKD Secure Against Quantum-Memory-Bounded Eavesdroppers} \label{chap:qkd}
\index{quantum key distribution|(} In this chapter, we present another
application for the uncertainty relation derived in
Section~\ref{sec:morerelation}. This illustrates that these relations
are useful in scenarios beyond the simple two-party setting.

In Quantum Key Distribution (\QKD), two honest players Alice and Bob
want to agree on a secure key, using only completely insecure quantum
and authentic classical communication. The computationally unbounded
eavesdropper Eve should not get any information about the key. A major
difficulty when implementing \QKD schemes is that they require a
low-noise quantum channel.  The tolerated noise level depends on the
actual protocol and on the desired security of the key.  Because the
quality of the channel typically decreases with its length, the
maximum tolerated \index{noise level}noise level
 is an important parameter limiting the
\index{maximum distance}maximum distance between Alice and Bob.

We consider a model in which the adversary has a limited amount of
quantum memory to store the information she intercepts during the
protocol execution. In this model, we show that the maximum
tolerated noise level is larger than in the standard scenario where
the adversary has unlimited resources.  

For simplicity, we restrict ourselves to {\em one-way \QKD protocols}
which are protocols where error-correction is performed
non-interactively, i.e., a single classical message is sent from one
party to the other.

The results in this chapter appeared in~\cite{DFRSS07}.


\section{Derivation of the Maximum Tolerated Noise Level}
Let $\bset$ be a set of orthonormal bases of a $d$-dimensional Hilbert
space $\cH_d$. For each basis $\vartheta \in \bset$, we assume that
the $d$ basis vectors are parametrized by the elements of the fixed
set $\cX$ of size $|\cX|=d$. We then consider \QKD protocols consisting
of the steps described in Figure~\ref{fig:QKDShape}.
\begin{myfigure}{h}
\begin{myprotocol}[let $N \in \naturals$ be arbitrary]{One-Way \QKD}
\item \emph{Preparation:} For $i=1 \ldots N$, Alice chooses at random
  a basis $\vartheta_i \in \bset$ and a random element $X_i \in \cX$.
  She encodes $X_i$ into the state of a quantum system according to
  the basis $\vartheta_i$ and sends this system to Bob.  Bob measures
  each of the states he receives according to a randomly chosen basis
  $\vartheta'_i$ and stores the outcome $Y_i \in \cX$ of this
  measurement. \index{preparation}
\item \emph{Sifting:} Alice and Bob publicly announce their choices of
  bases and keep their data at position $i$ only if $\vartheta_i =
  \vartheta'_i$. In the following, we denote by $X$ and $Y$ the
  concatenation of the remaining data $X_i$ and $Y_i$, respectively.
  $X$ and $Y$ are sometimes called the \emph{sifted raw key}. \index{sifting}
\item \emph{Error correction:} Alice computes some error correction
  information $C$ depending on $X$ and sends $C$ to Bob.  Bob computes
  a guess $\hat{X}$ for Alice's string $X$, using $C$ and
  $Y$. \index{error correction}
\item \emph{Privacy amplification:} Alice chooses at random a function
  $f$ from a two-universal family of hash functions and announces $f$
  to Bob. Alice and Bob then compute the final key by applying $f$ to
  their strings $X$ and $\hat{X}$, respectively. \index{privacy amplification}
\end{myprotocol}
\caption{General form for {\em one-way} \QKD protocols.}\label{fig:QKDShape}
\end{myfigure} \index{quantum key distribution!one-way}

Note that the quantum channel is only used in the preparation step.
Afterwards, the communication between Alice and Bob
is only classical (over an authentic channel).

As shown in~\cite[Lemma~6.4.1]{Renner05}, the length $\ell$ of the
secret key that can be generated by the protocol described above is
given by\footnote{The approximation in this and the following
  equations holds up to some small additive value which depends
  logarithmically on the desired security $\eps$ of the final key.}
\[
\ell \approx \hminee(\rho_{X \regE} \mid \regE) - \hmax(C)  \, ,
\]
where the cq-state $\rho_{X \regE}$ is the state of the quantum system
with the property that $\regE$ contains all the information Eve has
gained during the preparation step of the protocol and where
$\hmax(C)$ is the number of error correction bits sent from Alice to
Bob. Note that this formula can be seen as a generalization of the
well-known expression by Csisz\'{a}r and K\"{o}rner for classical key
agreement~\cite{CK78}.

Let us now assume that Eve's system $\regE$ can be decomposed into a
classical part $U$ and a purely quantum part $\regE'$. Then, by the
same derivation as in the proof of Corollary~\ref{thm:pasmooth}, we find
\[
  \ell 
\approx 
  \hminee(\rho_{X U \regE'} \mid U \regE') - \hmax(C) 
\geq
  \hiee{X \mid U} - \qhmax(\rho_{\regE'}) - \hmax(C) \ .
\]
As, during the preparation step, Eve does not know the encoding
bases which are chosen at random from the set $\bset$, we can apply our
uncertainty relation (Theorem~\ref{thm:genrel}) to get a lower bound for
the min-entropy of $X$ conditioned on Eve's classical information $\Theta$,
i.e.,
\[
\hiee{X \mid \Theta} \geq M h,
\]
where $M$ denotes the length of the sifted \index{raw key}raw key $X$ and $h$ is the
\index{average entropic uncertainty bound}average entropic uncertainty
bound for $\bset$. \comment{write much more!} Let $q$ be the bound
on the size of Eve's quantum memory $\qhmax(\rho_{\regE'}) \leq q$.
Moreover, let $e$ be the average amount of error correction
information that Alice has to send to Bob per symbol of the sifted raw
key $X$. Then
\[
  \ell 
\gtrapprox 
  M (h-e) - q \ .
\]
Hence, if the memory bound only grows sublinearly in the length $M$ of
the sifted raw key, then the \emph{key rate}, i.e., the number of key
bits generated per bit of the sifted raw key, is lower bounded by
\[
  \mathrm{rate}
\geq
  h-e \ .
\]

\section{The Binary-Channel Setting}
For a binary channel (with a two-dimensional Hilbert space $\cH_2$), the average
amount of error correction information $e$ is given by the binary Shannon
entropy\footnote{This value of $e$ is only achieved if an optimal
  error-correction scheme is used. In practical implementations, the
  value of $e$ might be slightly larger.} $h(p)$, where $p$ is
the bit-flip probability (for classical bits encoded according to some
orthonormal basis as described above). The achievable key rate of a
\QKD protocol using a binary quantum channel is thus given by
\[
  \mathrm{rate}_{\mathrm{binary}} \geq h - h(p) \ .
\]
Summing up, we have derived the following theorem.

\begin{theorem}
  Let $\bset$ be a set of orthonormal bases of $\cH_2$ with average
  entropic uncertainty bound $h$. Then, a \emph{one-way} \QKD protocol
  as in Figure~\ref{fig:QKDShape} produces a secure key against
  eavesdroppers whose quantum-memory size is sublinear in the length
  of the raw key (i.e., sublinear in the number of qubits sent from
  Alice to Bob) at a positive rate as long as the bit-flip probability
  $p$ fulfills
\begin{equation} \label{eq:noiselevel}
  h(p) < h \ .
\end{equation}
\end{theorem}

For the BB84 protocol~\cite{BB84}, we have $h = \frac{1}{2}$ (cf.
Inequality~\eqref{eq:maassenuffink}).  Inequality~\eqref{eq:noiselevel} is
thus satisfied as long as $p \leq 11\%$.  This bound coincides with
the known bound for one-way \QKD in the standard model (with an
unbounded eavesdropper). So, using our analysis here, the memory-bound
does not give an advantage.

The situation is different for the six-state protocol where $h =
\frac{2}{3}$. According to~\eqref{eq:noiselevel}, security against
memory-bounded adversaries is guaranteed (i.e. $h(p) < \frac{2}{3}$)
as long as $p \leq 17\%$. If one requires security against an
unbounded adversary, the threshold for the same protocol lies below
$13\%$ as shown by Lo~\cite{Lo01}, and even the best known QKD
protocol on binary channels with one-way classical post-processing can
only tolerate noise up to roughly $14.1\%$~\cite{RGK05}. It has also
been shown that, in the unbounded model, no such protocol can tolerate
an error rate of more than~$16.3\%$.

The performance of \QKD protocols against quantum-memory bounded
eavesdroppers can be improved further by making the choice of the
encoding bases more random. For example, they might be chosen from the
set of all possible orthonormal bases on a two-dimensional Hilbert
space.  As shown in Section~\ref{sec:uncertbound}, the overall average
entropic uncertainty bound is then given by $h \approx 0.72$
and~\eqref{eq:noiselevel} is satisfied if $p \lessapprox 20\%$.  For an
unbounded adversary, the thresholds are the same as for the six-state
protocol (i.e., $14.1\%$ for the best known one-way protocol).

\section{Possible Extensions}
It is an interesting open problem to consider protocols using
higher-dimensional quantum systems. The results described in
Section~\ref{sec:uncertbound} show that for high-dimensional systems,
the average entropic \index{average entropic uncertainty bound}
uncertainty bound converges to its theoretical maximum. The maximal
tolerated channel noise might thus be higher for such protocols
(depending on the noise model for higher-dimensional quantum
channels).

\index{quantum key distribution!one-way}Another interesting problem 
is to derive completely one-way
quantum-key-distribution schemes, i.e.~to eliminate the interactive
\index{sifting}sifting phase from the protocol in Figure~\ref{fig:QKDShape}. The idea
is to let the honest parties use a pre-shared secret key to determine
the bases of the encoding. If a key of size linear in the number of
qubits is used, the scheme has to guarantee that a big portion of the
key can be reused several times in order to yield a reasonable amount
of fresh key. Quantifying the amount of information an eavesdropper can
learn about the pre-shared key by interfering in the preparation step
and eavesdropping on the following classical communication is an open
problem. 

Another approach consists of expanding a pre-shared key of
size only logarithmic in the number of qubits into a pseudo-random linear-size key
to determine the bases of the encoding. It is an open question how to
extend our uncertainty relation from Section~\ref{sec:morerelation}
to the case of only pseudo-random bases.



\index{quantum key distribution|)}

\chapter{Conclusion} \label{chap:conclusions}


\section{Towards Practice} \label{sec:techpractice}
In the following two sections, we elaborate on the question how close
to practice our systems are. First, we argue that imperfections
occurring in practice like \emph{dark counts} and \emph{empty pulses}
are covered by our $(\phi,\eta)$-weak quantum model used in
Sections~\ref{sec:weakass}, \ref{sec:weakmodel12ot}, and~\ref{sec:weakassumptioncomm}. Second,
we sketch how our techniques can be extended to the more realistic
setting of \emph{noisy quantum memory}.

\subsection{More Imperfections} \label{sec:moreimperfect}
A natural approach for implementing two-party protocols like \BBqot,
\Randlqot,\ and \comm\ is to use the polarization of photons governed
by the laws of quantum optics. Such systems are nowadays at the stage
where they can be built in a optical physics lab. \index{laser} Besides the
already modeled bit errors and multi-pulse emissions, more
imperfections of the physical apparatus such as \emph{empty pulses}
and \emph{dark counts} need to be taken into account.
\index{empty pulse} \index{dark count} \index{weak quantum model}

The players have synchronized clocks and in every predefined time
\index{time slot} slot, the sender is supposed to send out a single qubit. In practice,
\index{weak coherent pulse} \index{single-photon source}
weak coherent pulses are used to approximate single-photon sources by
producing in average only a small fraction of one qubit per pulse.
This means that most of the pulses are \emph{empty}, but on the other
hand, there is also a small probability for a multi-qubit pulse. The
\index{multi-qubit emission}receiver reports to the sender in which 
time slots he received pulses.

Empty pulses also occur when the quantum channel lets a transmitted
qubit escape or when it is absorbed. It is realistic that a good
estimate on the rate at which empty pulses are produced (when no
adversary is present) is known, e.g., from the hardware specifications
and by measuring and calibrating the experimental setup. In this case,
the adversary can only take advantage of empty pulses caused by
absorption in the fiber. The best the adversary can do is to
substitute the fiber for one that preserves all qubits sent and to
report empty pulses when a single pulse has been received. The effect
is to increase the rate at which multi-qubit pulses occur. This attack
\index{photon-number-splitting attack} 
is known as \emph{Photon-Number-Splitting attack} as first noted by
Huttner, Imoto, Gisin, and Mor~\cite{HIGM95} and for instance
explained in \cite{BLMS00a,BLMS00b} in the setting of quantum key
distribution. It follows that empty pulses can also be included in the
$(\phi,\eta)$-weak quantum model by an appropriate adjustment of
parameter $\eta$.

Furthermore, thermal fluctuation in the detector hardware might result
in detection even though no qubit was received. This is called a
\emph{dark count}. In this time slot, the receiver will report the
\index{dark count}reception of a qubit and as the outcome is random, it agrees with the
actual bit sent with probability $\frac12$.

Formally, assume that a practical implementation of \BBqot,
\Randlqot,\ or \comm\ takes place in a setting where $\phix$ is the
probability for a bit error caused by the channel, $\phidc$ is the
probability for a dark count in a specific time slot, $\etamq$ is the
probability for a multi-qubit transmission in a non-empty pulse, and
$\etaab$ is the probability for an empty
pulse caused by absorption of a non-empty pulse. 
In these terms, dark counts contribute $\frac{\phidc}{2}$ to the
bit-error rate $\phix$. If the adversary is able to get perfect
transmission, she can suppress single-qubit pulses up to a rate of
$\etaab$, thereby increasing the rate $\etamq$ of multi-photon pulses
by $\frac{1}{1-\etaab}$.  It follows that if \BBqot, \comm,\ and
\Randlqot\ 
are secure in the $(\phix+\frac{\phidc}{2},
\frac{\etamq}{1-\etaab})$-weak quantum model, then their
implementation is also secure, provided it is accurately modeled by
these four parameters.  \index{weak quantum model}
\index{empty pulse} \index{dark count}

Likewise, a variety of imperfections specific to particular
implementations may be adapted to the weak quantum model.
 
\subsection{Generalizing the Memory Model} \label{sec:noisymem}
The \index{bounded-quantum-storage model}bounded-quantum-storage model
 limits the number of physical qubits
the adversary's memory can contain. A more realistic model would
rather address the noise process the adversary's memory undergoes.
For instance, it is not hard to build a very large, but unreliable
memory device containing a large number of qubits. It is reasonable to
expect that our protocols remain secure also in a scenario where the
adversary's memory is of arbitrary size, but where some quantum
operation (modeling noise) applies to it. If we do not substitute
$\qhmax(\rho_{\regE})$ with the number of qubits $q$ in
Term~(\ref{eq:lemma}) in the privacy-amplification
Section~\ref{sec:pa}, then our constructions can cope with slightly
more general memory models. In particular, all our protocols that are
secure against adversaries with memory of no more than $\gamma n$
qubits are also secure against any noise model that reduces the rank
$\qhmax(\rho_{\regE})$ of the mixed state $\rho_{\regE}$ held by the
adversary to at most $2^{\gamma n}$. \index{noisy-memory model}

An example of a noise process resulting in a reduction of $\qhmax(\rho_{\regE})$
is an \index{erasure channel}erasure channel. Assuming the $n$ initial qubits are each erased
with probability larger than $1-\gamma$ when the memory bound applies,
it holds except with negligible probability in $n$ that
$\qhmax(\rho_{\regE})<\gamma n$.  The same applies if the noise process is
modeled by a \index{depolarizing channel}depolarizing channel with error probability
$p=1-\gamma$. Such a depolarizing channel replaces each qubit by a
random one with probability $p$ and does nothing with probability
$1-p$.

The technique we have developed does not allow to deal with depolarizing channels
with $p < 1-\gamma$ although one would expect that some $0< p < 1-\gamma$
should be sufficient to ensure security against such adversaries.
The reason being that not knowing the positions where the errors occurred
should make it more difficult for the adversary than when the noise process
is modeled by an erasure channel. However, it seems that our uncertainty 
relations 
are not strong enough to address this case. Generalizing the
bounded-quantum-storage model to more realistic \index{noisy-memory model}noisy-memory models is
an interesting open question.

\section{Conclusion} \label{sec:conclusion}
The \index{bounded-quantum-storage model}bounded-quantum-storage model
presented in this thesis is an attractive model, in both the
theoretical and practical sense. On the theoretical side, it allows
for very simple protocols implementing basic two-party primitives such
as oblivious transfer and bit commitment. New high-order entropic
\index{uncertainty relation}uncertainty relations
 have been established in order to show the
security with the help of techniques such as purification and privacy
amplification by two-universal hashing. These uncertainty
relations can also be applied in different settings like quantum key
distribution.

On the practical side, the protocols do not require any quantum memory
for honest players and remain secure provided the adversary has a
quantum memory of size bounded by a constant fraction of all
transmitted qubits.  Such a gap between the amount of storage required
for honest players and adversaries is not achievable by classical
means.  The protocols can be adapted to tolerate various kinds of
errors and in fact, they can be implemented with today's technology. A
collaboration of people from the computer science and physics
departments of the University of Aarhus is currently working on the
implementation of these protocols\footnote{See
\texttt{http://www.brics.dk/{\textasciitilde}salvail/qusep.html}
for further information on the QUSEP project.}.

In summary, one can say that the bounded-quantum-storage model has
passed its first tests by proving its power (the possibility of
oblivious transfer) and by inspiring beautiful theoretical results
(quantum uncertainty relations). It is a good sign that the protocols
for the basic primitives are simple in structure. In principle, enough
instances of these protocols could be used to implement more involved
cryptographic tasks like secure identification, which reduces
essentially to securely checking whether two inputs are equal (without
revealing more than this mere bit of information). However, it is a
natural next step to find more efficient, direct protocols for those
tasks, secure in the bounded-quantum-storage model. Such a direct
approach gives a better ratio between storage-bound and
communication-complexity and is the topic of a recent paper
\cite{DFSS07}.

A major open problem is the optimality of the bounds on the
adversary's quantum memory. The bit-commitment protocol \comm\ for
instance appears to be secure against any adversary with memory less
than $n$ qubits, but our analysis requires the memory to be smaller
than $n/2$ (or $n/4$ for strong binding).  Also, finding protocols
secure against adversaries in more general \index{noisy-memory model}
noisy-memory models, as discussed in the last
Section~\ref{sec:noisymem}, would certainly be a natural and
interesting extension of this work to more practical settings~
\cite{DSTW07privcom}. Furthermore, there is still a lack of simple and
intuitive security definitions for primitives like \OT etc.\ with
rigorous composability results (like universal composability) in the
quantum setting. Very recent results in this direction have been
established in~\cite{WW07privcom}.
